\documentclass{article}
\usepackage{booktabs}
\usepackage{pdflscape}
\usepackage[english]{babel}
\usepackage[letterpaper,top=2cm,bottom=2cm,left=3cm,right=3cm,marginparwidth=1.75cm]{geometry}
\usepackage{booktabs}
\usepackage{amsmath}
\usepackage{amsthm}
 \usepackage{amssymb}
\usepackage{graphicx}
\usepackage[colorlinks=true, allcolors=blue]{hyperref}
\usepackage{multirow}
\usepackage{tikz}
\usepackage{subcaption}
\usetikzlibrary{positioning,arrows.meta}
\usepackage{float}
\usepackage{natbib}
\usepackage{authblk}
\usepackage{bbm}
\newtheorem{theorem}{Theorem}

\newtheorem{proposition}{Proposition}
\newtheorem{remark}{Remark}
\newtheorem{corollary}{Corollary}
\newtheorem{assumption}{Assumption}
\newtheorem{lemma}{Lemma}
\usepackage{tabularx}
\usepackage{makecell}
\usepackage{comment}
\usepackage{longtable}
\usepackage[dvipsnames]{xcolor}
\definecolor{blue}{HTML}{1F77B4}
\definecolor{orange}{HTML}{FF7F0E}
\definecolor{green}{HTML}{2CA02C}
\definecolor{red}{HTML}{D62728}
\definecolor{purple}{HTML}{9467BD}
\definecolor{brown}{HTML}{8C564B}
\definecolor{pink}{HTML}{E377C2}
\definecolor{grey}{HTML}{7F7F7F}
\definecolor{yellow}{HTML}{BCBD22}
\definecolor{cyan}{HTML}{17BECF}
\definecolor{turquoise}{HTML}{3FE0D0}
\definecolor{royalblue}{HTML}{0038A8} 
\definecolor{darkpurple}{HTML}{5B2C83}

\definecolor{royalblue}{named}{black} 

\usepackage[algo2e,ruled,vlined]{algorithm2e}
\definecolor{algoColorKeyword}{named}{blue}
\definecolor{algoColorComment}{named}{olive}

\SetKwComment{Comment}{$\triangleright$\ }{}
\DontPrintSemicolon

\usepackage{enumitem}
\setlist{leftmargin=*, topsep=0.5em, parsep=0pt, itemsep=1em, labelindent=0pt, align=left}

\usepackage{orcidlink}
\definecolor{orcidlogocol}{HTML}{2E7E9F}

\title{A deep learning approach for pricing convertible bonds with path-dependent reset and call provisions}

\author[1,2]{\orcidlinki{\textcolor{black}{Qinwen Zhu}}{0000-0001-5659-3309}}
\author[3]{\orcidlinki{\textcolor{black}{Wen Chen}}{0000-0002-7268-0979}\thanks{Corresponding author, wenchen@bnbu.edu.cn}} 
\author[3]{\orcidlinki{\textcolor{black}{Nicolas Langren\'e}}{0000-0001-7601-4618}} 

\affil[1]{\normalsize School of Financial Technology, Shanghai Lixin University of Accounting and Finance, Shanghai, P.R. China.} 
\affil[2]{\normalsize Shanghai Science and Technology Industry Research Center, Shanghai, P.R. China.}
\affil[3]{\normalsize Guangdong Provincial/Zhuhai Key Laboratory of Interdisciplinary Research and Application for Data Science, Beijing Normal-Hong Kong Baptist University, Zhuhai, P.R. China.} 

\begin{document}
\maketitle

\begin{abstract}
 This paper develops a deep learning framework for pricing convertible bonds with path-dependent downward reset and issuer call provisions governed by rolling-window triggers. We formulate the valuation problem as a path-dependent partial differential equation (PPDE) that captures both the historical stock-price path and the evolution of the conversion price. Model-specific PPDEs are derived under GBM, CEV, and Heston dynamics. \textcolor{royalblue}{Under suitable conditions, we establish the existence and uniqueness of a piecewise viscosity solution linked by contractual transmission conditions at monitoring dates.} For computation, we construct a fixed-grid backward dynamic programming scheme and approximate its conditional expectations using neural networks, with $L^2$ convergence to the exact fixed-grid recursion as the approximation errors vanish. An application to the China CITIC Bank Convertible Bond produces stable prices across the three models and close agreement with the LSMC benchmark, but outperforms in dimensional scaling. \textcolor{royalblue}{The results show that contractual provisions have a greater valuation effect than the choice of underlying dynamics. The call provision reduces the bond value by truncating upside gains, whereas the downward reset provision increases it under the benchmark specification because improved conversion terms dominate the effect of earlier redemption. Delta and Gamma obtained by automatic differentiation of smooth local network approximations closely agree with central finite-difference estimates.} The framework provides a flexible approach to pricing and sensitivity analysis for convertible bonds with complex path-dependent provisions.
\end{abstract}

\section{Introduction}

Convertible bonds (CBs) are hybrid financial instruments that give bondholders the right, but not the obligation, to convert the bond into a fixed number of the issuer's common shares at a predetermined conversion price under specific terms, either prior to maturity or  forced redemption by the issuer. Until conversion or maturity, holders retain the cash flow rights of a conventional bond, including periodic coupon payments and principal repayment at maturity. As a typical financial instrument, CBs combine the downside protection of conventional bonds with the upside potential of the issuer's equity. This asymmetric payoff structure limits investors' losses when the underlying stock performs poorly, while allowing them to benefit from stock-price
appreciation through voluntary conversion. From a corporate financing perspective, CBs reduce initial equity dilution relative to direct equity issuance and lower coupon costs compared to ordinary bonds, making them a widely used financing tool globally. 

Most real-world CBs embed two critical path-dependent provisions: an issuer call provision, which allows the issuer to redeem the bonds early at a contractually specified call price once a rolling-window price trigger is met, and a downward conversion price reset, in which the conversion price is reduced if the stock underperforms over a rolling window, partially protecting investors on the downside. These clauses introduce strong nonlinearity and path-dependence. The reset feature is particularly prevalent in Chinese CBs, yet poses severe challenges to traditional pricing methods. 

The valuation of CBs has been an important topic in financial engineering due to this inherent complexity. Existing methods fall into three broad categories: closed-form solutions, partial differential equation (PDE)-based numerical methods, and simulation-based approaches. 
An early closed-form approach is Ingersoll's contingent-claim pricing model \citep{ingersoll1977contingent}, in which the author extended the Black-Scholes option pricing framework to CBs pricing with dilution effects. \citet{ingersoll1977contingent} provides a structural approach for convertible bond pricing with embedded options, which can be viewed as an option determined by the value of the firm. Subsequent studies \citep{lewis1991convertible, nyborg1996use,zhu2006closed} developed closed-form solutions under optimal conversion strategies only at maturity \citep{nyborg1996use} and any time before maturity \citep{zhu2006closed}. These closed-form expressions cannot accommodate realistic path-dependent features.

Second, PDE-based methods formulate CB pricing as a system of variational inequalities, solved numerically via lattice-based methods or finite difference methods \citep{brennan1977convertible,brennan1980analyzing,tsiveriotis1998valuing,ayache2003valuation}. \citet{brennan1977convertible}~proposed PDEs governing the value of the CB with boundary conditions and incorporated stochastic interest rate \citep{brennan1980analyzing}. \citet{tsiveriotis1998valuing}~developed a coupled PDE system for the equity component under the standard Black-Scholes model and the debt component accounting for credit risk. \citet{ayache2003valuation}~unified the aforementioned coupled model into a single PDE framework. \citet{lau2004anatomy}~considered soft call and conversion condition with delayed calls.
 \citet{lin2020numerically}~used an ADI-based method to solve the PDEs with stochastic volatility and interest rates. \citet{lim2026structured}~proposed a structured framework for resettable convertible bonds.

Following the foundational binomial tree approach in option pricing \citep{cox1979option} and the bivariate tree method for two-factor models \citep{hull1994numerical}, lattice extensions such as binomial trees \citep{ammann2003} and trinomial trees \citep{xu2009tree} have also been used for CB pricing. \citet{ma2020valuation}~proposed a hybrid willow tree method to consider soft call and put provisions. Recently, \citet{costabile2026combining} combined lattice with regression techniques for path-dependent soft call and put provisions. Other PDE approaches, such as the finite element method, have also been used for CB pricing \citep{baroneadesi2003two}. However, these methods are not well suited to jointly handle early exercise and reset features,  stochastic volatility, and the rolling‑window path dependence that characterise Chinese CBs due to high dimensionality.

Third, simulation-based methods, particularly the least‑squares Monte Carlo (LSMC) method \citep{carriere1996valuation,longstaff2001valuing}, are more flexible for realistic dynamics of state variables and CB's path dependence features under the dynamic programming framework \citep{ammann2008,kimura2006monte,yang2010note}. \citet{li2025pricing}~considered call/put and reset clauses, approximated the conditional expectation approximation using quadratic basis functions under the LSMC framework while using constant volatility and Markovian trigger conditions. 

Recently, a few articles applied machine learning methods to convertible bonds pricing. \citet{zhu2024pricing} developed a learning based Monte Carlo simulation method to replace ordinary least squares regression without considering contractual provisions. Neural-network-based data-driven models such as financial time-series synthetic data generation have also been used to simulate the returns of underlying stocks. \citet{tan2022deep} simulated stock returns using time-series generative adversarial networks (GANs) to capture their fat-tail properties for valuation. \citet{li2025pricing} combined GANs with transformers to generate stock returns. So far, the existing learning-based approaches do not focus on optimal decision policies estimation, ignore joint provisions, and remain low-dimensional.  

Despite extensive development of theoretical CB pricing models, two critical gaps still remain. Empirical evidence consistently documents significant pricing discrepancies in real markets \citep{ammann2003,zabolotnyuk2010}. Using samples of pure convertible bonds and real transaction data, recent studies show that convertible bonds are systematically underpriced, with mispricing driven by factors such as stochastic volatility, liquidity and market conditions \citep{ammann2003}. In particular, pricing errors are more pronounced during market stress periods and for bonds with higher credit risk and lower liquidity \citep{batten2018pricing, li2025convertible}.

Incorporating embedded features such as call and reset clauses is also difficult, especially when these features are path-dependent. Earlier work, such as \citet{nyborg1996use}, only considered non-callable bond for simplicity. 
The call provision is a typical feature granting the issuer the right to redeem before maturity under certain conditions, while the investor has the right to convert the CB into equities upon and before calling. Hard call provisions, representing a fixed non-callable period during which the issuer cannot redeem the bond, have been considered \citep{tsiveriotis1998valuing}. Soft call provisions further require the stock price to meet certain triggering conditions: having the stock price to exceed a predefined trigger price \citep{brennan1977convertible, ayache2003valuation}, or to remain above a threshold for a specific period of time, e.g., 15 days, or for a fraction of a window, e.g., 15 days out of 30 days, also called $m$-out-of-$n$ soft call provision under a moving window setting, also known as Parisian feature. 

Early classical frameworks, including \citet{brennan1977convertible, ayache2003valuation, lau2004anatomy}, focus only on path-independent soft call/put provisions. While common in CB markets, path-dependent $m$-out-of-$n$ Parisian-style soft-call provisions are difficult to handle with traditional numerical methods. For instance, \citet{lau2004anatomy} noted that the computational cost of their finite-difference method approach grows exponentially with the length of the moving window. \citet{liu2020} used a Monte Carlo method to approximate the call activation probability, triggering the provision once the probability exceeds 50\%. Both methods have to resort to crude approximations which yield biased bond prices. In \citet{ma2020valuation}, the soft call activation probability was estimated using Brownian bridge simulations.
The recent article~\citet{costabile2026combining} improved the probability estimation of call activation at each time step, using a proxy function constructed by the Brownian bridge estimation under the classical LSMC framework with polynomial basis functions \citep{longstaff2001valuing} of the stock price and interest rate regressors. 
Another stream of literature concentrates exclusively on conversion price reset clauses. For example, \citet{kimura2006monte} priced non-callable convertible bonds with reset downwards clauses by estimating the conversion probability. \citet{zhu2018puttable} used an integral equation approach for resettable convertible bonds. \citet{cen2024pricing}~used decomposition and PDE methods for American-style resettable CBs. \citet{costabile2026combining}~considered the $m$-out-of-$n$ soft call and put provisions. 
Few models jointly incorporate both call and put clauses and downward resets at the same time. \citet{feng2016valuing} proposed a decomposition-based model, which requires explicit classification of price paths into discrete regimes. \citet{li2025pricing}~considered call/put and reset clauses and estimated the solution under the classical LSMC framework. 

These limitations are especially acute for Chinese CBs. Compared to conventional CBs in Europe and the United States \citep{costabile2026combining}, Chinese CBs exhibit distinctive path-dependent call and reset clauses \citep{li2025pricing}, making bond values highly sensitive to the historical trajectory of the stock price and trigger events. In particular, market expectations regarding potential conversion price revisions and issuer-forced redemptions can materially affect valuation. Consequently, this strong path-dependent feature makes them less amenable to conventional low-dimensional pricing methods. Moreover, the sensitivity of model outputs to inputs, such as volatility dynamics and credit spreads, further contributes to the discrepancies between model-implied values and market prices. Therefore, advanced valuation models must carefully incorporate both path dependence and the dynamics of state variables.
\textcolor{royalblue}{Motivated by these challenges, in this paper, we develop a deep learning-based numerical framework for pricing CBs with path-dependent downward reset and call provisions. \textcolor{royalblue}{The framework combines model-specific PPDE formulations under GBM, CEV, and Heston dynamics with a piecewise viscosity characterisation across monitoring dates. For computation, we develop a fixed-grid backward dynamic programming scheme and use DNNs to approximate its conditional expectations. We establish $L^2$ convergence to the exact recursion. We evaluate the framework using the China CITIC Bank Convertible Bond \citep{citicbank2019prospectus} across all three model specifications.}}

In summary, the three key contributions of our paper are the following:
\begin{enumerate}
    \item Theoretically, we are the first to formulate CB pricing with joint downward resets and path‑dependent calls as a PPDE, capturing historical price dependence and conversion price dynamics in a unified framework. \textcolor{royalblue}{We further establish a piecewise viscosity characterisation with comparison and monotone transmission across monitoring dates.}
    \item\textcolor{royalblue}{Numerically, we develop a backward DNN algorithm for the fixed-grid dynamic programming recursion and establish $L^2$ convergence as the network approximation errors vanish. The DNN prices closely align with the LSMC benchmarks. Taking advantage of the differentiable network representation, Delta and Gamma can be directly computed via automatic differentiation over smooth local approximations, avoiding costly repeated bump-and-revalue calculations. These Greek estimates match central finite differences results.}
    \item  In the empirical analysis, using the China CITIC Bank convertible bond as an example, we show that contractual features play a more important role in valuation than the specification of the underlying dynamics. \textcolor{royalblue}{The call provision reduces the bond value, whereas introducing the downward reset provision to the callable bond increases it. Under the benchmark parameters, the benefit from improved conversion terms dominates the effect of the lower call threshold and the associated increased likelihood of early redemption.}
\end{enumerate}

The paper proceeds as follows. Section~\ref{sec:PPDE} develops the path-dependent pricing framework and derives the PPDEs under the GBM, CEV, and Heston models. Section~\ref{sec:deepapp} presents the discrete-time recursion, its piecewise viscosity characterisation, the backward DNN algorithm, and the fixed-grid convergence analysis. Section~\ref{sec:numerical} reports the numerical results, including the LSMC benchmark, contract-feature analysis, parameter sensitivities, and Greeks. Section~\ref{sec:conclusion} concludes. Appendices~\ref{sec:a1} and~\ref{appen:a2} provide the theoretical proofs. Additional appendices present supplementary analysis of dividends and anti-dilution adjustments, Greeks across evaluation times, regularisation, alternative network architectures, and notation.

\section{A path-dependent pricing framework}\label{sec:PPDE}

This section formulates the pricing of convertible bonds with downward reset and call provisions as a path-dependent problem, leading to a path-dependent partial differential equation (PPDE) representation, and derives the corresponding equations under GBM, CEV, and Heston dynamics.

\subsection{PPDE representation}

A convertible bond can be viewed as a hybrid security that combines a straight bond component with embedded equity conversion option. At any time $t$, the payoff of the contract depends on the current stock price $S_t$ and the prevailing conversion price $H_t$. If held to maturity $T$, the bondholder receives the maximum between the face value of the bond $B$ and the conversion value $(B/H_T) S_T$, resulting in the terminal payoff
\[
\max\!\left(B,\frac{B}{H_T}S_T\right).
\]
In practice, most convertible bonds incorporate two critical path-dependent provisions: an issuer call and a downward reset clauses. The call provision grants the issuer the right to redeem the bond once certain path-dependent trigger conditions are met, which effectively caps the upside potential for the bondholders. The redemption call price $K_t$ is normally the par value $B$ plus accrued interest.
On the other hand, the downward reset provision allows a reduction in the conversion price $H_t$ following a prolonged poor stock price performance, which improves the conversion value and offer investors partial downside protection.

A distinctive feature of these contractual provisions is that the triggering conditions depend on how the underlying stock price evolves over time, rather than only its current level. In practice, the triggers are typically determined by rolling-windows rules: for instance, whether the stock price has stayed above or below a specific proportion of the conversion price for a consecutive number of trading days. As a result, the value of the contract is determined not only by the current state variables $(S_t, H_t)$, but also by the historical path of the stock price.
As a representative example, the China CITIC Bank convertible bond (code: 113021) is one of the largest bank-issued convertible bonds in the Chinese market, with an issuance size of approximately RMB 40 billion. It is a standard large-sized convertible security with a fixed conversion price and six-year maturity (from 4 March 2019 to 3 March 2025), which allows bondholders to convert the bond into a predetermined number of shares during the conversion period. In addition to this standard conversion feature, the contract incorporates typical path-dependent adjustment and redemption mechanisms that are widely adopted in the Chinese convertible bond market \citep{mcguinness2010listing, zhang2023financing}.

In particular, for this CITIC Bank CB, the conversion reset is activated when the stock price stays below 80\% of the current conversion price for at least 15 trading days within any consecutive 30 trading days. On the other side, when the stock price exceeds 130\% of the conversion price for at least 15 trading days within a 30-day rolling window, the issuer has the right to call the bond before maturity. These path-dependent provisions are designed to balance the interests of investors and issuers: the reset clause provides downside protection by adjusting the conversion price downward, while the call provision enables the issuer to force conversion when the equity option becomes sufficiently deep in the money. To simplify the analysis and focus on the joint 
path-dependent effects of call and reset provisions, we adopt the following assumptions.

\begin{assumption}\label{assump1}
For the pricing framework developed in this paper, we impose the following assumptions:
\begin{enumerate}
  
    \item  \textcolor{royalblue}{Following the resettable convertible-bond model of
    \citet{kimura2006monte}, the underlying stock pays no dividends. We also abstract from rights issues, share placements, and the associated corporate-action-based anti-dilution adjustments
    \citep{woronoff2005antidilution} so that analysis can focus on the
    path-dependent reset and call provisions.}\footnote{\textcolor{royalblue}{The implications of dividends and
    corporate-action-based anti-dilution adjustments, together with a
    possible extension of the pricing framework, are discussed in
    Appendix~\ref{app:antidilution}.}}
    
     \item \textcolor{royalblue}{For analytical tractability, the benchmark model assumes that voluntary early conversion prior to the triggering of the call provision condition is suboptimal and therefore excludes the holder's stopping
    decision from the ordinary continuation region. This is a modelling simplification rather than a contractual restriction. Once the call condition is satisfied, the issuer has the right to redeem the bond and the holder has the right to convert it into shares.}
    
    \item \textcolor{royalblue}{A reset-eligible state is determined by a rolling-window condition
    under which the closing stock price falls below a specified proportion of the prevailing conversion price. For analytical tractability, the revised conversion price is set equal to the stock price at the threshold. The benchmark abstracts from the approval procedure and the contractual floors based on trading-price averages, the latest audited net asset value per share, and the par value per share.}
    
    \item Credit risk is omitted in the baseline model to isolate and emphasise the impact of contractual path dependence. Also, as the vast majority of Chinese CBs come with guarantees provided by high-credit-rated banks, credit risks associated with China’s CBs can be reasonably excluded.
    
    \item Under the risk-neutral measure, the underlying stock price follows an arbitrage-free stochastic process. We consider three specifications: GBM, CEV, and the Heston models.
  \item \textcolor{royalblue}{
Let $0=t_0<t_1<\cdots<t_N=T$ denote the contractual monitoring
dates. On each interval $(t_i,t_{i+1})$, the value functional is
non-anticipative and is interpreted in the viscosity sense using
$C^{1,2}$ test functionals in the sense of Dupire
\citep{ContFournie2013}. No $C^{1,2}$ regularity of the value
functional is imposed across the monitoring dates.
}
    \item \textcolor{royalblue}{If the call and reset conditions are satisfied simultaneously, we evaluate the call condition first to ensure that the contractual boundary conditions are mutually exclusive.}
\end{enumerate}
\end{assumption}

Under these assumptions, the contractual provisions explicitly depend on the historical evolution of the underlying stock price, making the valuation problem inherently path-dependent. This motivates us to formulate the pricing problem within a path-dependent partial differential equation (PPDE) framework presented in the following theorem.

\begin{theorem}[\textcolor{blue}{PPDE representation for convertible bonds with reset and call provisions}]
\label{thm:ppde_cb}
 Let $\omega_t=\{S_u:0\le u\le t\}$ be the historical path of the underlying stock price $S_u$ up to time $t$, and let $H_t$ be the current conversion price. 
\textcolor{royalblue}{The convertible bond value is represented by the functional}
\begin{equation}\label{eq:valuefunc}
V(t,\omega_t,H_t),
\end{equation}
\textcolor{royalblue}{which is interpreted as a viscosity solution in the
continuation region \eqref{eq:C}.}
Under the risk-neutral measure $\mathbb Q$, let $r_t$ be the risk-free interest rate, $\sigma(t,\omega_t,H_t)$ be the volatility function and $W=(W_t)_{t\geq 0}$ be a standard Brownian motion. The stock price process follows:
\begin{equation}
dS_t = r_t S_t\,dt + \sigma(t,\omega_t,H_t)S_t\,dW_t .
\end{equation}
We then define the two path-dependent counting processes $N_t^{\mathrm{reset}}$ and $N_t^{\mathrm{call}}$ for the rolling-window triggers of reset and call, respectively, as follows:
\begin{align}
N_t^{\mathrm{reset}}(\omega_t;H_t)
&=
\sum_{i=1}^{30}
\mathbbm{1}_{\{S_{t-i}<a H_t\}},\\
N_t^{\mathrm{call}}(\omega_t;H_t)
&=
\sum_{i=1}^{30}
\mathbbm{1}_{\{S_{t-i}\ge b H_t\}},
\end{align}
where $a\in (0,1)$ and $b>1$ denote the downward revision ratio and the call threshold ratio, respectively.
The continuation region where neither provision is triggered is defined as
\begin{equation}\label{eq:C}
\mathcal C
=
\left\{
(t,\omega_t,H_t):
N_t^{\mathrm{reset}}(\omega_t;H_t)<N_r
\;\text{and}\;
N_t^{\mathrm{call}}(\omega_t;H_t)<N_c
\right\}.
\end{equation}
where $N_r$ and $N_c$ are the contractual minimum number of trigger days required for reset and call activation. In the following specific model settings (GBM, CEV, and Heston), the path-dependent functional $V(t,\omega_t,H_t)$ admits a Markovian representation by augmenting the state space with sufficient statistics of the path. In such cases, we write the value function as $V(t,S_t,H_t)$ or $V(t,S_t,v_t,H_t)$ for notational convenience.

Let $B$ be the bond face value, $C$ be the annual coupon rate, $K_t$ the redemption call price at time $t$ which is a par value $B$ plus accrued interest, and $\Psi(\omega_t,H_t)$ be the conversion price adjustment function triggered by the downward reset provision. Then, the convertible bond value function $V(t,\omega_t,H_t)$ satisfies the following path-dependent partial differential equation with boundary and terminal conditions:
\begin{equation}\label{eq:ppde}
\left\{
\begin{aligned}
&\partial_t V(t,\omega_t,H_t)
+\mathcal{A}V(t,\omega_t,H_t)
-r_t V(t,\omega_t,H_t)
+B\cdot C = 0,
&& (t,\omega_t,H_t)\in\mathcal{C}, \\[6pt]
&V(t,\omega_t,H_t)
=
\max\left(
\frac{B}{H_t}S_t,\;
\min\bigl(V^{\mathrm{cont}}(t,\omega_t,H_t),K_t\bigr)
\right),
&& N_t^{\mathrm{call}}(\omega_t;H_t)\ge N_c, \\[6pt]
&V(t,\omega_t,H_t)
=
V\bigl(t,\omega_t,\Psi(\omega_t,H_t)\bigr),
&& N_t^{\mathrm{reset}}(\omega_t;H_t)\ge N_r, \\[6pt]
&V(T,\omega_T,H_T)
=
\max\!\left(B,\frac{B}{H_T}S_T\right).
\end{aligned}
\right.
\end{equation}
\textcolor{royalblue}{The call and reset conditions above are understood according to a mutually exclusive partition of the state space.}\footnote{\textcolor{royalblue}{In the exceptional case of simultaneous activation,
we adopt a call-first convention according to Assumption~\ref{assump1}:
\[
\underbrace{
\{N^{\mathrm{call}}<N_c\;\text{and}\; 
N^{\mathrm{reset}}<N_r\}
}_{\text{continuation}}
\cup
\underbrace{
\{N^{\mathrm{call}}\ge N_c\}
}_{\text{call}}
\cup
\underbrace{
\{N^{\mathrm{call}}<N_c\;\text{and}\;
N^{\mathrm{reset}}\ge N_r\}
}_{\text{reset}}.
\]}}
In equation~\eqref{eq:ppde}, $\mathcal A$ denotes the infinitesimal generator associated with the underlying price dynamics and $\Psi(\omega_t,H_t)$ denotes the reset mapping that specifies the updated conversion price once the reset condition is triggered. In practice, the revised conversion price is determined by recent stock price averages and other contractual constraints. \textcolor{royalblue}{Then, $\partial_t V$ denotes the horizontal derivative, while
$\partial_{\omega}V$ and $\partial_{\omega\omega}V$ denote the first- and second-order vertical derivatives, respectively.}
Finally, the continuation value \[
V^{\mathrm{cont}}(t,\omega_t,H_t)
:=
\mathbb{E}^{\mathbb Q}\!\left[
e^{-\int_t^{t+\Delta} r_s ds}
V(t+\Delta,\omega_{t+\Delta},H_{t+\Delta})
+
\int_t^{t+\Delta} e^{-\int_t^u r_s ds} BC\,du
\;\middle|\;
\omega_t,H_t
\right]
\]
represents the bond value when the conversion option is not exercised although the call provision has been triggered, and $\Delta>0$ is a small time increment.\footnote{Corresponding to the discretisation step size $h$ in the dynamic programming scheme.} \textcolor{royalblue}{\textcolor{royalblue}{
    The functional It\^{o} formula is applied only locally in the continuation region $\mathcal C$, where the conversion price $H$ remains unchanged. At the reset boundary, the change in $H$ is imposed through the nonlocal switching condition
    $
    V(t,\omega_t,H)
    =
    V\!\left(t,\omega_t,\Psi(\omega_t,H)\right),
    $
    rather than being included in the local differential operator. Consequently, no functional $C^{1,2}$ regularity is required across the reset boundary. If a classical solution does not exist, the PPDE is interpreted in the viscosity sense in $\mathcal C$, subject to the call, reset, and terminal conditions
    in \eqref{eq:ppde}.
    }}
\end{theorem}

\begin{proof}
    A detailed proof is given in Appendix~\ref{app:proof1}.
\end{proof}

\begin{remark}
    The above formulation captures the essential contractual features of Chinese convertible bonds. The call provision introduces an
    early termination condition, while the downward reset clause
    introduces a regime-switching structure that endogenously updates the conversion price upon trigger activation.
    Since both trigger events depend on the historical stock price
    path, the valuation problem naturally leads to a PPDE framework. \textcolor{royalblue}{The PPDE applies between consecutive monitoring dates. The reset and call conditions are imposed through contractual transmission rules at the monitoring dates as formalised in Theorem~\ref{thm:ppde_cb_viscosity}.}
\end{remark}

\begin{remark}
    Regarding the continuation value situation, when the stock price exceeds the conversion price, bondholders have an incentive to delay conversion until forced redemption in order to benefit from further upside. Converting immediately would forfeit the remaining fixed-income component and the potential for additional gains.
\end{remark}

\subsection{Model-specific PPDE formulations}

Theorem~\ref{thm:ppde_cb} establishes a general PPDE representation for CBs with typical features in the Chinese CB market, namely reset and call provisions. We now derive the explicit forms of the pricing equations under three commonly used stock price dynamics, the GBM, CEV, and the Heston model via specifying the infinitesimal generator associated for each dynamics. In doing so, the general PPDE~\eqref{eq:ppde} can be specialised into three concrete valuation equations shown below.

\subsubsection{GBM model}
Suppose that the stock price
process $S=\{S_t\}_{t\ge0}$ follows a geometric Brownian motion
\begin{equation}
dS_t = r S_t\, dt + \sigma S_t\, dW_t, \qquad S_0>0,
\end{equation}
under the risk-neutral probability measure $\mathbb{Q}$, where
$r>0$ denotes the constant risk-free interest rate, $\sigma>0$ is
the constant volatility parameter, and $W=\{W_t\}_{t\ge0}$ is a standard
Brownian motion defined on a filtered probability space
$(\Omega,\mathcal{F},\{\mathcal{F}_t\}_{t\geq0},\mathbb{Q})$.
Based on Theorem~\ref{thm:ppde_cb}, the PPDE representation under the GBM dynamics can be derived as follows.
\begin{corollary}[\textcolor{blue}{PPDE under the GBM model}]
\label{cor:gbm_ppde}
Under the GBM stock price dynamics, the value function
$V(t,S,H_t)$ of a convertible bond with downward reset
and call provisions satisfies
\begin{equation}\label{eq:gbm}
\left\{
\begin{aligned}
&\partial_t V
+
r
S\partial_S V
+
\frac12\sigma^2 S^2\partial_{SS}V
-rV
+B\cdot C
=0,
&& (t,S,H_t)\in\mathcal C, \\[6pt]
&V(t,S,H_t)
=
\max\left(
\frac{B}{H_t}S,\;
\min\bigl(V^{\mathrm{cont}}(t,S,H_t),K_t\bigr)
\right),
&& N_t^{\mathrm{call}}\ge N_c, \\[6pt]
&V(t,S,H_t)
=
V\!\left(t,S,\Psi(\omega_t,H_t)\right),
&& N_t^{\mathrm{reset}}\ge N_r, \\[6pt]
&V(T,S,H_T)=
\max\!\left(B,\frac{B}{H_T}S_T\right),
\end{aligned}
\right.
\end{equation}
where
\[
V^{\mathrm{cont}}(t,S,H_t)
=
\mathbb{E}^{\mathbb Q}\!\left[
e^{-\int_t^{t+\Delta} r_u\,du}\,
V(t+\Delta,S_{t+\Delta},H_{t+\Delta})
+
\int_t^{t+\Delta} e^{-\int_t^s r_u\,du}\, B C \, ds
\;\middle|\;
S_t=S,\; H_t
\right].
\]
\end{corollary}
\begin{proof}
    See Appendix~\ref{app:proofGBM}.
\end{proof}

\subsubsection{CEV model}

Suppose that the stock price process $S=\{S_t\}_{t\ge0}$ follows the
CEV model
\begin{equation}
dS_t = r S_t\, dt + \sigma S_t^{\gamma}\, dW_t, \qquad S_0>0,
\end{equation}
under the risk-neutral probability measure $\mathbb{Q}$, where
$r>0$ denotes the constant risk-free interest rate, $\sigma>0$
is the volatility parameter, $\gamma$ represents the elasticity
coefficient, and $W=\{W_t\}_{t\ge0}$ is a standard Brownian motion
defined on a filtered probability space
$(\Omega,\mathcal{F},\{\mathcal{F}_t\}_{t\geq0},\mathbb{Q})$.
Based on Theorem~\ref{thm:ppde_cb}, the PPDE representation under
the CEV dynamics is derived as follows.

\begin{corollary}[\textcolor{blue}{PPDE under the CEV model}]
\label{cor:cev_ppde}

Under the CEV stock price dynamics, the value function
$V(t,S,H_t)$ of a convertible bond with downward reset
and call provisions satisfies 
\begin{equation}\label{eq:cev}
\left\{
\begin{aligned}
&\partial_t V
+
rS\partial_S V
+
\frac12\sigma^2 S^{2\gamma}\partial_{SS}V
-rV
+B\cdot C
=0,
&& (t,S,H_t)\in\mathcal C, \\[6pt]
&V(t,S,H_t)=\max\left(
\frac{B}{H_t}S,\;
\min\bigl(V^{\mathrm{cont}}(t,S,H_t),K_t\bigr)
\right),
&& N_t^{\mathrm{call}}\ge N_c, \\[6pt]
&V(t,S,H_t)=V\!\left(t,S,\Psi(\omega_t,H_t)\right),
&& N_t^{\mathrm{reset}}\ge N_r, \\[6pt]
&V(T,S,H_T)=
\max\!\left(B,\frac{B}{H_T}S_T\right),
\end{aligned}
\right.
\end{equation}
where
\[
V^{\mathrm{cont}}(t,S,H_t)
=
\mathbb{E}^{\mathbb Q}\!\left[
e^{-\int_t^{t+\Delta} r_u\,du}\,
V(t+\Delta,S_{t+\Delta},H_{t+\Delta})
+
\int_t^{t+\Delta} e^{-\int_t^s r_u\,du}\, B C \, ds
\;\middle|\;
S_t=S,\; H_t
\right].
\]
\end{corollary}
\begin{proof}
    See Appendix~\ref{appenCEV}.
\end{proof}

\subsubsection{Heston stochastic volatility model}

Suppose that the stock price process $S=\{S_t\}_{t\ge0}$ follows
the Heston stochastic volatility model
\begin{equation}\label{eq:heston_model}
\begin{aligned}
dS_t &= r S_t\, dt + \sqrt{v_t}\, S_t\, dW_t^S, \qquad S_0>0,\\
dv_t &= \kappa(\theta-v_t)\, dt + \eta\sqrt{v_t}\, dW_t^v, \qquad v_0>0,
\end{aligned}
\end{equation}
under the risk-neutral probability measure $\mathbb{Q}$, where $v_t$ is the stochastic variance of the stock price at time $t\geq 0$,
$r>0$ is the constant risk-free interest rate, $\kappa>0$
is the mean-reversion rate, $\theta>0$ represents the long-term
variance level, $\eta>0$ is the volatility of volatility, and $\rho\in[-1,1]$ is the correlation coefficient between the stock price and variance. The two standard Brownian motions $W^S=\{W_t^S\}$ and $W^v=\{W_t^v\}$ satisfy the correlation condition
\[
dW_t^S dW_t^v = \rho\, dt.
\]
Based on Theorem~\ref{thm:ppde_cb}, the PPDE representation under
the Heston stochastic volatility dynamics is derived as follows.

\begin{corollary}[\textcolor{blue}{PPDE under the Heston model}]
\label{cor:heston_ppde}

Under the Heston stochastic volatility dynamics, the value function
$V(t,S,v_t,H_t)$ of a convertible bond with downward reset
and call provisions satisfies

\begin{equation}\label{eq:heston}
\left\{
\begin{aligned}
&\partial_t V
\!+\!
rS\partial_S V
\!+\!
\kappa(\theta-v)\partial_v V
\!+\!
\frac12 vS^2\partial_{SS}V
\!+\!
\frac12\eta^2 v\partial_{vv}V
\!+\!
\rho\eta vS\partial_{Sv}V
\!-rV
\!+\!B\cdot C
=0,
&& (t,S,v,H_t)\in\mathcal C, \\[6pt]
&V(t,S,v,H_t)=\max\left(
\frac{B}{H_t}S,\;
\min\bigl(V^{\mathrm{cont}}(t,S,v,H_t),K_t\bigr)
\right),
&& N_t^{\mathrm{call}}\ge N_c, \\[6pt]
&V(t,S,v,H_t)=V\!\left(t,S,v,\Psi(\omega_t,H_t)\right),
&& N_t^{\mathrm{reset}}\ge N_r, \\[6pt]
&V(T,S,v,H_T)=
\max\!\left(B,\frac{B}{H_T}S_T\right),
\end{aligned}
\right.
\end{equation}
where
\[
V^{\mathrm{cont}}(t,S,v,H_t)
=
\mathbb{E}^{\mathbb Q}\!\!\left[
e^{-\!\int_t^{t+\Delta}\!r_udu}\,
V(t\!+\!\Delta,S_{t+\Delta},v_{t+\Delta},H_{t+\Delta})
+\!\!
\int_t^{t+\Delta}\!\!\!\!\! e^{-\!\int_t^s \!r_udu}\, B C ds
\middle|
S_t\!=S, v_t\!=v, H_t
\right]
\]
and $v_t$ denotes the instantaneous variance process. For notational simplicity, we write $v$ in place of $v_t$ in (\ref{eq:heston}).

\end{corollary}
\begin{proof}
    See Appendix~\ref{appenheston}. Here, the value function depends on the latent variance state $v$, while the continuation region $\mathcal C$ is written in accordance with its earlier definition, which is expressed in terms of $(t,S,H_t)$.
\end{proof}

\section{Deep learning approximation in the backward recursion}\label{sec:deepapp}

In this section, we establish the theoretical foundation of the proposed deep learning-based numerical scheme for solving this path-dependent pricing problem. Starting from a discrete-time dynamic programming formulation, we use neural networks to approximate the conditional expectations involved in the backward recursion. \textcolor{royalblue}{We then analyse how the neural-network approximation errors propagate and accumulate through the backward recursion on a fixed time grid.} We also provide detailed implementation, including the neural network approximation, training procedure, and overall numerical workflow, to ensure the method is reproducible.

\subsection{Fixed-grid convergence analysis}

Let $T>0$ denote the contract maturity, $N$ the total number of time steps and $h = T/N$ the uniform time mesh size. We define the discrete time grid as $t_i = ih$, $i=0,1,\dots,N$. The vector of state variables is defined as

\begin{equation}\label{eq:X}
 X_i = (S_i, H_i, \mathcal{P}_i),   
\end{equation}
where
\begin{itemize}
\item $S_i := S_{t_i}$ denotes the underlying stock price at time $t_i$ under the risk-neutral measurer $\mathbb{Q}$.

\item $H_i$ denotes the conversion price of the convertible bond at time $t_i$. This quantity evolves over time due to the downward reset provision, and is updated according to a contractually specified rule whenever the reset condition is triggered.

\item $\mathcal{P}_i$ represents a collection of path-dependent summary statistics constructed from the historical price trajectory $\{S_{t_j}\}_{j \le i}$. It encodes all information needed to evaluate the trigger conditions, including the rolling-window counts of the form
\[
N_i^{\mathrm{reset}}(x) = \sum_{j=1}^{30} \mathbbm{1}_{\{S_{i-j}(x) < aH\}},
\quad
N_i^{\mathrm{call}}(x) = \sum_{j=1}^{30} \mathbbm{1}_{\{S_{i-j}(x) \ge b H\}},
\]
where $a\in(0,1)$ and $b>1$, as well as any additional auxiliary variables, such as moving averages. Here $x$ denotes a generic realisation (deterministic value) of $X_i$.
\end{itemize}
This finite-dimensional state variable $X_i$ provides a Markovian representation of the original path-dependent pricing problem, which allows us to apply dynamic programming and neural network-based regression methods. Based on the rolling-window counting processes, we define the indicator functions for call and reset triggers as follows,
\[
C_i^{\mathrm{trig}}(x) := \mathbbm{1}_{\{N^{\mathrm{call}}_i(x) \ge N_c\}}, 
\qquad
R_i^{\mathrm{trig}}(x) := \mathbbm{1}_{\{N^{\mathrm{reset}}_i(x) \ge N_r\}}.
\]
where $N_c$ and $N_r$ are the contractual threshold days for call and reset activation, respectively. Consistent with the state-space partition established in Assumption \ref{assump1}, we adopt a call-first priority convention. The call provision takes precedence over the reset provision, meaning that the call condition is verified first when both triggers are activated simultaneously. The reset rule is incorporated into the state dynamics through the pathwise update of $H_i$ whenever the reset condition is triggered. 
We next define the one-step backward operator $\mathcal{T}_i$ as follows: 
\[
\mathcal{T}_i \phi(x)
=
\begin{cases}
\displaystyle
\max\left(
\frac{B}{H_i(x)}S_i(x),\;
\min\left(
\mathbb{E}\!\left[e^{-r_i h}\phi(X_{i+1}) + h\,B \cdot C \mid X_i = x \right],
\, K_i
\right)
\right),
& \text{if } C^{\mathrm{trig}}_i(x)=1, \\[10pt]
\displaystyle
\mathbb{E}\!\left[e^{-r_i h}\phi(X_{i+1}) + h\,B \cdot C \mid X_i = x \right],
& \text{otherwise},
\end{cases}
\]
where $\phi(\cdot) $ denotes a measurable function on the state space, representing the value function at the next time step $t_{i+1}$. 

The operator $\mathcal{T}_i$ defines the one-step backward recursion of the convertible bond value and integrates all three key contractual features. First, when the bond call condition is satisfied, i.e., $C_i(x)=1$, the issuer has the right to exercise the call option and redeem the bond at price $K_i$, while bondholders retain the right to convert the bond into equity. When $K_i<V^{\mathrm{cont}}_i(x)$, the issuer exercises the call option and the bond value is immediately determined by the call price $K_i$, corresponding to an optimal stopping event. Otherwise, the bondholder may either continue holding the bond or convert to stocks.
Second, when the reset condition is triggered, i.e., $R_i^{\mathrm{trig}}(x)=1$, the conversion price is adjusted according to the reset rule. Finally, in the continuation region $\mathcal{C}$ where neither condition is met, the bond value is equal to the discounted conditional expected next-step value plus coupon payments, consistent with the standard risk-neutral valuation principle. The priority of call provision ensures that the stopping region dominates the switching region when both conditions are simultaneously activated.

We now introduce the exact discrete-time value function $V_i^h$ generated by the backward recursion
\[
V_i^h = \mathcal{T}_i V_{i+1}^h, \quad \text{for}\quad i=0,1,\dots,N-1,\quad \text{and}\quad V_N^h = g(X_N),
\]
where $h$ is the mesh size, $g (\cdot)$ denotes the terminal payoff function at maturity $T$, and the value at each earlier time step is obtained by applying the operator $\mathcal{T}_i$ to the next-step value $V_{i+1}^h$. This recursion characterises the discrete-time solution through a backward dynamic programming scheme. 

{\color{royalblue}
\begin{theorem}[Piecewise viscosity characterisation of the convertible-bond value]
\label{thm:ppde_cb_viscosity}
Let $0=t_0<t_1<\cdots<t_N=T$, with $t_i=ih$ and $h=T/N$, be the uniform grid on which the contractual reset and call conditions are monitored.\footnote{\textcolor{royalblue}{In Theorems~\ref{thm:ppde_cb_viscosity} and~\ref{thm:conver}, $h$ refers to the same fixed time step size of the uniform time grid used for both contractual monitoring and the dynamic programming recursion. Our analysis does not consider the limit $h\to0$.}} Between two monitoring
dates, $H_i$ and $\mathcal P_i$ in \eqref{eq:X} remain fixed and the stock price path evolves
according to
\begin{equation}\label{eq:unified_state}
dS_t
=r_tS_t\,dt
+\sigma(t,\omega_t,H_i)S_t\,dW_t,
\qquad t\in(t_i,t_{i+1}).
\end{equation}
For a stopped stock-price path $\omega_t=(\omega_s)_{0\leq s\leq t}$, define
the $t$-norm by $\|\omega\|_t:=\sup_{0\leq s\leq t}|\omega_s|$ and set
\begin{equation*}
\mu(t,\omega,H):=r_t\omega_t,
\qquad
\nu(t,\omega,H):=\sigma(t,\omega,H)\omega_t.
\end{equation*}
The interval pathwise generator acting on
$\varphi\in C^{1,2}$ is defined by
\begin{equation}
\mathcal A\varphi
:=
\mu(t,\omega,H)\partial_\omega\varphi
+\frac12\nu^2(t,\omega,H)
\partial_{\omega\omega}^2\varphi.\label{eq:generator}
\end{equation}
Assume the following conditions.
\begin{enumerate}
\item[(i)] The coefficients $\mu$ and $\nu$ are non-anticipative and continuous.
For some $L>0$, they satisfy the pathwise Lipschitz condition
\begin{align}
|\mu(t,\omega,H)-\mu(t,\omega',H')|
+|\nu(t,\omega,H)-\nu(t,\omega',H')|
&\le L\bigl(\|\omega-\omega'\|_t+|H-H'|\bigr),
\label{eq:path_lipschitz_proof}
\end{align}
and the linear-growth condition
\begin{align}
|\mu(t,\omega,H)|+|\nu(t,\omega,H)|
&\le L\bigl(1+\|\omega\|_t+|H|\bigr).
\label{eq:linear_growth_proof}
\end{align}

\item[(ii)] On every bounded stopped-path domain
$
\Lambda_R:=\{(t,\omega):\|\omega\|_t\leq R\},
$ 
where $R$ is a positive constant, 
the coefficients are bounded and uniformly continuous. Their pathwise Lipschitz constants are independent of $R$. There also exists a nonnegative functional $\chi\in C^{1,2}$ whose growth dominates that of the candidate
solutions and which satisfies
\begin{equation}\label{eq:lyapunov_comparison_proof}
-\partial_t\chi-\mathcal A\chi+r_t\chi\geq1.
\end{equation}
where $\mathcal{A}$ is defined by~\eqref{eq:generator}.
\item[(iii)] The risk-free rate process $(r_t)_{0\leq t\leq T}$ is
non-anticipative and uniformly bounded:
$
0\leq r_t\leq\bar r.
$
The terminal payoff, coupon stream, call price, and contractual state-update
maps are measurable and have at most polynomial growth. For some $p\geq1$,
the conversion ratio satisfies the uniform $p$-th moment condition
\begin{equation}\label{eq:conversion_ratio_moment_proof}
\mathbb E^{\mathbb Q}\left[
\sup_{0\leq t\leq T}\left(\frac{S_t}{H_t}\right)^{\!p\,}
\right]<\infty.
\end{equation}

\item[(iv)] At each $t_i$, the ordered rolling window is retained. The call
condition is checked before the reset condition. At most one reset is applied.
A reset sets $H_i=S_i$.

\item[(v)] At maturity, no further call or reset decision is made. The terminal
payoff is computed using the conversion price $H_T$ prevailing at maturity:
\begin{equation}\label{eq:terminal_payoff_proof}
V(T,X_N)=\max\left\{B,\frac{B}{H_T}S_T\right\}.
\end{equation}
\end{enumerate}
Write $x=X_i=(S_i,H_i,\mathcal P_i)$. The updated state $\Gamma_i x$ is obtained by appending the current $S_i$ to $\mathcal P_i$ and dropping the oldest observation, while leaving $H_i$ unchanged. The reset state $\Gamma_i^R x$ uses the same update
of $\mathcal P_i$ and replaces $H_i$ with $S_i$.
For a post-monitoring value $\phi$ and a pre-maturity monitoring date
$t_i<T$, we define the transmission operator $\mathcal J_i$ according to \eqref{eq:ppde}
by
\begin{equation}\label{eq:J_operator_proof}
(\mathcal J_i\phi)(x)=
\begin{cases}
\displaystyle
\max\left\{\frac{B}{H_i}S_i,
\min\bigl[\phi(\Gamma_i x),K_i\bigr]\right\},
&C_i^{\rm trig}(x)=1,\\[7pt]
\phi(\Gamma_i^R x),
&C_i^{\rm trig}(x)=0,\ R_i^{\rm trig}(x)=1,\\[3pt]
\phi(\Gamma_i x),
&C_i^{\rm trig}(x)=R_i^{\rm trig}(x)=0.
\end{cases}
\end{equation}
Under conditions (i)--(v), the following statements hold.
\begin{enumerate}
\item[(i)] \emph{Interval PPDE.} On each open monitoring interval
$(t_i,t_{i+1})$, the risk-neutral convertible-bond value solves
\begin{equation}\label{eq:unified_interval_ppde}
-\partial_tV-\mathcal AV+r_tV-BC=0,
\end{equation}
where
\begin{equation*}
\mathcal A\varphi
=r_tS_t\partial_\omega\varphi
+\frac12\sigma^2(t,\omega_t,H_i)S_t^2
\partial_{\omega\omega}^2\varphi.
\end{equation*}
For a prescribed value at the left limit $t_{i+1}^-$, the interval PPDE has a unique viscosity
solution in the polynomial-growth class.

\item[(ii)] \emph{Monitoring-date transmission.} The interval solutions are
connected at every pre-maturity monitoring date $t_i<T$ by the contractual
transmission condition
\begin{equation}\label{eq:pointwise_transmission_proof}
V(t_i^-,x)=\bigl(\mathcal J_iV(t_i^+,\cdot)\bigr)(x).
\end{equation}
At a threshold boundary where a hard call or reset condition changes its
trigger status, the upper- and lower-semicontinuous envelope formulation
standard in viscosity obstacle problems is used \citep{ekren2017obstacle}.
The corresponding transmission condition is understood as
\begin{equation}\label{eq:relaxed_transmission_proof}
V^*(t_i^-,x)\leq(\mathcal J_iV(t_i^+,\cdot))^*(x),
\qquad
V_*(t_i^-,x)\geq(\mathcal J_iV(t_i^+,\cdot))_*(x).
\end{equation}
\end{enumerate}
Starting from the terminal payoff and proceeding backward through the
monitoring dates uniquely determines $V$ on the whole horizon $[0,T]$.
Accordingly, $V$ solves equation~\eqref{eq:ppde} in the viscosity sense on each open
interval $(t_i,t_{i+1})$ and satisfies the contractual transmission conditions
at the monitoring dates.
\end{theorem}

\begin{proof}
A detailed proof is provided in Appendix~\ref{app:proof_ppde_viscosity}.
\end{proof}
}
{\color{royalblue}
\begin{remark}[Extension to Heston dynamics]
\label{rem:heston_viscosity}
The argument of Theorem~\ref{thm:ppde_cb_viscosity} can be extended to a multidimensional state containing the variance path. Let the augmented stopped path be
\[
\boldsymbol{\omega}_t=(\omega_t^S,\omega_t^v)
=\{(S_s,v_s):0\leq s\leq t\},
\]
where
$(S_t,v_t)$ satisfies the Heston model \eqref{eq:heston_model},
and $dW_t^S\,dW_t^v=\rho\,dt$. If $v_0>0$ and
$2\kappa\theta>\eta^2$, the variance process remains strictly positive. The Heston coefficients are locally regular on domains bounded away from $v=0$. Assume suitable localisation and moment conditions. Also assume the required Lyapunov and multidimensional comparison conditions. The multidimensional functional It\^o formula then gives the following interval operator for
$V(t,\boldsymbol{\omega}_t,H_i)$:
\begin{equation*}
\begin{aligned}
\mathcal A^H\varphi
={}&r_tS_t\partial_{\omega^S}\varphi
+\kappa(\theta-v_t)\partial_{\omega^v}\varphi
+\frac12v_tS_t^2\partial_{\omega^S\omega^S}^2\varphi 
+\rho\eta v_tS_t\partial_{\omega^S\omega^v}^2\varphi
+\frac12\eta^2v_t\partial_{\omega^v\omega^v}^2\varphi.
\end{aligned}
\end{equation*}
This operator agrees with the Heston PPDE established in Corollary~\ref{cor:heston_ppde}. Since the call-first transmission operator $\mathcal J_i$ remains order-preserving, the same backward transmission argument applies across the monitoring dates.
\end{remark}
}

In practice, the conditional expectation in $\mathcal{T}_i$ is analytically intractable, especially under high-dimensional and path-dependent settings. To address this issue, we approximate the value function at each time step $t_i$ via a parametric artificial neural network $\widehat V_i(\cdot;\theta_i)$, whose parameters $\theta_i$ are estimated via least-squares regression over simulated paths. 
The network approximating the value function $\widehat V_i$ can be interpreted as a realisation of the nonlinear operator $\mathcal{T}_i$ within a restricted functional class, transforming the original path-dependent problem into a tractable sequence of supervised learning tasks. In the literature, estimating conditional expectations by neural networks in a backward induction has been investigated in the context of finite-horizon stochastic control problems in \citet{hure2021convergence}. Such a ``deep least-squares Monte Carlo'' approach has recently been applied to option hedging \citep{fecamp2021deep,bachouch2022numerical} and pricing \citep{langrene2026deep}, storage optimisation \citep{bachouch2022numerical,warin2023reservoir}, portfolio optimisation \citep{franco2022discrete,roch2023optimal}, and utility maximisation \citep{arandjelovic2026solving}, but has so far never been used for the pricing of convertible bonds with path-dependent provisions.

\textcolor{royalblue}{We now establish theoretical results to characterise the stability of the dynamic programming operator and the propagation of neural network approximation errors on a fixed time grid in the following Proposition~\ref{prop:errorp}. The following lemma establishes the Lipschitz continuity of the one-step operator.}
\textcolor{royalblue}{
\begin{lemma}[Lipschitz continuity of the one-step operator]
    \label{lem:operator_lipschitz}
    Assume that the risk-free rate $r_i\ge 0$. Then, for any
    $L^2$-measurable functions $\phi$ and $\psi$, the one-step operator satisfies
    \[
    \|\mathcal T_i\phi-\mathcal T_i\psi\|_{L^2}
    \leq
    e^{-r_ih}\|\phi-\psi\|_{L^2}.
    \]
    Consequently, $\mathcal T_i$ is non-expansive in $L^2$, since
    $e^{-r_ih}\leq1$.
    \end{lemma}
    \begin{proof}
    A detailed proof is given in Appendix~\ref{appen:lem1}.
    \end{proof}
}
\textcolor{royalblue}{
\begin{proposition}[One-step propagation of the neural network approximation error]
\label{prop:errorp}
Suppose that the neural network approximation at time step $i$ satisfies the local error bound
\begin{equation}\label{eq:errorp}
\|\widehat V_i - \mathcal{T}_i \widehat V_{i+1}\|_{L^2} \le \varepsilon_i.
\end{equation}
Then, the approximation error follows the recursive relation
\[
\|\widehat V_i - V_i^h\|_{L^2}
\le
\varepsilon_i + e^{-r_i h}
\|\widehat V_{i+1} - V_{i+1}^h\|_{L^2}.
\]
\end{proposition}
\begin{proof}
A detailed proof is given in Appendix~\ref{appen:prop1}.
\end{proof}
}

\textcolor{royalblue}{Iterating the one-step error estimate from Proposition~\ref{prop:errorp} yields the following global error bound for the full backward recursion neural network approximation on a fixed time grid.}

\textcolor{royalblue}{\begin{theorem}[Global neural network approximation error on a fixed time grid]
\label{thm:conver}
Fix the uniform time grid with constant step size $h$
\[
0=t_0<t_1<\cdots<t_N=T,
\qquad h=t_{i+1}-t_i.
\]
Assume that the terminal payoff is imposed exactly, i.e.,
$
\widehat V_N=V_N^h,
$
and suppose that the one-step neural network approximation error satisfies \eqref{eq:errorp} for every $i=0,\ldots,N-1$.
Then the global pricing error satisfies
\[
\left\|
\widehat V_0-V_0^h
\right\|_{L^2}
\leq
\varepsilon_0
+
\sum_{j=1}^{N-1}
\exp\left(
-h\sum_{k=0}^{j-1}r_k
\right)\varepsilon_j.
\]
In particular,
\[
\left\|
\widehat V_0-V_0^h
\right\|_{L^2}
\leq
\sum_{j=0}^{N-1}\varepsilon_j.
\]
Consequently, for fixed time grid, if the cumulative network training error
\[
\sum_{j=0}^{N-1}\varepsilon_j\longrightarrow0
\]
as the neural network approximation capacity and training accuracy increase, then
\[
\widehat V_0\longrightarrow V_0^h
\qquad\text{in }L^2.
\]
\end{theorem}
\begin{proof}
A detailed proof is given in Appendix~\ref{appen:thm2}.
\end{proof}}

\textcolor{royalblue}{\begin{remark}[Scope and limitation of the convergence results]
\label{rem:convergence_scope}
The preceding Lemma~\ref{lem:operator_lipschitz}, Proposition~\ref{prop:errorp}, and Theorem~\ref{thm:conver} establish the propagation and accumulation of the neural-network approximation error relative to the exact fixed-grid dynamic programming recursion. Thus, Theorem~\ref{thm:conver} shows that $\widehat V_0\to V_0^h$ in $L^2$ as the aggregate training error vanishes, while Theorem~\ref{thm:ppde_cb_viscosity} establishes the existence and uniqueness of the piecewise viscosity solution. Extending the convergence analysis to the limit $V_0^h\to V(0,\cdot)$ as $h\to0$ would require adapting the monotone viscosity-scheme frameworks of \citet{barles1991convergence} and their path-dependent extensions in \citet{zhang2014monotone} and \citet{ren2017convergence} to the present rolling-window setting. In particular, such an extension would need to account for uniform stability and the discontinuities induced by the contractual trigger boundaries, and therefore constitutes a separate theoretical analysis beyond the scope of this paper. To complement this discrete-time error analysis, Section~\ref{sec:grid_sensitivity} provides a time-grid refinement study to numerically validate the empirical stability and convergence of the discrete approximate scheme.
\end{remark}}

\subsection{Algorithm implementation}

\paragraph{Neural network approximation.}

For each time step $t_i$, we employ a feedforward neural network
$\widehat V_i(\cdot;\theta_i)$ to approximate the value function at time $t_i$.
In particular, $\widehat V_i$ is trained to approximate the one-step operator $\mathcal{T}_i$ applied to $\widehat V_{i+1}$:
\begin{equation}
\widehat V_i(X_i;\theta_i)
\approx
\mathbb{E}\!\left[
e^{-r_i h}\widehat V_{i+1}(X_{i+1}) + h\,B\cdot C
\mid X_i
\right],
\end{equation}
which holds for states where the call condition is not activated, with the reset effect already incorporated into the state evolution. Each $\widehat V_i$ is implemented as a fully connected neural network.
In our baseline network, We use three hidden layers, each containing 64 neurons; ReLU activation function is applied following each hidden layer. \textcolor{royalblue}{In Appendix~\ref{app:architecture_assessment}, we assess alternative choices for network depths, widths and activation functions, and show that the baseline specification yields a reasonable balance between pricing stability and computational cost.}
The network mapping is given by
\begin{align}
h^{(1)} &= \text{ReLU}(W_1 x + b_1), \\
h^{(2)} &= \text{ReLU}(W_2 h^{(1)} + b_2), \\
h^{(3)} &= \text{ReLU}(W_3 h^{(2)} + b_3), \\
\widehat V_i(x;\theta_i) &= W_4 h^{(3)} + b_4,
\end{align}
where $\theta_i = \{W_\ell, b_\ell\}_{\ell=1}^4$ denotes the collection of trainable parameters. The output is a scalar representing the continuation value.

\paragraph{State variables.}

The input features are constructed from the Markovian state representation (see equation~\eqref{eq:X})
\[
X_i = (S_i, H_i, \mathcal{P}_i), \quad \text{for}\quad i=0,1,\dots,N
\]
where $\mathcal{P}_i$ encodes the path-dependent statistics needed to evaluate the reset and call triggers. Rather than using $\mathcal{P}_i$ directly, we construct a finite-dimensional feature representation based on quantities derived from $\mathcal{P}_i$, including the trigger indicator functions for the call and reset events. These features provide sufficient information about the price path for numerical approximation.

For the GBM and CEV models, the underlying dynamics depend only on the state variables $(S_i,t_i)$, and the conversion price $H_i$ evolves deterministically according to the reset rule.
The neural network normalised input is therefore defined as
\[
x_i =
\left(
\frac{S_i}{S_0},
\frac{H_i}{H_0},
\frac{t_i}{T},
\mathbbm{1}_{\{C_i(X_i)=1\}},
\mathbbm{1}_{\{R_i^{\mathrm{trig}}(X_i)=1\}}
\right),
\]
where the indicator functions encode the relevant path-dependent information. For the Heston model, the vector of state variables additionally includes the variance process:
\[
X_i = (S_i, v_i, H_i, t_i),
\]
and the corresponding normalised input becomes
\[
x_i =
\left(
\frac{S_i}{S_0},
\frac{v_i}{v_0},
\frac{H_i}{H_0},
\frac{t_i}{T},
\mathbbm{1}_{\{C_i(X_i)=1\}},
\mathbbm{1}_{\{R_i^{\mathrm{trig}}(X_i)=1\}}
\right).
\]

\paragraph{Training objective.}

At each time step $t_i$, the neural network $\widehat V_i$ is trained by minimising the mean squared error between its output and the Monte Carlo estimate of the one-step operator $\mathcal{T}_i \widehat V_{i+1}$.
The loss function is defined as
\begin{equation}
 \mathcal{L}_i(\theta_i)
=
\frac{1}{M}
\sum_{k=1}^{M}
\left(
\widehat V_i(x_i^{(k)};\theta_i)
-
\bigl[e^{-r_i h}\widehat V_{i+1}(X_{i+1}^{(k)}) + h\,B\cdot C\bigr]
\right)^2,
\end{equation}
where $r_i$ denotes the risk-free rate at time $t_i$, which is assumed constant in our experiments, and $M$ is the number of Monte Carlo simulations.

\paragraph{Backward recursion.}

Starting from the terminal condition
\begin{equation}
\widehat V_N(X_N) = g(X_N),
\end{equation}
we recursively compute the value function backward in time.
At each time step $t_i$, the approximate value $\widehat V_i$ is updated according to

\begin{equation}
\widehat V_i(X_i) =
\begin{cases}
\displaystyle
\max\left(
\frac{B}{H_i}S_i,\;
\min\left(
e^{-r_i h}\widehat V_{i+1}(X_{i+1}) + h\,B\cdot C,\;
K_i
\right)
\right),
& C_i(X_i)=1, \\[10pt]
\displaystyle
e^{-r_i h}\widehat V_{i+1}(X_{i+1}) + h\,B\cdot C,
& \text{otherwise}.
\end{cases}
\end{equation}
The conversion price reset rule is incorporated into the state evolution through the pathwise update of $H_i$ during the simulation phase.

\paragraph{Optimisation.}
\textcolor{royalblue}{
The parameters $\theta_i$ are optimised via the Adam optimiser with an initial learning rate of $10^{-3}$ and mini-batches of paths simulated under the risk-neutral measure. The baseline specification does not impose an explicit weight regularisation, corresponding to $\lambda=0$. At each backward step, 10\% of the simulated training paths are reserved for validation. Training is terminated after three consecutive epochs without validation-loss improvement, and the learning rate is reduced when the validation loss plateaus.  We retain the parameter values yielding the lowest validation loss. Appendix~\ref{app:regularisation} reports a robustness analysis across a range of $L^2$-penalty coefficients and training paths sample sizes. This analysis reveals that moderate regularisation has only a limited effect on the estimated bond price, justifying our unregularised baseline for the main numerical experiments. The networks are trained sequentially backward in time following the dynamic programming recursion. The full network architecture and backward training procedure are summarised in Figure~\ref{fig:nn_architecture} and Algorithm~\ref{alg:deep_ppde}, respectively.
}
\begin{figure}[htb]
\centering
\resizebox{\textwidth}{!}{
\begin{tikzpicture}[
    scale=1.1,
    transform shape,
    >=latex,
    node distance=0.5cm,
    every node/.style={font=\normalsize},
    box/.style={
        draw,
        rounded corners,
        minimum width=2.0cm,
        minimum height=1.6cm,
        inner sep=1.2pt,
        align=center
    },
    arrow/.style={->, line width=1.5pt}
]

\node[box] (input) {Input features\\[2pt]
GBM:\\
$\left(\dfrac{S_i}{S_0},\dfrac{H_i}{H_0},\dfrac{t_i}{T},
\mathbbm{1}_{\{C_i(X_i)=1\}},
\mathbbm{1}_{\{R_i^{\mathrm{trig}}(X_i)=1\}}\right)$\\[6pt]
CEV:\\
$\left(\dfrac{S_i}{S_0},\dfrac{H_i}{H_0},\dfrac{t_i}{T},
\mathbbm{1}_{\{C_i(X_i)=1\}},
\mathbbm{1}_{\{R_i^{\mathrm{trig}}(X_i)=1\}}\right)$\\[6pt]
Heston:\\
$\left(\dfrac{S_i}{S_0},\dfrac{v_i}{v_0},\dfrac{H_i}{H_0},\dfrac{t_i}{T},
\mathbbm{1}_{\{C_i(X_i)=1\}},
\mathbbm{1}_{\{R_i^{\mathrm{trig}}(X_i)=1\}}\right)$};

\node[box, right=of input] (h1) {Hidden layer 1\\
64 neurons\\
ReLU};

\node[box, right=of h1] (h2) {Hidden layer 2\\
64 neurons\\
ReLU};

\node[box, right=of h2] (h3) {Hidden layer 3\\
64 neurons\\
ReLU};

\node[box, right=of h3] (output) {Output layer\\
1 neuron\\
$\widehat V_i(X_i;\theta_i)$\\[2pt]
$\approx \mathcal{T}_i\widehat V_{i+1}(X_i)$};

\draw[arrow] (input) -- (h1);
\draw[arrow] (h1) -- (h2);
\draw[arrow] (h2) -- (h3);
\draw[arrow] (h3) -- (output);

\end{tikzpicture}
}
\caption{Neural network architecture for approximating the continuation value in the backward recursion under different underlying GBM, CEV, and Heston models.}
\label{fig:nn_architecture}
\end{figure}

\begin{algorithm2e}[t!]
\DontPrintSemicolon
\SetAlgoLined
\small
\caption{Backward deep learning algorithm for path-dependent convertible bond pricing}
\label{alg:deep_ppde}
\vspace{1mm}
\KwIn{Initial state $X_0$, maturity $T$, number of time steps $N$, number of training paths $M_{\mathrm{train}}$, number of test paths $M_{\mathrm{test}}$, contract parameters, and model parameters (GBM/CEV/Heston)}
\KwOut{Approximate CB price $\widehat V_0(X_0)$}

\vspace{0.5mm}
Set $h=T/N$. Map the contractual rolling-window counts to the selected time grid. For each $i=0,1,\dots,N-1$, initialise a feedforward neural network
$\widehat C_i(\cdot;\theta_i)$ with three hidden layers, width $64$, and ReLU activation

\vspace{0.5mm}
\Comment*[l]{Simulate paths under the risk-neutral measure}
Simulate training paths $\{X_i^{(m)}\}_{i=0,\dots,N}^{m=1,\dots,M_{\mathrm{train}}}$ and test paths $\{\widetilde X_i^{(m)}\}_{i=0,\dots,N}^{m=1,\dots,M_{\mathrm{test}}}$. Record the call and reset trigger states. Update $H_i$ when the reset condition is satisfied and the call condition is not satisfied

\vspace{0.5mm}

Set $\widehat V_N^{(m)}
=
\max\!\left(B,\frac{B}{H_N^{(m)}}S_N^{(m)}\right),
\  m=1,\dots,M_{\mathrm{train}}$ \Comment*[l]{Set training terminal payoff}

Set $\widetilde V_N^{(m)}
=
\max\!\left(B,\frac{B}{\widetilde H_N^{(m)}}\widetilde S_N^{(m)}\right),
\  m=1,\dots,M_{\mathrm{test}}$ \Comment*[l]{Set test terminal payoff}

\For{$i=N-1,\dots,0$}{
\vspace{0.5mm}
Compute features $x_i^{(m)}=\Phi(X_i^{(m)}),
\  m=1,\dots,M_{\mathrm{train}}$, including features derived from the rolling-window stock-price history, the current trigger indicators, and, under Heston dynamics, the variance state \Comment*[l]{Construct neural network inputs}

\vspace{0.5mm}
$\mathrm{Cont}_i^{(m)}
=
e^{-r_i h}\widehat V_{i+1}^{(m)} + Bh\,C$ \Comment*[l]{Construct continuation targets}

\vspace{0.5mm}
\Comment*[l]{Train neural network at time step $i$}
Standardise the inputs and continuation-value targets. Reserve $10\%$ of the training paths for validation

\vspace{0.5mm}
Minimise $\mathcal L_i(\theta_i)
=
\frac{1}{M_{\mathrm{train}}}\sum_{m=1}^{M_{\mathrm{train}}}
\left|
\widehat C_i(x_i^{(m)};\theta_i)-\mathrm{Cont}_i^{(m)}
\right|^2$ using Adam, validation-based early stopping, and learning-rate adjustment

\vspace{0.5mm}
Transform the fitted continuation values back to the original scale and project them onto the range of the continuation-value targets

\vspace{0.5mm}
$\widehat{\mathrm{Cont}}_i^{(m)}
=\widehat C_i(x_i^{(m)};\theta_i)$ \Comment*[l]{Estimate conditional continuation value}

\vspace{0.5mm}
$\mathrm{Conv}_i^{(m)}
=
\frac{B}{H_i^{(m)}}S_i^{(m)}$ \Comment*[l]{Compute conversion value}

\vspace{0.5mm}
Update $\widehat V_i^{(m)}=
\begin{cases}
\max\!\left(\mathrm{Conv}_i^{(m)},\,\min\!\left(\widehat{\mathrm{Cont}}_i^{(m)},K_i\right)\right),
& N_i^{\mathrm{call},(m)}\ge N_c,\\[6pt]
\widehat{\mathrm{Cont}}_i^{(m)},
& \text{otherwise}
\end{cases}$ \Comment*[l]{Apply the contractual operator}

\vspace{0.5mm}
Compute the test-path features and evaluate $\widehat C_i(\cdot;\theta_i)$. Apply the same call-first contractual operator to obtain $\widetilde V_i^{(m)}$

}
\KwRet{$\widehat V_0(X_0)$}
\end{algorithm2e}

\section{Numerical Experiments}\label{sec:numerical}

In this section, we conduct numerical experiments to evaluate the performance of the proposed deep learning-based pricing framework. We first compare pricing results generated under  GBM, CEV, and Heston dynamics to assess the impact of model specification. We then analyse the impact of the downward reset and call provisions on bond values. Finally, we perform sensitivity analysis with respect to key contractual and market parameters.

To illustrate the practical relevance of our model, we consider a real-market convertible bond: the China CITIC Bank Convertible Bond (code: 113021), issued on March 4, 2019. It has a total issuance size of 40 billion yuan, with a face value of 100 yuan per unit, and was listed on the Shanghai Stock Exchange on March 19, 2019. The  bond term spans from March 4, 2019 to March 3, 2025, and the conversion period runs from September 11, 2019 to March 3, 2025. 

The contract incorporates two key path-dependent provisions. First, if the closing price of CITIC Bank’s A-shares is below 80\% of the effective conversion price for at least 15 of any consecutive 30 trading days, the board of directors has the right to propose a downward revision of the conversion price, which is subject to shareholder approval and regulatory constraints. The revised conversion price shall not be lower than the average trading price of the A-shares on the 30 trading days, 20 trading days, and the most recent trading day, nor can it be lower than the most recently audited net asset value per share or the par value of the stock.  Second, during the conversion period, if the closing price of CITIC Bank’s A-shares is not lower than the call barrier that is at least 130\% of the then-effective conversion price for at least 15 out of 30 consecutive trading days, the issuer may redeem all or part of the unconverted convertible bonds at par value plus accrued interest for the period. For simplicity, we set the call redemption price to $K=B$. We employ the contractual terms of this bond in our numerical analysis, with key parameters summarised in Table~\ref{tab:cb_params}.

\begin{table}[htbp]
    \centering
    \caption{Parameters for China CITIC convertible bond (code: 113021)}
    \begin{tabular*}{12cm}{@{\extracolsep{\fill}}lccc}
        \toprule
        Description & Symbol & Value & Unit \\
        \midrule
        Initial stock price          & $S_0$ & 6.40  & CNY \\
        Conversion price             & $H$   & 7.45  & CNY \\
        Downward revision ratio       & $a$   & 0.80  & --- \\
        Call threshold ratio         & $b$   & 1.30  & --- \\
        Time to maturity             & $T$   & 6     & years \\
        Face value                   & $B$   & 100   & CNY \\
        Annual coupon rate           & $C$   & 0.3   & \% \\
        Constant risk‑free rate      & $r$   & 1.6   & \% \\
        \bottomrule
    \end{tabular*}
    \label{tab:cb_params}
    \begin{minipage}{0.95\textwidth}
        \footnotesize
        \vspace{0.1cm}
        \textit{Notes:} The risk‑free rate $r$ is held constant across all numerical experiments.
    \end{minipage}
\end{table}

\subsection{Pricing under different stock dynamics}

\textcolor{royalblue}{We implement the proposed deep-learning algorithm using a fully connected feedforward neural network with three hidden layers of 64 ReLU neurons each.
At each backward time step, the network is trained using the Adam optimiser with mini-batches of size 2,048, an initial learning rate of $10^{-3}$, and
at most 20 epochs. Ten percent of the training sample is reserved for validation with early stopping and validation-based learning-rate adjustment.
Both the inputs and regression targets are standardised at each time step. The network predictions are projected onto the range of the corresponding one-step training targets. We use 104 time steps per year, yielding $N=624$ steps over the six-year maturity, together with 12,000 training paths and
4,000 independent test paths. The GBM, CEV, and Heston parameters are reported in Table~\ref{tab:model_parameters}.}\footnote{\textcolor{royalblue}{
All experiments use the same baseline contractual and model parameters reported in Tables~\ref{tab:cb_params} and~\ref{tab:model_parameters}. Numerical price levels may differ slightly across experiments because they are obtained from separate simulation and training runs. Comparisons within each experiment are based on paired runs under an identical numerical specification.
}}
\begin{table}
\centering
\caption{Parameter specifications for GBM, CEV, and Heston models}
\begin{tabular*}{14cm}{@{\extracolsep{\fill}}cccccccc}
\toprule
Model 
& $\sigma$ 
& $\gamma$ 
& $v_0$ 
& $\kappa$ 
& $\theta$ 
& $\eta$ 
& $\rho$ \\
\midrule
GBM     & 0.30 & —    & —    & —    & —    & —    & — \\
CEV     & 0.35 & 0.90 & —    & —    & —    & —    & — \\
Heston  & —    & —    & 0.09 & 2.0  & 0.09 & 0.45 & -0.50 \\
\bottomrule
\end{tabular*}
\label{tab:model_parameters}
\end{table}

\begin{figure}
    \centering
    
    \begin{subfigure}{1.0\textwidth}
        \centering
        \includegraphics[width=\linewidth]{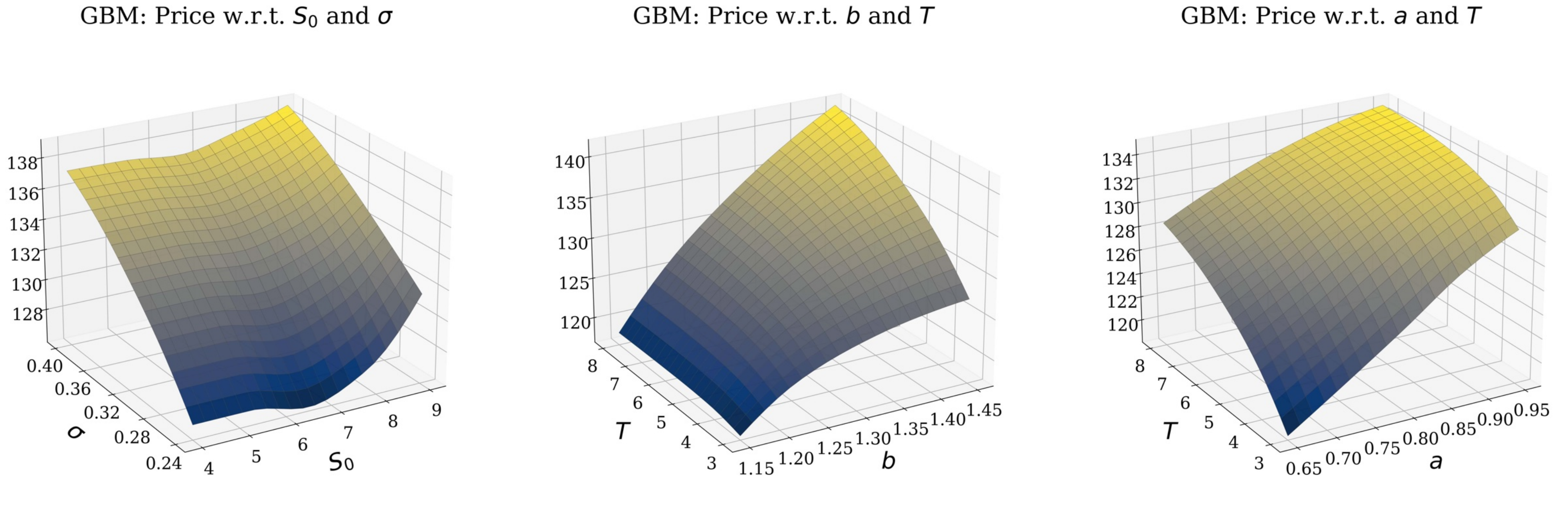}
        \caption{GBM model}
        \label{fig:gbm_sub}
    \end{subfigure}
    
    \vspace{0.4cm}
    
    \begin{subfigure}{1.0\textwidth}
        \centering
        \includegraphics[width=\linewidth]{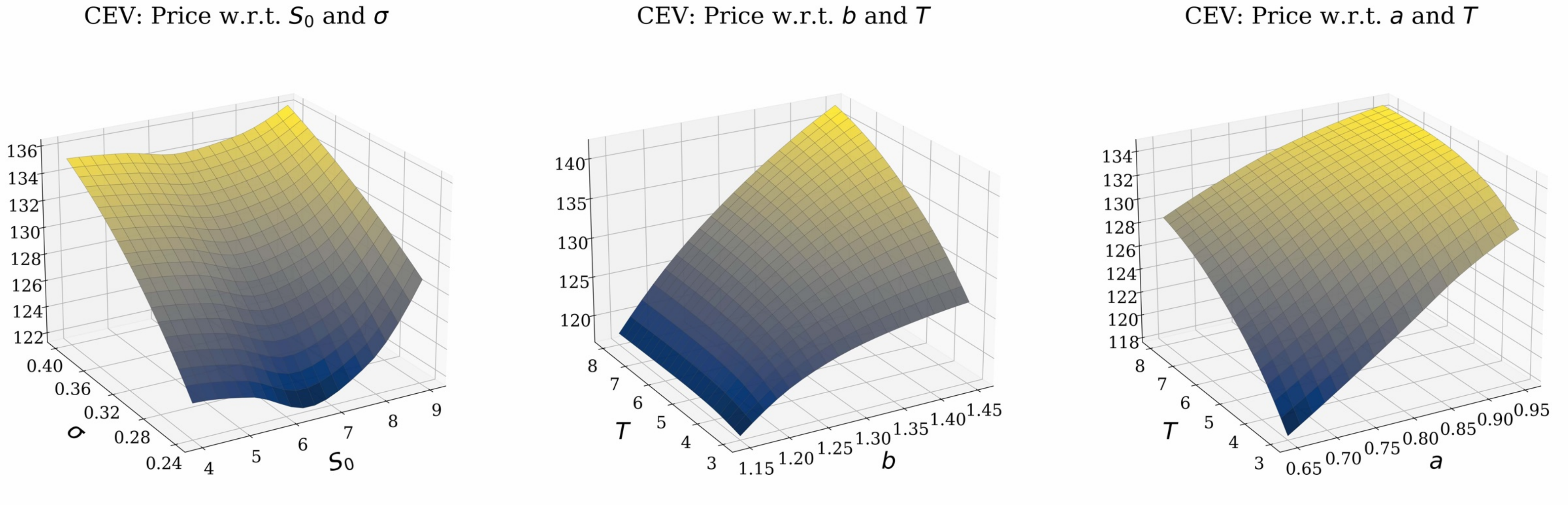}
        \caption{CEV model}
        \label{fig:cev_sub}
    \end{subfigure}
    
    \vspace{0.4cm}
    
    \begin{subfigure}{1.0\textwidth}
        \centering
        \includegraphics[width=\linewidth]{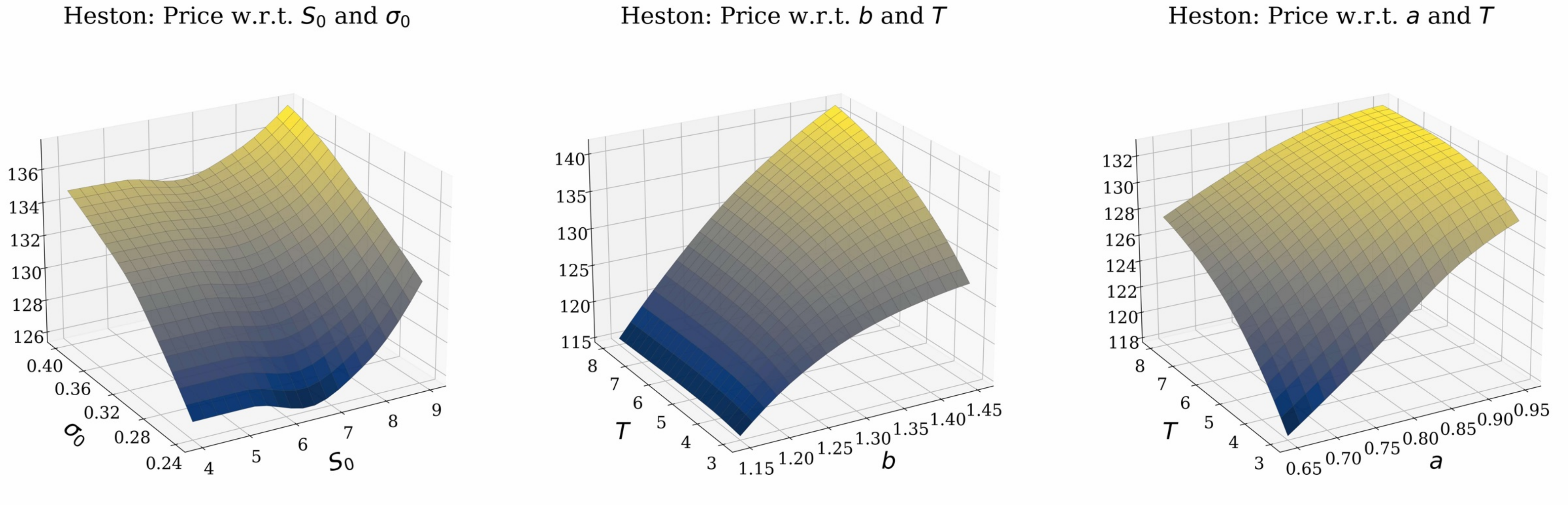}
        \caption{Heston model}
        \label{fig:heston_sub}
    \end{subfigure}
    
    \caption{PPDE-based convertible bond pricing surfaces under different underlying dynamics.}
    \label{fig:pricing_surfaces}
    \begin{minipage}{0.95\textwidth}
\footnotesize
\vspace{0.1cm}
\textit{Notes:} In each panel, the plotted surface is obtained by varying the indicated variables while keeping all other parameters fixed at the baseline values reported in Table~\ref{tab:model_parameters}. For GBM and CEV,
$\sigma$ denotes the volatility parameter. For Heston, $\sigma_0=\sqrt{v_0}$ denotes the initial volatility, and the stock--volatility surface is generated by setting $v_0=\theta=\sigma_0^2$ at each grid point.
\end{minipage}
\end{figure}


{\color{royalblue}
The left-hand panels of Figure~\ref{fig:pricing_surfaces} display similar qualitative patterns across the three models. The convertible bond value responds nonlinearly to the initial stock price $S_0$ and generally increases with volatility. At lower stock prices, the conversion option is unlikely and the fixed-income component dominates the bond value. As $S_0$ rises, the conversion opportunity becomes more valuable and the bond price responds more strongly to the stock price. Higher volatility also increases the probability of favourable stock-price movements, lifting the value of the embedded conversion option. The contractual provisions modify these standard effects. Downward stock-price movements can trigger reset and improve subsequent conversion terms, whereas sufficiently strong upward movements may activate the issuer call and cap the investor's future upside. For the Heston model, the plotted initial volatility is $\sigma_0=\sqrt{v_0}$, with $v_0=\theta=\sigma_0^2$ at each grid point, to make the comparison with GBM and CEV meaningful. Differences across the three models are mainly reflected in the price levels and surface curvature.

The middle panels show that the bond price increases with the call threshold $b>1$ and maturity $T$. A higher value of $b$ makes the call condition more difficult to satisfy because the stock price must reach a higher level relative to the prevailing conversion price before the issuer can redeem the bond. This reduces the likelihood of early redemption and allows bondholders to retain more of the potential upside from future conversion. Conversely, a lower value of $b$ strengthens the negative valuation effect of the call provision by making early redemption more likely. A longer maturity creates more opportunities for favourable stock-price movements and conversion. Its value is nevertheless limited when $b$ is low because the bond is more likely to be called before maturity. As $b$ increases, the call constraint weakens and the positive effect of extending maturity becomes more pronounced.

The right-hand panels show that the bond price increases with the reset threshold ratio $a<1$ and maturity $T$ over the considered parameter range. A higher value of $a$ makes the reset condition easier to satisfy and increases the probability of a downward adjustment in the conversion price. A lower conversion price raises the conversion ratio and improves the value received by investors upon conversion. The reset also lowers the effective call barrier $bH_t$, which may increase the likelihood of early redemption and restrict future upside. The effect of $a$ therefore reflects a trade-off between improved conversion terms and a greater likelihood of a subsequent call. Under the benchmark specification, the improvement in conversion terms dominates and the reset provision raises the bond value. The surfaces flatten at higher values of $a$ and longer maturities, indicating that the marginal value of additional reset opportunities gradually declines and is partly constrained by the issuer's call right.

The pricing surfaces are consistent with the clause comparison in Table~\ref{tab:price2} and the sensitivity analysis in Section~\ref{sec:sensitivity}. The call provision lowers the bond value
by allowing the issuer to terminate the contract and truncate the investor's future upside. A higher call threshold $b$ makes this outcome less likely and therefore raises the bond value. By contrast, the reset provision raises the bond value under the benchmark specification because the benefit of improved conversion terms outweighs the effect of a lower call barrier and the associated increased likelihood of early redemption. A higher reset threshold $a$ makes the reset more likely and therefore increases the bond value over the considered range. Volatility has a positive effect across all three models, while the influence of $S_0$ remains relatively limited around the baseline. The main economic patterns are consistent across GBM, CEV, and Heston, with the choice of volatility dynamics marginally affecting price levels and surface curvature.}

{\color{royalblue}
\subsection{Benchmark comparison}
\label{sec:benchmark}
We now assess the numerical accuracy of the proposed neural network scheme by comparing its price outputs with 
the LSMC method of \citet{longstaff2001valuing}. The comparison is conducted under the GBM model with the full reset and call provisions enabled. Both methods share identical simulated stock price paths, contractual state updates and time discretisation. The only difference lies in how the one-step conditional expectation is approximated. The LSMC method uses all polynomial basis terms up to total degree two and is fitted separately over the continuation, reset, and call regions. Since no closed-form analytical solution exists for this path-dependent contract, the LSMC estimate serves as a numerical benchmark rather than an exact reference value.

The benchmark uses 104 time steps per year, corresponding to $N=624$ time steps over the six-year maturity, together with 12000 training paths and 4000 independent test paths. Under this discretisation, the contractual 15-out-of-30 trading-day trigger rule is represented by a 6-out-of-12 grid-point rule to preserve the same trigger ratio. The DNN consists of three hidden layers with 64 neurons per layer with the ReLU activation function. The inputs and regression targets are standardised at every backward time step. To mitigate accumulation of extrapolation errors through the backward recursion, the network predictions at each time step are projected onto the range of the corresponding one-step training targets.

In addition to the baseline five-dimensional input, we test alternative inputs augmented with 2, 4, and 6 normalised lagged stock prices appended to the input vector. These variants are used to mimic partial rolling-window path-history information and to examine whether the agreement between the two methods persists as state dimension grows. For each specification, we report the absolute difference between the price from the DNN approach $V_0^{\mathrm{DNN}}$ and the price from LSMC $V_0^{\mathrm{LSMC}}$ 
\[
D_{\mathrm{abs}}
=
\left|V_0^{\mathrm{DNN}}-V_0^{\mathrm{LSMC}}\right|
\]
and the relative difference
\[
D_{\mathrm{rel}}
=
\frac{
\left|V_0^{\mathrm{DNN}}-V_0^{\mathrm{LSMC}}\right|
}{
\left|V_0^{\mathrm{LSMC}}\right|
}.
\]

\begin{table}[H]
\centering
\captionsetup{labelfont={color=royalblue},textfont={color=royalblue}}
\color{royalblue}
\small
\caption{\textcolor{royalblue}{Comparison of DNN and LSMC price estimates}}
\label{tab:dnn_lsmc_benchmark}
\begin{tabular*}{\textwidth}{@{\extracolsep{\fill}}cccccc}
\toprule
History lags & Input dimension & DNN price & LSMC price
& $D_{\mathrm{abs}}$ & $D_{\mathrm{rel}}$ (\%) \\
\midrule
0 & 5  & 129.555 & 129.760 & 0.206 & 0.158 \\
2 & 7  & 129.405 & 129.788 & 0.383 & 0.295 \\
4 & 9  & 130.354 & 129.786 & 0.568 & 0.437 \\
6 & 11 & 130.336 & 129.781 & 0.554 & 0.427 \\
\bottomrule
\end{tabular*}
\begin{minipage}{0.95\textwidth}
\footnotesize
\vspace{0.1cm}
\textit{Notes:} Both methods use identical simulated paths and contractual event rules. The baseline specification includes no additional lagged stock prices. The subsequent rows append the stated number of normalised lagged stock prices to the inputs for both methods. LSMC uses a quadratic polynomial basis. The absolute and relative differences measure cross-method agreement rather than errors relative to an exact analytical solution.
\end{minipage}
\end{table}

Table~\ref{tab:dnn_lsmc_benchmark} demonstrates close alignment between DNN and LSMC price outputs. Under the baseline five-dimensional input setting, the DNN price is 129.555, compared with 129.760 for LSMC, corresponding to a relative difference of approximately 0.158\%. The relative difference remains below 0.5\% when lagged stock price features are added. The LSMC estimates are stable across all four state specifications, ranging from 129.760 to 129.788. These results confirm that the estimated bond prices are not a numerical artefact of the neural network approximation.

\begin{table}[H]
\centering
\captionsetup{labelfont={color=royalblue},textfont={color=royalblue}}
\color{royalblue}
\small
\caption{\textcolor{royalblue}{Growth of the LSMC basis and DNN parameterisation}}
\label{tab:dnn_lsmc_dimension}
\begin{tabular*}{\textwidth}{@{\extracolsep{\fill}}c c c c c}
\toprule
History lags & Input dimension
& Quadratic LSMC & Cubic LSMC
& DNN parameters \\
\midrule
0 & 5  & 21 (1.000) & 56 (1.000)  & 8769 (1.000) \\
2 & 7  & 36 (1.714) & 120 (2.143) & 8897 (1.015) \\
4 & 9  & 55 (2.619) & 220 (3.929) & 9025 (1.029) \\
6 & 11 & 78 (3.714) & 364 (6.500) & 9153 (1.044) \\
\bottomrule
\end{tabular*}
\begin{minipage}{0.95\textwidth}
\footnotesize
\vspace{0.1cm}
\textit{Notes:} The quadratic and cubic LSMC columns include all monomials up to total degrees two and three, giving $\binom{d+2}{2}$ and $\binom{d+3}{3}$ basis terms, respectively. For the DNN with three hidden layers of 64 neurons, the total number of trainable parameters is $64d+8{,}449$. Values in parentheses are growth factors relative to the corresponding five-dimensional specification. The benchmark prices in Table~\ref{tab:dnn_lsmc_benchmark} are obtained using the quadratic basis.
\end{minipage}
\end{table}

The quadratic and cubic LSMC include all monomials up to total degrees of two and three. For the three-hidden-layer DNN with 64 neurons per layer, the total number of trainable parameters is equal to $64d+8449$. Table~\ref{tab:dnn_lsmc_dimension} shows how increasing input dimensionality impacts the two approximations. As the input dimension rises from 5 to 11, the number of quadratic LSMC basis terms grows from 21 to 78, while a cubic specification would expand from 56 to 364 basis functions. Over the same range, the number of DNN parameters incurs a modest relative increase of about 4.4\% from 8769 to 9153. 

For the quadratic-basis LSMC implementation used in the benchmark, this difference is also reflected in the observed computation times: the DNN computation time increases by approximately 5.9\%, whereas the LSMC runtime increases by approximately 261\%. Although LSMC remains faster in these low-dimensional test cases, its computational cost escalates far more rapidly as rolling-history state variables are incorporated. Accordingly, the comparative strength of the DNN is its more moderate dimensional scaling rather than in a lower absolute runtime for low-dimensional problems.

A further advantage of the DNN approximation is that it does not require selection of a polynomial basis. This is well-suited to present convertible-bond contract, where continuation values depend nonlinearly on stock price, conversion price, trigger counts, and rolling-history variables, with interactions induced by the reset and call provisions. Representing such interactions in LSMC requires additional higher-order and cross-product basis terms, whereas neural networks learn these interactions endogenously. Once trained, the network can be evaluated over large sets of state configurations without refitting regressions, which facilitates the construction of pricing surfaces and repeated scenario analyses. 

Finally, its differentiable representation further permits gradient-based sensitivities calculations with respect to continuous state variables via automatic differentiation, subject to the usual caution near contractual boundaries where the value function may not be smooth.
}

{\color{royalblue}\subsection{Time-grid refinement and empirical stability}}
\label{sec:grid_sensitivity}

{\color{royalblue}
To examine the sensitivity of the estimated bond price to the time discretisation, we conduct a grid-refinement experiment under the benchmark GBM dynamics with $\sigma=0.30$, retaining both the call and reset provisions. We consider 52, 104, 208, and 416 time steps per year, corresponding to the total step counts $N=312,~624,~1248,$ and $2496$ over the six-year maturity. The contractual 15-out-of-30 trading-day trigger rule is mapped to each numerical grid, giving 3-out-of-6, 6-out-of-12, 12-out-of-25, and 25-out-of-50 trigger rules, respectively.

The DNN 
is trained using 12,000 paths, with 4,000 independent paths used for evaluation. At each backward step, the inputs and regression targets are normalised, and training is performed for at most 20 epochs with validation-based early stopping. The five current-state variables are augmented by the complete rolling-window stock price history, resulting in input dimensions of 11, 17, 30, and 55 across the four grid resolutions.

The Brownian increments on the coarser grids are obtained by aggregating increments simulated on the finest grid to reduce the sampling noise when comparing adjacent grids. Thus, each replicate shares common random numbers across all four resolutions. Three pairs of simulation and network seeds are used: $(42,1042)$, $(123,1123)$, and $(2026,3026)$. 
Let $V_{0,j}^{N}$ denote the bond price computed on an $N$-step grid in replicate $j$, and let $\overline V_0^{N}$ be the sample mean over the three replicates. We report the relative change between successive mean price estimates,
\[
R_N=
\frac{\left|\overline V_0^{2N}-\overline V_0^{N}\right|}
{\left|\overline V_0^{2N}\right|},
\]
and the mean paired absolute change
\[
A_N=\frac{1}{3}\sum_{j=1}^{3}
\left|V_{0,j}^{2N}-V_{0,j}^{N}\right|.
\]
The metric $A_N$ prevents changes with opposite signs across replicates from cancelling when the mean prices are compared.

\begin{table}[H]
\centering
\captionsetup{labelfont={color=royalblue},textfont={color=royalblue}}
\small
\color{royalblue}
\caption{{Time-grid refinement under the benchmark GBM specification}}
\label{tab:grid_convergence}
\begin{tabular*}{\textwidth}{@{\extracolsep{\fill}}cccccccc}
\toprule
Steps/year & Total $N$ & Dimension & Mean price & Seed SD & $R_N$ (\%) & $A_N$ & Runtime (s) \\
\midrule
52  & 312  & 11 & 128.6589 & 0.3670 & --    & --     & 191  \\
104 & 624  & 17 & 129.0421 & 0.6104 & 0.297 & 0.3832 & 387  \\
208 & 1248 & 30 & 132.1162 & 0.8067 & 2.327 & 3.0741 & 785 \\
416 & 2496 & 55 & 132.0315 & 3.5623 & 0.064 & 2.3908 & 1629 \\
\bottomrule
\end{tabular*}
\begin{minipage}{0.95\textwidth}
\footnotesize
\vspace{0.1cm}
\textit{Notes:} Mean price and Seed SD are computed from the three paired replicates. $R_N$ is the relative change between adjacent mean prices, and $A_N$ is the mean paired absolute change. Runtime is the mean wall-clock time per replicate in the same CPU environment. The computations are performed using Python~3.12.4 and PyTorch~2.9.0 in CPU-only mode.
\end{minipage}
\end{table}

The estimates on the two coarsest grids are close: increasing the resolution from 52 to 104 steps per year changes the mean price by 0.297\%, with a mean paired absolute change of 0.3832. The mean prices on the two finest grids also differ by only 0.064\%. This small difference should, however, be considered together with the paired results. The mean paired absolute change between 208 and 416 steps per year is 2.3908, and the seed standard deviation rises to 3.5623 on the finest grid. Hence, the small observed difference between the two mean prices partly reflects changes with opposite signs across replicates and does not provide a reliable estimate of an empirical convergence order.

The results also show the computational cost of increasing both the temporal and path-state resolutions. The mean runtime rises from 191 seconds at 52 steps per year to 1629 seconds at 416 steps per year, while the input dimension increases from 11 to 55. The 104-step grid requires less than one quarter of the finest-grid runtime and produces a price close to the benchmark estimates in Table~\ref{tab:dnn_lsmc_benchmark}. It therefore provides a reasonable discretisation default choice for the subsequent numerical comparisons. The grid refinement study documents the sensitivity of the computed price to the selected grid and the remaining sampling and training variation.}

\subsection{Effects of reset and call provisions}

Table~\ref{tab:price2} quantifies the marginal effects of the call and reset provisions on the convertible bond value. We compare three contract structures: a straight bond without embedded options,
a bond with only a call provision,
and a bond with both call and downward reset provisions. We address three questions: the value of a plain bond without any embedded provisions, the change in value when a call provision is added, and further change when a reset provision is also included.

{\color{royalblue}
Two robust and consistent patterns emerge across all model specifications. First, introducing the call provision reduces the bond value. Second, incorporating the reset provision on top of the call provision increases the bond value. These relationships hold under GBM, CEV, and Heston dynamics, suggesting that the signs of the estimated clause effects are driven by the contractual structure rather than the choice of underlying dynamics.

Quantitatively, the mean call effect ranges from approximately $-3.8$ to $-4.4$. This is consistent with the call provision imposing an upper bound on the continuation value: when the contractual trigger is satisfied, the issuer may redeem the bond and thereby truncate part of the investor's future upside potential. The magnitude is similar across the three stock-dynamics specifications, and the paired-seed standard deviations are below $0.5$.
The mean reset effect is around $+18.0$ across the three models. A reset lowers the effective conversion price $H_t$, increases the conversion ratio and hence the value delivered to investors in high-stock-price states. Note that resetting also reduces the effective call threshold $bH_t$, which may then accelerate redemption and limit subsequent upside. However, under the present parameter configuration, the improvement in conversion terms dominates this opposing call-threshold channel. This yields $V_{\text{call+reset}}>V_{\text{call-only}}$ as reported in Table~\ref{tab:price2}. The paired-seed standard deviations range from $0.62$ to $0.88$, indicating that this ordering is not driven by simulation variation.
}

\begin{table}[H]
\centering
\captionsetup{labelfont={color=royalblue},textfont={color=royalblue}}
\caption{Convertible bond prices under different contract specifications}
\color{royalblue}
\small
\begin{tabular*}{\textwidth}{@{\extracolsep{\fill}}llllll}
\toprule
Model & Plain & Call-only & Call and reset & Call effect & Reset effect \\
\midrule
GBM    & \makecell{115.48\\[-1pt]{\scriptsize (0.49)}} & \makecell{111.11\\[-1pt]{\scriptsize (0.09)}} & \makecell{129.04\\[-1pt]{\scriptsize (0.61)}} & \makecell{-4.36\\[-1pt]{\scriptsize (0.46)}} & \makecell{17.93\\[-1pt]{\scriptsize (0.62)}} \\
CEV    & \makecell{114.56\\[-1pt]{\scriptsize (0.42)}} & \makecell{110.76\\[-1pt]{\scriptsize (0.01)}} & \makecell{128.81\\[-1pt]{\scriptsize (0.87)}} & \makecell{-3.80\\[-1pt]{\scriptsize (0.41)}} & \makecell{18.04\\[-1pt]{\scriptsize (0.88)}} \\
Heston & \makecell{114.73\\[-1pt]{\scriptsize (0.36)}} & \makecell{110.38\\[-1pt]{\scriptsize (0.12)}} & \makecell{128.35\\[-1pt]{\scriptsize (0.63)}} & \makecell{-4.35\\[-1pt]{\scriptsize (0.28)}} & \makecell{17.97\\[-1pt]{\scriptsize (0.67)}} \\
\bottomrule
\end{tabular*}
\label{tab:price2}
\begin{minipage}{0.95\textwidth}
\footnotesize
\vspace{0.1cm}
\textit{Notes:} Entries report the mean values across three seeds, with the standard deviation in parentheses. Plain-contract values are computed by standard Monte Carlo, call-only values are obtained using regime-wise quadratic LSMC, and call-and-reset values are obtained using the backward DNN with complete rolling-window history. The GBM call-and-reset estimate is the 104-step result reported in Table~\ref{tab:grid_convergence}. The call effect is measured as $V_{\text{call}} - V_{\text{plain}}$, while the reset effect is measured as $V_{\text{call+reset}} - V_{\text{call}}$. The effect standard deviations are computed from paired seed-level differences.
\end{minipage}
\end{table}

\subsection{Path-dependent state structure and decision regions}
To better understand how path-dependent features affect the optimal conversion decision, we examine the joint behaviour of the key state variables. Specifically, we consider the interaction between the call counter $N^{\mathrm{call}}$, which captures short-term trigger conditions, and the cumulative number of reset events $N^{\mathrm{reset}}$, which reflects longer-term path history. Since contractual features play a more important role in valuation than the choice of the underlying stock dynamics, we adopt the GBM model for all state-space illustrations in Figure~\ref{fig:evo}. This figure displays cross-sectional slices of simulated paths at several selected time points. Each point represents a state ($N^{\mathrm{call}}$, cumulative resets) together with the corresponding optimal decision. This representation makes it possible to track how the decision regions evolve over time and how they are shaped by the interaction between call and reset dynamics.

\textcolor{royalblue}{We use a discrete-time approximation with 252 time steps per year, which preserves the contractual 15-out-of-30 trigger rule on the monitoring grid. The continuation value functions are estimated using $12{,}000$ simulated training paths. For the state-space visualisations, $5{,}000$ non-terminated state paths are generated to construct cross-sectional slices at selected time points. Trigger activation is recorded during the forward simulation, whereas optimal exercise is determined through the backward decision rule. All parameters are fixed at the baseline values reported in Table~\ref{tab:model_parameters}.}

\begin{figure}[H]
    \centering
    \includegraphics[width=0.99\linewidth]{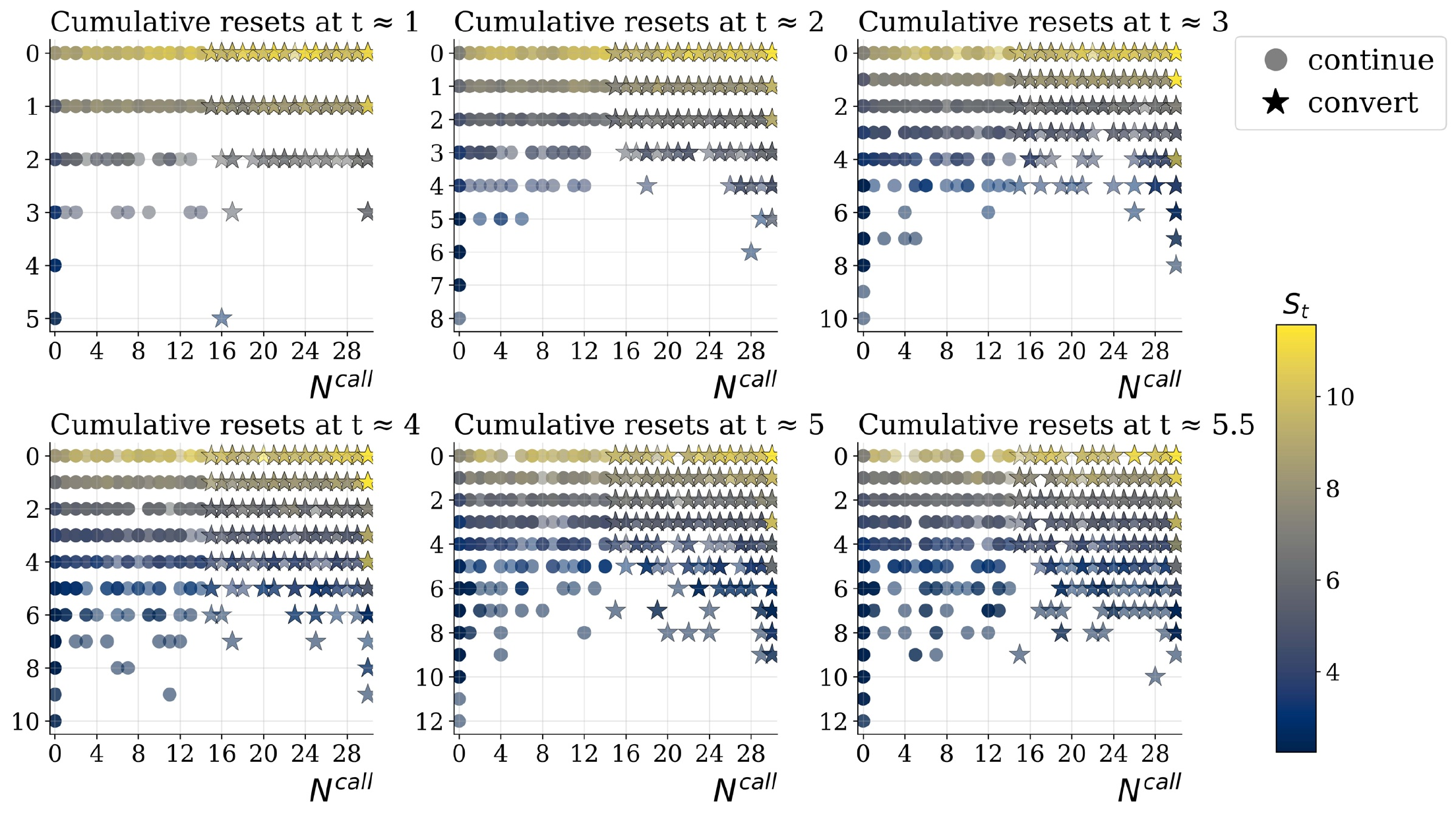}
    \caption{\textcolor{royalblue}{State-space evolution of the call counter, cumulative resets, and optimal decisions.}}
    \label{fig:evo}
     \begin{minipage}{0.95\textwidth}
\footnotesize
\vspace{0.1cm}
\textcolor{royalblue}{\textit{Notes:} Each panel presents a cross-sectional slice of $5{,}000$ unstopped simulated state paths at $t=1,2,3,4,5,$ and $5.5$ years. The horizontal axis reports the call counter $N^{\mathrm{call}}$, while the vertical axis reports the cumulative number of reset events. Colours indicate the contemporaneous stock price $S_t$, with the colour scale clipped at the 5th and 90th percentiles to prevent a small number of extreme unstopped paths from compressing the visible colour range. Circles denote continuation and stars denote conversion. Under the baseline parameters, no issuer-redemption states occur in these slices. Call-inactive states, $N^{\mathrm{call}}<15$, remain in the continuation region under the maintained modelling assumption, whereas call-active states are evaluated using the fitted pre-call continuation value and the call price. The maturity date $t=6$ is omitted because no further decision is required at expiration.}
\end{minipage}
\end{figure}

Figure~\ref{fig:evo} shows a clear negative relationship between the stock price $S_t$ and the cumulative number of reset events. Paths associated with persistently weak stock prices tend to accumulate more reset events over time. Although each reset lowers the conversion price $H_t$ and therefore reduces the future reset threshold, sustained low stock price dynamics may still generate repeated resets during the life of the bond.
Moreover, repeated resets can indirectly increase the likelihood of satisfying call-related conversion conditions. Since a lower conversion price increases the conversion value ratio $S_t/H_t$, the convertible bond becomes more deeply in the money after multiple resets, making conversion and issuer call triggers more likely in later periods.

In a standard convertible bond without path-dependent features, the optimal conversion decision is mainly driven by the comparison between the stock price and the conversion price. Conversion typically occurs only when the stock price is sufficiently high relative to the conversion price, reflecting the intrinsic value of the embedded option. Moreover, due to the time value of the conversion option, even when the stock price exceeds the conversion price, immediate conversion is often suboptimal, and investors tend to delay exercise.
By contrast, the presence of path-dependent contractual features changes this behaviour. In our framework, repeated reset events gradually reduce the effective conversion price, while the call provision introduces an additional constraint once the trigger condition is met. As a result, conversion may occur even when the contemporaneous stock price is not particularly high in absolute terms. This does not reflect a change in the intrinsic payoff structure, but rather the cumulative effect of past price movements and contractual triggers.

Overall, when both reset and call provisions are present, the optimal conversion policy can no longer be characterised by a simple threshold in the stock price. Instead, it depends on the joint configuration of multiple path-dependent state variables. This is also evident in the state-space visualisation, where similar levels of the stock price can be associated with different optimal decisions depending on the accumulated reset history and the degree of proximity to the call trigger.

\subsection{Sensitivity analysis}\label{sec:sensitivity}

\textcolor{royalblue}{We now examine the sensitivity of convertible bond prices to the initial stock price, the volatility parameter, the reset and the call threshold ratio. The analysis is performed under the GBM, CEV, and Heston specifications using the same full call-and-reset benchmark and the same three paired seeds as in the preceding numerical experiments.}

\begin{table}[H]
\centering
\captionsetup{labelfont={color=royalblue},textfont={color=royalblue}}
\color{royalblue}
\small
\caption{Sensitivity analysis under the GBM specification}
\begin{tabular}{c | c c c | c c c | c c c | c c c}
\toprule
Pert. 
& $S_0$ & Price & $\Delta(\%)$ 
& $\sigma$ & Price & $\Delta(\%)$ 
& $a$ & Price & $\Delta(\%)$ 
& $b$ & Price & $\Delta(\%)$ \\
\cmidrule(lr){1-1} \cmidrule(lr){2-4} \cmidrule(lr){5-7} \cmidrule(lr){8-10} \cmidrule(lr){11-13}
-20\% 
& 5.12 & 129.45 & \hspace{0.17em}0.02
& 0.24 & 126.39 & -2.35
& 0.64 & 126.73 & -2.09
& 1.04 & 104.57 & -19.21 \\

-10\% 
& 5.76 & 128.96 & -0.37
& 0.27 & 127.95 & -1.15
& 0.72 & 128.13 & -1.00
& 1.17 & 119.56 & -7.63 \\

\quad 0\% 
& 6.40 & 129.43 & \enskip 0.00
& 0.30 & 129.43 & \enskip 0.00
& 0.80 & 129.43 & \enskip 0.00
& 1.30 & 129.43 & \enskip 0.00 \\

+10\% 
& 7.04 & 129.51 & \enskip 0.06
& 0.33 & 131.27 & \enskip 1.42
& 0.88 & 131.11 & \enskip 1.29
& 1.43 & 138.01 & \enskip 6.63 \\

+20\% 
& 7.68 & 130.35 & \enskip 0.71
& 0.36 & 132.92 & \enskip 2.70
& 0.96 & 133.35 & \enskip 3.03
& 1.56 & 137.99 & \enskip 6.61 \\

\bottomrule
\end{tabular}
\begin{minipage}{0.95\textwidth}
\footnotesize
\vspace{0.1cm}
\textit{Notes:} Prices report the mean across three paired seeds. Each parameter is perturbed around its baseline value while all other parameters in Table~\ref{tab:cb_params} are held fixed. $\Delta(\%)$ is first computed relative to the corresponding seed-level baseline and then averaged across seeds.
\end{minipage}
\label{tab:gbm_sensitivity}
\end{table}

\begin{table}[H]
\centering
\captionsetup{labelfont={color=royalblue},textfont={color=royalblue}}
\color{royalblue}
\small
\caption{Sensitivity analysis under the CEV specification}
\begin{tabular}{ c | c c c | c c c | c c c | c c c}
\toprule
Pert. 
& $S_0$ & Price & $\Delta(\%)$ 
& $\sigma$ & Price & $\Delta(\%)$ 
& $a$ & Price & $\Delta(\%)$ 
& $b$ & Price & $\Delta(\%)$ \\
\cmidrule(lr){1-1}\cmidrule(lr){2-4} \cmidrule(lr){5-7} \cmidrule(lr){8-10} \cmidrule(lr){11-13}

-20\% 
& 5.12 & 130.17 & \enskip 0.19
& 0.28 & 126.37 & -2.73
& 0.64 & 126.92 & -2.30
& 1.04 & 104.45 & -19.60 \\

-10\% 
& 5.76 & 129.85 & -0.06
& 0.32 & 127.98 & -1.49
& 0.72 & 128.13 & -1.38
& 1.17 & 119.02 & -8.38 \\

\quad  0\% 
& 6.40 & 129.92 & \enskip 0.00
& 0.35 & 129.92 & \enskip 0.00
& 0.80 & 129.92 & \enskip 0.00
& 1.30 & 129.92 & \enskip 0.00 \\

 +10\% 
& 7.04 & 130.17 & \enskip 0.20
& 0.39 & 131.39 & \enskip 1.14
& 0.88 & 131.74 & \enskip 1.40
& 1.43 & 138.80 & \enskip 6.84 \\

+20\% 
& 7.68 & 130.76 & \enskip 0.65
& 0.42 & 133.76 & \enskip 2.96
& 0.96 & 133.15 & \enskip 2.49
& 1.56 & 140.18 & \enskip 7.91 \\

\bottomrule
\end{tabular}

\begin{minipage}{0.95\textwidth}
\footnotesize
\vspace{0.1cm}
\textit{Notes:} Prices report the mean across three paired seeds. Each parameter is perturbed around its baseline value while all other parameters in Table~\ref{tab:cb_params} are held fixed. $\Delta(\%)$ is first computed relative to the corresponding seed-level baseline and then averaged across seeds.
\end{minipage}
\label{tab:cev_sensitivity}
\end{table}
\begin{table}[H]
\centering
\captionsetup{labelfont={color=royalblue},textfont={color=royalblue}}
\color{royalblue}
\small
\caption{Sensitivity analysis under the Heston specification}
\begin{tabular}{c | c c c | c c c | c c c | c c c}
\toprule
Pert.
& $S_0$ & Price & $\Delta(\%)$
& $\sigma_0$ & Price & $\Delta(\%)$
& $a$ & Price & $\Delta(\%)$
& $b$ & Price & $\Delta(\%)$ \\
\cmidrule(lr){1-1}
\cmidrule(lr){2-4}
\cmidrule(lr){5-7}
\cmidrule(lr){8-10}
\cmidrule(lr){11-13}

$-20\%$
& 5.12 & 128.59 & \enskip 0.34
& 0.24 & 126.66 & -1.17
& 0.64 & 126.76 & -1.09
& 1.04 & 104.27 & -18.64 \\

$-10\%$
& 5.76 & 128.55 & \enskip 0.32
& 0.27 & 127.37 & -0.61
& 0.72 & 127.36 & -0.62
& 1.17 & 115.96 & -9.51 \\

\quad $0\%$
& 6.40 & 128.16 & \enskip 0.00
& 0.30 & 128.16 & \enskip 0.00
& 0.80 & 128.16 & \enskip 0.00
& 1.30 & 128.16 & \enskip 0.00 \\

$+10\%$
& 7.04 & 128.43 & \enskip 0.22
& 0.33 & 129.08 & \enskip 0.72
& 0.88 & 129.31 & \enskip 0.91
& 1.43 & 139.29 & \enskip 8.69 \\

$+20\%$
& 7.68 & 129.02 & \enskip 0.68
& 0.36 & 130.14 & \enskip 1.56
& 0.96 & 130.64 & \enskip 1.94
& 1.56 & 144.69 & 12.93 \\

\bottomrule
\end{tabular}

\begin{minipage}{0.95\textwidth}
\footnotesize
\vspace{0.1cm}
\textit{Notes:} Prices report the mean across three paired seeds.
Each sensitivity factor is perturbed around its baseline value.
$\Delta(\%)$ is first computed relative to the corresponding seed-level baseline and then averaged across seeds. In the second sensitivity block, the initial volatility $\sigma_0=\sqrt{v_0}$ is perturbed around its baseline value. At each perturbation level, the Heston variance parameters are set to $v_0=\theta=\sigma_0^2$.
\end{minipage}

\label{tab:heston_sensitivity}
\end{table}

\textcolor{royalblue}{Tables~\ref{tab:gbm_sensitivity}--\ref{tab:heston_sensitivity} report the mean price and the mean paired percentage change when each sensitivity factor is perturbed by $-20\%$, $-10\%$, $+10\%$, or $+20\%$ around its baseline value. The baseline contractual and model parameters are reported in Tables~\ref{tab:cb_params} and~\ref{tab:model_parameters}. For the Heston specification, the initial volatility $\sigma_0=\sqrt{v_0}$ is perturbed by the same percentages, with $v_0=\theta=\sigma_0^2$ at each perturbation level.
The initial stock price $S_0$ has only a limited effect under the reset-enabled contract. Across the three models, changing $S_0$ by $\pm20\%$ alters the estimated bond price by less than 1\%. This weak sensitivity reflects the endogenous adjustment of the conversion price after reset activation. Movements in the stock-price level are partly offset by changes in the effective conversion terms. Consequently, the full-contract value is less sensitive to the absolute level of $S_0$ than a standard convertible bond with a fixed conversion price.}

\textcolor{royalblue}{Volatility has a positive effect across all three models. A 20\% reduction in $\sigma$ lowers the estimated value by approximately 2.35\% under GBM and 2.73\% under CEV, while a 20\% increase raises it by approximately 2.70\% and 2.96\%, respectively. Under the Heston specification, reducing the initial volatility $\sigma_0$ by 20\% lowers the estimated value by approximately 1.17\%, while increasing it by 20\% raises the value by approximately 1.56\%. The Heston results therefore display the same positive volatility effect as the GBM and CEV results.}

\textcolor{royalblue}{The reset threshold ratio $a$ has a moderate positive effect. Raising $a$ makes the reset condition easier to satisfy and increases the estimated value by 3.03\% under GBM, 2.49\% under CEV, and 1.94\% under Heston at the $+20\%$ perturbation. Conversely, lowering $a$ by 20\% reduces the value by approximately 1.09\%--2.30\% across the three models.}
\textcolor{royalblue}{The bond value is particularly sensitive to the call threshold ratio $b$. Reducing $b$ by 20\% lowers the estimated value by approximately 19\% in all three models because the issuer call becomes substantially easier to activate. Increasing $b$ delays call activation and preserves more upside for bondholders. At the $+20\%$ perturbation, the estimated changes are 6.61\% under GBM, 7.91\% under CEV, and 12.93\% under Heston. The near-flat response between the two highest values of $b$ under GBM suggests that the bond value may approach a plateau once the call threshold becomes sufficiently difficult to reach.}
\textcolor{royalblue}{These sensitivity patterns are broadly consistent across the three stock-price models and agree with the pricing surfaces in Figure~\ref{fig:pricing_surfaces}. Among the local parameter perturbations considered, the bond value is most sensitive to the call threshold $b$. The reset threshold $a$ has a moderate positive effect, while the influence of the initial stock price is relatively limited. This ranking concerns local changes in the threshold parameters and should not be interpreted as a ranking of the contractual provisions themselves. In the contract-feature comparison reported in Table~\ref{tab:price2}, the value effect of introducing the reset provision is larger in absolute magnitude than that of introducing the call provision.}

{\color{royalblue}
\subsection{Automatic-differentiation Greeks}\label{sec:ad_greeks}

\paragraph{Greek definitions and automatic differentiation.}

In addition to the bond price, the trained neural network can provide local sensitivities with respect to the continuous state variables. Let $\widehat V_t$ denote the fitted value function, where $\mathcal P_t$ contains the path-dependent state variables associated with the reset and call provisions. Greeks summarise the local response of a derivative value to changes in its underlying state and are therefore central to hedging and risk management \citep{shirzadi2021optimal}. We focus on Delta and Gamma because the stock price is included as a continuous state variable under all three model specifications, allowing the same differentiation procedure to be applied consistently across models. We define
\begin{equation}\label{eq:greeks}
\Delta_t
=
\frac{\partial \widehat V_t}{\partial S_t},
\qquad
\Gamma_t
=
\frac{\partial^2 \widehat V_t}{\partial S_t^2}.
\end{equation}
These derivatives are computed by automatic differentiation while the trained network parameters and all other state variables are held fixed.
Unlike finite-difference bump-and-revalue calculations, automatic differentiation avoids repeated network evaluations at perturbed values. It therefore provides a convenient way to obtain sensitivities from the trained DNN. For the Greek computation, we fit a local Softplus network to the continuation values produced by the local ReLU network. The Softplus activation function is twice continuously differentiable, enabling the evaluation of both Delta and Gamma. Both the automatic-differentiation estimates and the central finite-difference estimates reported in Table~\ref{tab:ad_greeks_validation} are computed from the same fitted Softplus network.
The resulting quantities should be interpreted as local sensitivities of the neural-network approximation. Near the reset and call conversion boundaries, the underlying bond value may lack smoothness, so the estimated Greeks can exhibit sharp jumps across adjacent states.

\paragraph{Experimental design and finite-difference validation.}

The experiment is conducted under the GBM, CEV, and Heston specifications using the contractual and model parameters reported in Tables~\ref{tab:cb_params} and~\ref{tab:model_parameters}. We use 104 time steps per year, corresponding to 624 steps over the six-year maturity, with 12,000 simulated paths for training and 4,000 independent paths for evaluation. For the Heston specification, 24,000 training paths and 6,000 test paths are used because the current variance introduces an additional state variable. All Greeks are evaluated at $t=3$ years, the midpoint of the six-year maturity. For the controlled Heston stock price curve, the current variance is fixed at $v_t=0.09$, corresponding to an instantaneous volatility of $\sqrt{v_t}=0.30$. The backward pricing recursion uses the same three-hidden-layer ReLU network as in Section~\ref{sec:benchmark}. At the selected observation date, a local ReLU network is first fitted separately within each contractual region. A Softplus network with the same architecture is then fitted to the local ReLU continuation values. This construction retains the pricing approximation while providing the twice differentiable representation required for Gamma.

Delta and Gamma are evaluated over $S_t\in[3.5,10]$, with all remaining continuous and path-dependent state variables held at their reference values. The automatic-differentiation estimates are checked against central finite-difference estimates obtained from the same fitted Softplus network using a stock price perturbation $\varepsilon_S=0.10$. Since both calculations are applied to the same smooth approximation, their comparison tests the derivative accuracy rather than the pricing performance of the network. We separately assess the local approximation using the mean absolute error (MAE), root mean squared error (RMSE), and the coefficient of determination $R^2$ between the local ReLU predictions and the corresponding one-step dynamic-programming targets.

{\color{royalblue}
\begin{table}[H]
\color{royalblue}
\captionsetup{labelfont={color=royalblue},textfont={color=royalblue}}
\centering
\small
\caption{Local approximation quality and validation of automatic-differentiation Greeks}
\label{tab:ad_greeks_validation}
\textit{Panel A: Local approximation quality}\\[2pt]
\begin{tabular*}{\textwidth}{@{\extracolsep{\fill}}lcccc}
\toprule
Model & Relative MAE (\%) & Relative RMSE (\%) & Target $R^2$ & ReLU--Softplus error (\%) \\
\midrule
GBM    & 1.054 & 1.504 & 0.903 & 0.108 \\
CEV    & 0.809 & 1.156 & 0.903 & 0.095 \\
Heston & 0.408 & 0.640 & 0.947 & 0.047 \\
\bottomrule
\end{tabular*}

\vspace{0.35cm}
\textit{Panel B: Agreement between automatic differentiation and finite differences}\\[2pt]
\begin{tabular*}{\textwidth}{@{\extracolsep{\fill}}lcccc}
\toprule
Model
& Mean $|e_{\Delta}|$
& Max $|e_{\Delta}|$
& Mean $|e_{\Gamma}|$
& Max $|e_{\Gamma}|$ \\
\midrule
GBM    & 0.0036 & 0.0142 & 0.0045 & 0.0413 \\
CEV    & 0.0029 & 0.0081 & 0.0030 & 0.0219 \\
Heston & 0.0016 & 0.0068 & 0.0023 & 0.0109 \\
\bottomrule
\end{tabular*}

\begin{minipage}{0.95\textwidth}
\footnotesize
\vspace{0.1cm}
\textit{Notes:} Relative MAE and Relative RMSE are the target MAE and RMSE divided by the corresponding mean fitted bond value and expressed as percentages. Target $R^2$ compares the local ReLU predictions with the corresponding one-step dynamic-programming targets. ReLU--Softplus error is the test-sample MAE between the local ReLU and smooth Softplus approximations, divided by the corresponding mean fitted value and expressed as a percentage. In Panel B, $e_{\Delta}=\Delta_{\mathrm{AD}}-\Delta_{\mathrm{FD}}$ and $e_{\Gamma}=\Gamma_{\mathrm{AD}}-\Gamma_{\mathrm{FD}}$. The mean and maximum absolute errors are calculated over the 81 stock-price grid points in $S_t\in[3.5,10]$ using $\varepsilon_S=0.10$ for the central finite differences.
\end{minipage}
\end{table}
}

Panel A of Table~\ref{tab:ad_greeks_validation} shows that the local ReLU fits explain more than 90\% of the variation in the one-step targets under all three dynamics. Relative to the corresponding mean fitted bond values, the MAE ranges from 0.41\% to 1.05\% and the RMSE from 0.64\% to 1.50\%. The discrepancy between the ReLU and Softplus approximations remains below 0.11\% across all models, indicating that smoothing has only a limited effect on the fitted value. Panel B shows close agreement between automatic differentiation and central finite differences. The mean absolute Delta error $|e_{\Delta}|$ is below 0.004 and the mean absolute Gamma error $|e_{\Gamma}|$ is below 0.005 for all models, even the largest observed errors are bounded by 0.015 for Delta and 0.042 for Gamma. These comparisons validate the derivative calculations on the fitted smooth networks, while the target-fit statistics separately describe the quality of the local value approximation.

\paragraph{Delta and Gamma estimates and contractual effects.}

Delta and Gamma, defined in \eqref{eq:greeks}, describe the local exposure of the bond value to movements in the underlying stock and are therefore relevant for both numerical assessment and hedging. Their automatic-differentiation estimates are validated against the corresponding central finite-difference values in Table~\ref{tab:ad_greeks_validation}. We examine the effect of the reset provision by comparing the full-contract Greeks with those obtained under a no-reset specification that retains the call provision. Both specifications use the same DNN backward recursion and region-specific fitting procedure, so the comparison isolates the reset effect without changing the approximation method.

\begin{figure}[htbp]
    \centering
    \begin{subfigure}[t]{0.49\textwidth}
        \centering
        \includegraphics[width=\textwidth]{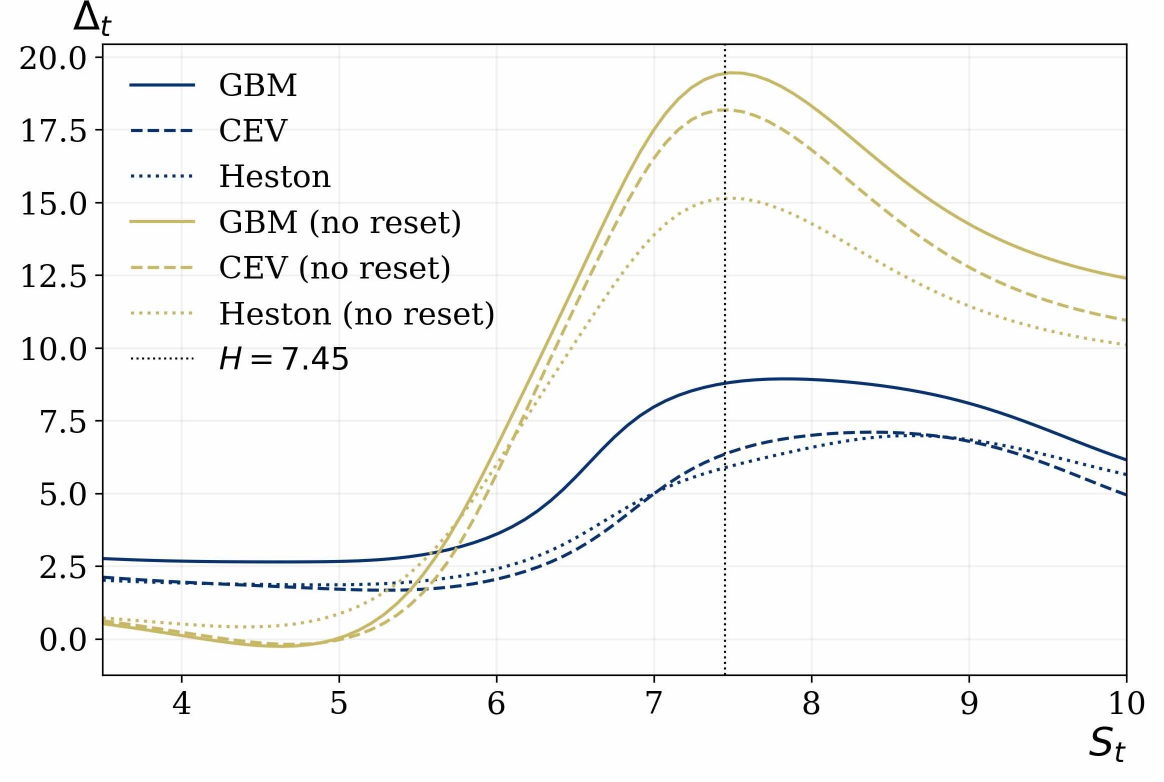}
        \caption{Delta}
        \label{fig:ad_greeks_delta}
    \end{subfigure}
    \hfill
    \begin{subfigure}[t]{0.49\textwidth}
        \centering
        \includegraphics[width=\textwidth]{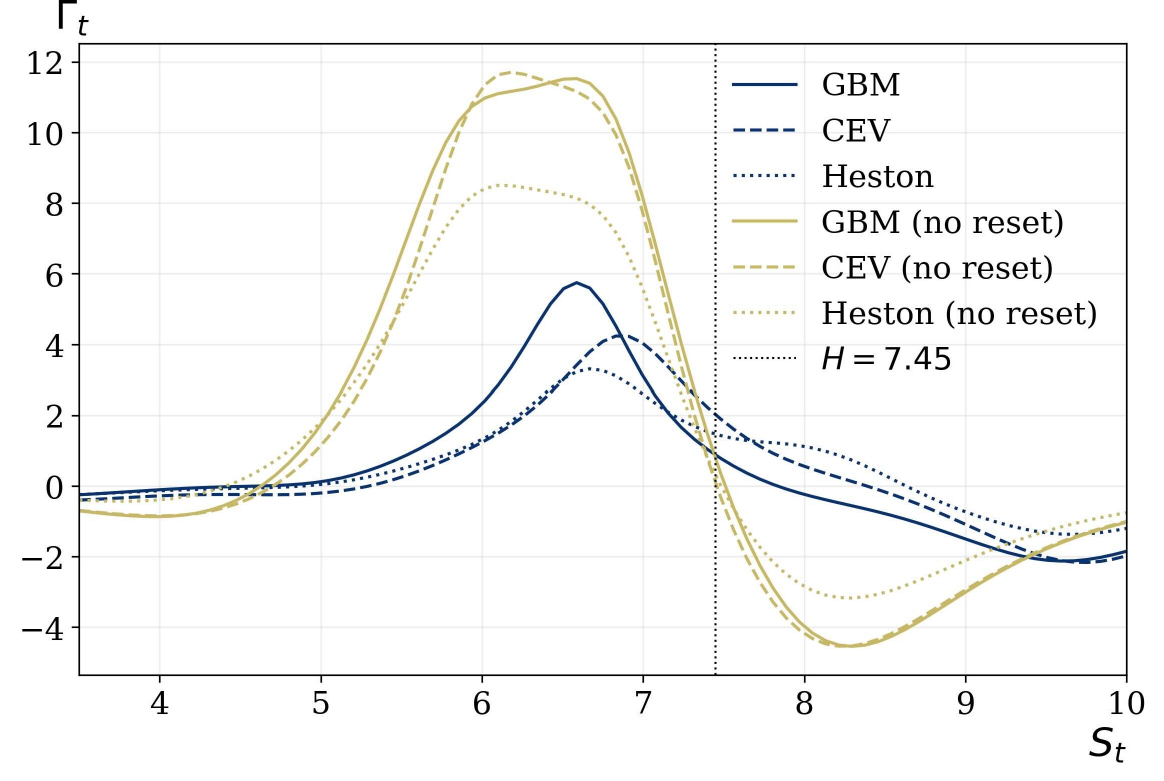}
        \caption{Gamma}
        \label{fig:ad_greeks_gamma}
    \end{subfigure}
    \caption{Automatic-differentiation Greeks under the full-contract and no-reset specifications.}
    \label{fig:ad_greeks_boundaries}
\end{figure}

Figure~\ref{fig:ad_greeks_boundaries} shows the Delta and Gamma curves at
$t=3$ years and illustrates how the reset provision affects their magnitudes and shapes. Additional Delta and Gamma estimates across the GBM, CEV, and Heston models at $t=2$, $3$, and $4$ years are reported in Appendix~\ref{app:greeks_time}. Across all three stock price models, no-reset Delta rises much more steeply as $S_t$ approaches the conversion price $H_t=7.45$ and then declines gradually at higher stock prices. The corresponding Gamma is strongly positive below $H_t$, decreases sharply near the conversion region, and turns negative as $S_t$ rises further. Negative Gamma is not confined to resettable convertible bonds.
Using finite-difference sensitivities for a sample of USD convertible bonds, \citet{zeitsch2024convertible} reports that Gamma be either positive or negative. In a snapshot-reset setting, \citet{yu2008valuation} further shows that both Delta and Gamma can turn negative and that reset discontinuities can produce pronounced variations in these sensitivities.
Although the reset mechanism here is path-dependent rather than a snapshot reset, the observed sign change similarly reflects a change in the curvature of the fitted bond-value function. This interpretation is consistent with the abrupt changes in slope and curvature near contractual boundaries discussed by \citet{lim2026structured}. Compared with the no-reset results, the full-contract Delta curves are substantially lower, while the positive and negative movements in Gamma are less pronounced. The reset provision therefore appears to dampen the local sensitivity of the bond value to the contemporaneous stock price rather than eliminate the curvature change around the conversion region. Because call activation depends on the number of qualifying days within a rolling window \citep{liu2020,ma2020valuation}, however, the sensitivity patterns should be interpreted conditional on the prevailing call state rather than as responses to a single stock price threshold.

}
\section{Conclusion\label{sec:conclusion}}
This paper develops a deep learning-based framework for pricing convertible bonds with downward reset and call provisions in a path-dependent setting. By formulating the valuation problem as a path-dependent PPDE, we explicitly capture the dependence of contractual triggers on the historical trajectory of the underlying stock price. The proposed numerical scheme combines a discrete-time dynamic programming representation with neural network approximation of conditional expectations, providing a flexible and scalable approach for high-dimensional problems.

Within this unified framework, we derive model-specific PPDE formulations under GBM, CEV, and Heston dynamics \textcolor{royalblue}{and establish a piecewise viscosity characterisation.} We implement a backward deep learning algorithm to approximate the value functions, \textcolor{royalblue}{and establish $L^2$ convergence of the neural-network approximation to the exact dynamic programming recursion on a fixed time grid. Numerical validation confirms the accuracy and stability of our DNN prices, which closely align with the LSMC benchmarks and exhibit robust behaviour across time-grid refinements and alternative network architectures.} The resulting pricing surfaces exhibit clear and economically intuitive structure, reflecting the hybrid nature of convertible bonds as a combination of debt and embedded equity options.

\textcolor{royalblue}{Our analysis highlights several consistent and meaningful findings. First, under the maintained reset specification, the bond value shows limited sensitivity to perturbations in the initial stock price near the baseline, despite exhibiting a nonlinear response over a wider stock-price range. Second, volatility has a positive effect across all three models.} Third, the call provision \textcolor{royalblue}{reduces the bond value} by truncating upside potential. Finally, \textcolor{royalblue}{the reset provision increases the bond value under the benchmark specification because the improvement in conversion terms outweighs the effect of the associated reduction in the effective call threshold.} These findings are \textcolor{royalblue}{broadly} consistent across the three dynamics, suggesting that contractual design \textcolor{royalblue}{has a greater influence on valuation} than the specific choice of diffusion process. \textcolor{royalblue}{In addition, Delta and Gamma computed by automatic differentiation from smooth local neural-network approximations closely match their central finite-difference counterparts and exhibit pronounced variation near contractual boundaries.}

Several directions for future research remain open. \textcolor{royalblue}{Extending the framework to incorporate credit risk could provide a more realistic treatment of default and recovery. Incorporating dividend payments and corporate-action-based anti-dilution adjustments would allow their effects on stock-price dynamics, conversion incentives, and contractual conversion-price changes to be analysed jointly, as discussed in Appendix~\ref{app:antidilution}. The reset mapping could also be extended to incorporate more realistic constraints, including the approval process and conversion-price floors tied to trading-price averages, the most recently audited net asset value per share, and the par value per share. Further gains may be achieved by integrating the proposed methodology with more advanced machine learning architectures to capture complex nonlinear interactions in high-dimensional pricing environments.}

\section*{Acknowledgements}

Qinwen Zhu acknowledges the support of the Artificial Intelligence Initiative to Promote the Reform of Scientific Research Paradigms and Empower Discipline Advancement (RGZN2024002), and  the National Social Science Fund of China (25BJY094). Wen Chen acknowledges the support of the BNBU Research Grant with No. of UICR0700082-24 at Beijing Normal-Hong Kong Baptist University, Zhuhai, PR China.
Wen Chen and Nicolas Langren\'e acknowledge the support of the Guangdong Provincial/Zhuhai Key Laboratory of IRADS (2022B1212010006). The authors also acknowledge the support from the Qinghai Provincial Key Laboratory of Big Data in Finance and Artificial Intelligence Application Technology, Qinghai Institute of Technology, Xining 810016, China. The authors are grateful to Dr. Shanqi Liu and two anonymous referees for their valuable comments and suggestions.

\bibliographystyle{plainnat_lastnamefirst} 

\bibliography{biblio}

\clearpage 
\appendix
\section*{Appendix}
\counterwithin{table}{section}
\section{Proofs in Section~\ref{sec:PPDE}}\label{sec:a1}
\subsection{Proof of Theorem~\ref{thm:ppde_cb}}\label{app:proof1}
\begin{proof}
\textcolor{royalblue}{We derive the PPDE satisfied by the convertible bond value
functional $V(t,\omega_t,H_t)$ locally in the continuation region
$\mathcal C$. By definition of $\mathcal C$, neither the reset nor the call condition is active here. Therefore, the conversion price $H_t$ remains locally unchanged, and the bond price evolves continuously along the stock price path. By the assumed local functional $C^{1,2}$ regularity, the functional It\^{o} formula is
applicable as long as the state remains in $\mathcal C$.}
Under the risk-neutral measure $\mathbb Q$, the stock price follows
\[
dS_t=r_tS_t\,dt+\sigma(t,\omega_t,H_t)S_t\,dW_t .
\]
Applying the functional It\^{o} formula to the non-anticipative functional $V(t,\omega_t,H_t)$ yields
\begin{equation}
\label{ito_ppde_step}
dV(t,\omega_t,H_t)
=
\partial_t V(t,\omega_t,H_t)\,dt
+
\partial_{\omega}V(t,\omega_t,H_t)\,dS_t
+
\frac12\partial_{\omega\omega}V(t,\omega_t,H_t)\,d\langle S\rangle_t .
\end{equation}
Since the quadratic variation of the stock price satisfies
\[
d\langle S\rangle_t
=
\sigma^2(t,\omega_t,H_t)S_t^2\,dt,
\]
substituting the stock price dynamics into equation~\eqref{ito_ppde_step} leads to
\begin{align}
dV(t,\omega_t,H_t)
&=
\Bigl(
\partial_tV
+r_tS_t\partial_{\omega}V
+\frac12\sigma^2(t,\omega_t,H_t)S_t^2
\partial_{\omega\omega}V
\Bigr)dt
\notag\\
&\quad
+
\partial_{\omega}V(t,\omega_t,H_t)
\sigma(t,\omega_t,H_t)S_t\,dW_t .
\end{align}
For notational convenience, define the infinitesimal generator $\mathcal{A}$
associated with the stock price dynamics
\[
\mathcal AV
=
r_tS_t\,\partial_{\omega}V
+
\frac12\sigma^2(t,\omega_t,H_t)S_t^2
\partial_{\omega\omega}V .
\]
We then obtain
\begin{equation}
dV(t,\omega_t,H_t)
=
\bigl(\partial_tV+\mathcal AV\bigr)dt
+
\partial_{\omega}V
\sigma(t,\omega_t,H_t)S_t\,dW_t .
\end{equation}
We now consider the pricing condition implied by the absence of
arbitrage under the risk-neutral measure. Since the bond pays coupons
continuously at rate $B\cdot C$, we define the discount factor
\[
D_t=e^{-\int_0^t r_u\,du},
\]
and the discounted cum-dividend value process
\[
M_t
=
D_tV(t,\omega_t,H_t)
+
\int_0^t D_s\,B\cdot C\,ds.
\]
Applying It\^{o}'s product rule to $D_tV(t,\omega_t,H_t)$ yields
\[
d(D_tV_t)
=
D_t(\partial_tV+\mathcal AV-r_tV)\,dt
+
D_t\,\partial_{\omega}V
\sigma(t,\omega_t,H_t)S_t\,dW_t .
\]
Subsequently,
\begin{align}
dM_t
=
D_t\bigl(\partial_tV+\mathcal AV-r_tV+BC\bigr)dt
+
D_t\,\partial_{\omega}V
\sigma(t,\omega_t,H_t)S_t\,dW_t .
\end{align}
\textcolor{royalblue}{As long as the state remains in the continuation region $\mathcal C$, the discounted cum-dividend value process $M_t$ is a local martingale under the risk-neutral measure $\mathbb Q$. Its drift vanishes locally in $\mathcal C$, which
gives}
\[
\partial_tV(t,\omega_t,H_t)
+\mathcal AV(t,\omega_t,H_t)
-r_tV(t,\omega_t,H_t)
+BC
=0.
\]
Finally, the boundary and switching conditions follow from the contractual specifications of the convertible bond. When the call provision is triggered, i.e., $N_t^{\mathrm{call}} \ge N_c$, the issuer imposes the call constraint, and calls the unconverted bond at the redemption price $K_t$ if the bond continuation value is greater than $K_t$. Meanwhile, the bondholder chooses optimally between conversion and continuation. The bond value is therefore
\[
V(t,\omega_t,H_t)
=
\max\left(
\frac{B}{H_t}S_t,\;
\min\bigl(V^{\mathrm{cont}}(t,\omega_t,H_t),K_t\bigr)
\right).
\]
\textcolor{royalblue}{Under the call-first convention in
Assumption~\ref{assump1}, the call condition takes precedence in the exceptional case where the call and reset conditions are
satisfied simultaneously.
When the call condition is not active but the reset condition is
satisfied, namely,
\[
N_t^{\mathrm{call}}(\omega_t;H_t)<N_c,
\qquad
N_t^{\mathrm{reset}}(\omega_t;H_t)\geq N_r,
\]
the conversion price is updated via the reset mapping
$\Psi(\omega_t,H_t)$. Hence, the bond value obeys the
nonlocal switching condition}
\begin{equation}\label{eq:switching}
V(t,\omega_t,H_t)
=
V\bigl(t,\omega_t,\Psi(\omega_t,H_t)\bigr).
\end{equation}
At maturity, the payoff is given by
\[
V(T,\omega_T,H_T)
=
\max\!\left(B,\frac{B}{H_T}S_T\right).
\]
\textcolor{royalblue}{The above argument proves the PPDE in $\mathcal C$ under the assumed functional $C^{1,2}$ regularity. The reset condition is imposed separately through \eqref{eq:switching} and therefore does not enter the local differential operator $\partial_t+\mathcal A$. Hence, no smoothness across the reset boundary is required. More generally, when a classical formulation is not available, the PPDE is interpreted in the viscosity sense in $\mathcal C$, together with the call, reset, and terminal conditions in \eqref{eq:ppde}.}
\end{proof}

\subsection{Proof of Corollary~\ref{cor:gbm_ppde}}\label{app:proofGBM}
\begin{proof}
Applying It\^{o}'s formula to $V(t,S,H_t)$ yields
\begin{align}\label{eq:V}
dV(t,S,H_t)
&=
\partial_t V(t,S,H_t)\,dt
+
\partial_S V(t,S,H_t)\,dS_t
+
\frac12 \partial_{SS}V(t,S,H_t)\,(dS_t)^2.
\end{align}
Substituting the GBM dynamics into the above expression, together with quadratic variation
\[
(dS_t)^2=\sigma^2 S_t^2\,dt,
\]
we obtain
\begin{align}
dV(t,S,H_t)
&=
\partial_t V\,dt
+
\partial_S V\bigl(rS_t\,dt+\sigma S_t\,dW_t\bigr)
+
\frac12 \partial_{SS}V\,\sigma^2 S_t^2\,dt \notag\\
&=
\Bigl(
\partial_t V
+
rS_t\,\partial_S V
+
\frac12\sigma^2 S_t^2\,\partial_{SS}V
\Bigr)\,dt
+
\sigma S_t\,\partial_S V\,dW_t.
\end{align}
Therefore, the infinitesimal generator for the GBM stock price
process is
\[
\mathcal A V
=
rS\,\partial_S V
+
\frac12\sigma^2 S^2\,\partial_{SS}V.
\]
Substituting this generator into the general PPDE in
Theorem~\ref{thm:ppde_cb} yields \eqref{eq:gbm}.
\end{proof}

\subsection{Proof of Corollary~\ref{cor:cev_ppde}}\label{appenCEV}
\begin{proof}

Substituting the CEV dynamics into \eqref{eq:V}, together with quadratic variation
\[
(dS_t)^2=\sigma^2 S_t^{2\gamma}\,dt,
\]
we have
\begin{align}
dV(t,S,H_t)
&=
\partial_t V\,dt
+
\partial_S V\bigl(rS_t\,dt+\sigma S_t^{\gamma}\,dW_t\bigr)
+
\frac12 \partial_{SS}V\,\sigma^2 S_t^{2\gamma}\,dt \notag\\
&=
\Bigl(
\partial_t V
+
rS_t\,\partial_S V
+
\frac12\sigma^2 S_t^{2\gamma}\,\partial_{SS}V
\Bigr)\,dt
+
\sigma S_t^{\gamma}\,\partial_S V\,dW_t.
\end{align}
Therefore, the infinitesimal generator associated with the CEV stock price
process is
\[
\mathcal A V
=
rS\,\partial_S V
+
\frac12\sigma^2 S^{2\gamma}\,\partial_{SS}V.
\]
Substituting this generator into the general PPDE in
Theorem~\ref{thm:ppde_cb} yields equation~\eqref{eq:cev}.
\end{proof}

\subsection{Proof of Corollary~\ref{cor:heston_ppde}}\label{appenheston}
\begin{proof}
Applying It\^{o}'s formula to $V(t,S,v,H_t)$ yields
\begin{align}
dV(t,S,v,H_t)
&=
\partial_t V(t,S,v,H_t)\,dt
+
\partial_S V(t,S,v,H_t)\,dS_t
+
\partial_v V(t,S,v,H_t)\,dv_t \notag\\
&\quad
+
\frac12 \partial_{SS}V(t,S,v,H_t)\,(dS_t)^2
+
\frac12 \partial_{vv}V(t,S,v,H_t)\,(dv_t)^2
+
\partial_{Sv}V(t,S,v,H_t)\,dS_t\,dv_t.
\end{align}
Substituting the Heston dynamics into the above expression, together with
\[
(dS_t)^2=v_t S_t^2\,dt,\qquad
(dv_t)^2=\eta^2 v_t\,dt,\qquad
dS_t\,dv_t=\rho\eta v_t S_t\,dt,
\]
we have
\begin{align}
dV(t,S,v,H_t)
&=
\partial_t V\,dt
+
\partial_S V\bigl(rS_t\,dt+\sqrt{v_t}S_t\,dW_t^S\bigr)
+
\partial_v V\bigl(\kappa(\theta-v_t)\,dt+\eta\sqrt{v_t}\,dW_t^v\bigr)
\notag\\
&\quad
+
\frac12 \partial_{SS}V\,v_t S_t^2\,dt
+
\frac12 \partial_{vv}V\,\eta^2 v_t\,dt
+
\partial_{Sv}V\,\rho\eta v_t S_t\,dt.
\end{align}
Rearranging the drift and diffusion terms, we obtain
\begin{align}
dV(t,S,v,H_t)
&=
\Bigl(
\partial_t V
+
rS_t\,\partial_S V
+
\kappa(\theta-v_t)\,\partial_v V
+
\frac12 v_t S_t^2\,\partial_{SS}V
\notag\\
&\qquad
+
\frac12\eta^2 v_t\,\partial_{vv}V
+
\rho\eta v_t S_t\,\partial_{Sv}V
\Bigr)\,dt
\notag\\
&\quad
+
\sqrt{v_t}S_t\,\partial_S V\,dW_t^S
+
\eta\sqrt{v_t}\,\partial_v V\,dW_t^v.
\end{align}
Therefore, the infinitesimal generator for the Heston dynamics is
\[
\mathcal A V
=
rS\,\partial_S V
+
\kappa(\theta-v)\,\partial_v V
+
\frac12 vS^2\,\partial_{SS}V
+
\frac12\eta^2 v\,\partial_{vv}V
+
\rho\eta vS\,\partial_{Sv}V.
\]
Substituting this generator into the general PPDE in
Theorem~\ref{thm:ppde_cb} yields equation~\eqref{eq:heston}.
\end{proof}

\section{Proofs in Section~\ref{sec:deepapp}}\label{appen:a2}

{\color{royalblue}
\subsection{Proof of Theorem~\ref{thm:ppde_cb_viscosity}}
\label{app:proof_ppde_viscosity}
\begin{proof}
\medskip\noindent\textcolor{royalblue}{\textbf{Step 1: Well-posedness of the state process.}}\par
\vspace{0.2cm}
Fix $i\in\{0,\ldots,N-1\}$ and $t\in(t_i,t_{i+1})$. No monitoring occurs within this interval, so $H$ and the ordered window state remain unchanged. For
completeness, we first verify that the state process used in the valuation is well posed. 
Let $S$ and $\widetilde S$ be two solutions driven by the same
Brownian motion and sharing identical stopped initial path. {For $R>0$, define
\[
\tau_R:=\inf\{q\geq t:|S_q|\vee|\widetilde S_q|\geq R\}\wedge t_{i+1}.
\]
The estimates below are first applied on $[t,s\wedge\tau_R]$. We suppress the stopping notation until the application of Gr{\"o}nwall's inequality.} Then
\begin{equation*}
\Delta S_q
=\int_t^q\Delta\mu_u\,du+\int_t^q\Delta\nu_u\,dW_u,
\end{equation*}
where
$
\Delta S_q:=S_q-\widetilde S_q,~
\Delta\mu_u:=\mu(u,\omega_u,H)-\mu(u,\widetilde\omega_u,H),
~
\Delta\nu_u:=\nu(u,\omega_u,H)-\nu(u,\widetilde\omega_u,H).
$
Using $|x+y|^2\leq2|x|^2+2|y|^2$, we have
\begin{align*}
\mathbb E\left[\sup_{t\leq q\leq s}|\Delta S_q|^2\right]
\leq 2\mathbb E\left[\sup_{t\leq q\leq s}
\left|\int_t^q\Delta\mu_u\,du\right|^2\right]
+2\mathbb E\left[\sup_{t\leq q\leq s}
\left|\int_t^q\Delta\nu_u\,dW_u\right|^2\right].
\end{align*}
The Cauchy--Schwarz inequality implies
\begin{equation*}
\sup_{t\leq q\leq s}\left|\int_t^q\Delta\mu_u\,du\right|^2
\leq(s-t)\int_t^s|\Delta\mu_u|^2\,du,
\end{equation*}
while the Burkholder--Davis--Gundy (BDG) inequality gives
\begin{equation*}
\mathbb E\left[\sup_{t\leq q\leq s}
\left|\int_t^q\Delta\nu_u\,dW_u\right|^2\right]
\leq C_{\rm BDG}\mathbb E\int_t^s|\Delta\nu_u|^2\,du.
\end{equation*}
Here and below, $C_{\rm BDG}>0$ denotes the constant in the BDG inequality, $C_p>0$ denotes a generic constant depending on $p$ that may vary from line to line.
Consequently,
\begin{align*}
\mathbb E\left[\sup_{t\le q\le s}|S_q-\widetilde S_q|^2\right]
&\leq2(s-t)\mathbb E\int_t^s|\Delta\mu_u|^2\,du
+2C_{\rm BDG}\mathbb E\int_t^s|\Delta\nu_u|^2\,du.
\end{align*}
Since the two solutions have the same history up to $t$, the pathwise
Lipschitz condition \eqref{eq:path_lipschitz_proof} gives
\begin{equation*}
|\Delta\mu_u|^2+|\Delta\nu_u|^2
\leq C_L\|\omega-\widetilde\omega\|_u^2
=C_L\sup_{t\leq q\leq u}|S_q-\widetilde S_q|^2.
\end{equation*}
Substituting this bound into the preceding estimate gives
\begin{align*}
\mathbb E\left[\sup_{t\le q\le s}|S_q-\widetilde S_q|^2\right]
&\le C_L\int_t^s
\mathbb E\left[\sup_{t\le q\le u}|S_q-\widetilde S_q|^2\right]du.
\end{align*}
The Gr{\"o}nwall inequality therefore implies
\begin{equation*}
{\mathbb E\left[\sup_{t\le q\le s\wedge\tau_R}|S_q-\widetilde S_q|^2\right]=0,}
\end{equation*}
{Taking $R\to\infty$ gives pathwise uniqueness up to the explosion time.} Picard iteration under
\eqref{eq:path_lipschitz_proof}--\eqref{eq:linear_growth_proof} gives existence of a strong solution. {For this solution, set
$
\sigma_R:=\inf\{q\geq t:|S_q|\geq R\}\wedge t_{i+1}.
$
The following moment estimates are first applied to $S_{\cdot\wedge\sigma_R}$, again with the stopping notation suppressed.} Moreover, for any $p\ge2$, the BDG inequality and the
linear-growth condition give the following moment estimate. From the integral
representation of the state equation and the inequality
$|x+y+z|^p\leq3^{p-1}(|x|^p+|y|^p+|z|^p)$, we obtain
\begin{equation}\label{eq:moment_decomposition_proof}
\begin{aligned}
\mathbb E_{t,\omega}^{\mathbb Q}
\left[\sup_{t\leq q\leq s}|S_q|^p\right]
&\leq C_p\|\omega\|_t^p
+C_p\mathbb E_{t,\omega}^{\mathbb Q}\left[
\sup_{t\leq q\leq s}
\left|\int_t^q\mu(u,\omega_u,H)\,du\right|^p\right]\\
&\quad+C_p\mathbb E_{t,\omega}^{\mathbb Q}\left[
\sup_{t\leq q\leq s}
\left|\int_t^q\nu(u,\omega_u,H)\,dW_u\right|^p\right].
\end{aligned}
\end{equation}
H\"older's inequality bounds the drift term:
\begin{equation*}
\sup_{t\leq q\leq s}
\left|\int_t^q\mu(u,\omega_u,H)\,du\right|^p
\leq(s-t)^{p-1}\int_t^s|\mu(u,\omega_u,H)|^p\,du.
\end{equation*}
For the stochastic integral, the BDG and H\"older inequalities imply
\begin{align*}
\mathbb E_{t,\omega}^{\mathbb Q}\left[
\sup_{t\leq q\leq s}
\left|\int_t^q\nu(u,\omega_u,H)\,dW_u\right|^p\right]
&\leq C_p\mathbb E_{t,\omega}^{\mathbb Q}\left[
\left(\int_t^s|\nu(u,\omega_u,H)|^2\,du\right)^{p/2}\right]\\
&\leq C_p(s-t)^{p/2-1}
\mathbb E_{t,\omega}^{\mathbb Q}
\int_t^s|\nu(u,\omega_u,H)|^p\,du.
\end{align*}
By the linear-growth condition \eqref{eq:linear_growth_proof},
\begin{equation*}
|\mu(u,\omega_u,H)|^p+|\nu(u,\omega_u,H)|^p
\leq C_p\bigl(1+\|\omega_u\|_u^p+|H|^p\bigr).
\end{equation*}
Since
\begin{equation*}
\|\omega_u\|_u^p
\leq\|\omega\|_t^p+\sup_{t\leq v\leq u}|S_v|^p,
\end{equation*}
Combining these bounds with \eqref{eq:linear_growth_proof} in
\eqref{eq:moment_decomposition_proof} yields
\begin{align*}
\mathbb E_{t,\omega}^{\mathbb Q}
\left[\sup_{t\le q\le s}|S_q|^p\right]
&\le C_p\left(1+\|\omega\|_t^p+|H|^p
+\int_t^s\mathbb E_{t,\omega}^{\mathbb Q}
\left[\sup_{t\le v\le u}|S_v|^p\right]du\right).
\end{align*}
Applying Gr{\"o}nwall's inequality once more, we have
{
\begin{equation*}
\mathbb E_{t,\omega}^{\mathbb Q}
\left[\sup_{t\le s\le t_{i+1}}|S_{s\wedge\sigma_R}|^p\right]
\le C_{p,t_{i+1}}\bigl(1+\|\omega\|_t^p+|H|^p\bigr),
\end{equation*}
where the right-hand side is independent of $R$. We now take $R\to\infty$ and apply Fatou's lemma to obtain
}
\begin{equation}\label{eq:moment_bound}
\mathbb E_{t,\omega}^{\mathbb Q}
\left[\sup_{t\le s\le T}|S_s|^p\right]
\le C_{p,T}\bigl(1+\|\omega\|_t^p+|H|^p\bigr)<\infty.
\end{equation}
{Particularly,
$\mathbb P(\sigma_R<t_{i+1})\leq C_{p,t_{i+1}}(1+\|\omega\|_t^p+|H|^p)/R^p\to0$.}
This estimate excludes finite-time explosion. Thus
\eqref{eq:unified_state} admits a unique non-explosive strong solution {on $[t,t_{i+1}]$}.

\medskip\noindent\textcolor{royalblue}{\textbf{Step 2: Dynamic programming principle.}}\par
\vspace{0.1cm}
This step establishes the dynamic programming principle over each monitoring interval using the tower property of conditional expectation. This is the probabilistic foundation for the PPDE viscosity framework of \textcolor{darkpurple}{\cite{ekren2014ppde}}.
For any stopping time $t\le\tau\le t_{i+1}$, the tower property gives the dynamic programming principle
\begin{equation}\label{eq:DPP_proof}
V(t,\omega,H)=\mathbb E_{t,\omega,H}^{\mathbb Q}\!\left[
e^{-\int_t^\tau r_u du}V(\tau,\omega_\tau,H)
+\int_t^\tau e^{-\int_t^s r_u du}BC\,ds\right].
\end{equation}
The moment estimates \eqref{eq:moment_bound} and
\eqref{eq:conversion_ratio_moment_proof}, together with the polynomial-growth assumptions, ensure that the discounted continuation value and accumulated coupon payments have finite expectations. Hence the conditional expectation in \eqref{eq:DPP_proof} is well defined.

To verify this relation, we denote
$
D_{t,s}=\exp\left(-\int_t^s r_u\,du\right),
~
A_{t,s}=\int_t^sD_{t,u}BC\,du.
$
{Taking $\tau=t_{i+1}^-$ in \eqref{eq:DPP_proof} and using
$V(t_{i+1}^-,x)=G_{i+1}(x)$ gives}
{
\begin{equation}\label{eq:interval_representation_proof}
\begin{aligned}
V(t,\omega,H)
&=\mathbb E_{t,\omega,H}^{\mathbb Q}\!\left[
e^{-\int_t^{t_{i+1}}r_u\,du}
V(t_{i+1}^-,\omega_{t_{i+1}},H)
+\int_t^{t_{i+1}}e^{-\int_t^s r_u\,du}BC\,ds
\right]\\
&=\mathbb E_{t,\omega,H}^{\mathbb Q}\!\left[
D_{t,t_{i+1}}V(t_{i+1}^-,\omega_{t_{i+1}},H)
+A_{t,t_{i+1}}\right]\\
&=\mathbb E_{t,\omega,H}^{\mathbb Q}\!\left[
D_{t,t_{i+1}}G_{i+1}(\omega_{t_{i+1}},H)
+A_{t,t_{i+1}}\right],
\end{aligned}
\end{equation}
}
where
\begin{equation*}
G_{i+1}(x)
:=
\begin{cases}
\bigl(\mathcal J_{i+1}V(t_{i+1}+,\cdot)\bigr)(x),
&i+1<N,\\[4pt]
\displaystyle\max\left\{B,\frac{B}{H_T}S_T\right\},
&i+1=N.
\end{cases}
\end{equation*}
For $t\le\tau\le t_{i+1}$, the multiplicative identity
$
D_{t,t_{i+1}}=D_{t,\tau}D_{\tau,t_{i+1}}
$
and the additive decomposition
$
A_{t,t_{i+1}}
=A_{t,\tau}+D_{t,\tau}A_{\tau,t_{i+1}}
$
imply
\begin{align*}
V(t,\omega,H)
&=\mathbb E_{t,\omega,H}^{\mathbb Q}\!\left[
A_{t,\tau}+D_{t,\tau}
\mathbb E_{\tau,\omega_\tau,H}^{\mathbb Q}
\left[D_{\tau,t_{i+1}}G_{i+1}+A_{\tau,t_{i+1}}\right]
\right]\notag\\
&=\mathbb E_{t,\omega,H}^{\mathbb Q}\!\left[
A_{t,\tau}+D_{t,\tau}V(\tau,\omega_\tau,H)\right],
\end{align*}
which proves equation~\eqref{eq:DPP_proof}.

\medskip\noindent\textcolor{royalblue}{\textbf{Step 3: Discounted functional It\^o formula.}}\par
We apply the functional It\^o formula to a discounted test functional and
identify the PPDE operator in the resulting drift term.
On $(t_i,t_{i+1})$, the conversion price $H$ is constant and
\begin{equation*}
dS_s=r_sS_s\,ds+\sigma(s,\omega_s,H)S_s\,dW_s.
\end{equation*}
The quadratic variation of $S$ therefore satisfies
\begin{equation*}
d\langle S\rangle_s
=\sigma^2(s,\omega_s,H)S_s^2\,ds.
\end{equation*}
For any $\varphi\in C^{1,2}$, the functional It\^o formula gives
\begin{equation*}
d\varphi(s,\omega_s,H)
=\partial_s\varphi\,ds
+\partial_\omega\varphi\,dS_s
+\frac12\partial_{\omega\omega}^2\varphi\,d\langle S\rangle_s.
\end{equation*}
Substituting the preceding expressions for $dS_s$ and
$d\langle S\rangle_s$ yields
\begin{equation*}
\begin{aligned}
d\varphi(s,\omega_s,H)
&=\partial_s\varphi\,ds
+\partial_\omega\varphi
\bigl[r_sS_s\,ds+\sigma(s,\omega_s,H)S_s\,dW_s\bigr]
+\frac12\partial_{\omega\omega}^2\varphi
\bigl[\sigma^2(s,\omega_s,H)S_s^2\,ds\bigr]\\
&=\partial_s\varphi\,ds
+r_sS_s\partial_\omega\varphi\,ds
+\sigma(s,\omega_s,H)S_s\partial_\omega\varphi\,dW_s
+\frac12\sigma^2(s,\omega_s,H)S_s^2
\partial_{\omega\omega}^2\varphi\,ds\\
&=\left(\partial_s\varphi
+r_sS_s\partial_\omega\varphi
+\frac12\sigma^2(s,\omega_s,H)S_s^2
\partial_{\omega\omega}^2\varphi\right)ds
+\sigma(s,\omega_s,H)S_s
\partial_\omega\varphi\,dW_s\\
&=(\partial_s\varphi+\mathcal A\varphi)ds
+\sigma(s,\omega_s,H)S_s\partial_\omega\varphi\,dW_s.
\end{aligned}
\end{equation*}
Since $dD_{t,s}=-r_sD_{t,s}ds$, the product rule yields
\begin{align}
d\bigl(D_{t,s}\varphi(s,\omega_s,H)\bigr)
&=D_{t,s}(\partial_s\varphi+\mathcal A\varphi-r_s\varphi)ds
+D_{t,s}\sigma(s,\omega_s,H)S_s
\partial_\omega\varphi\,dW_s.
\label{eq:discounted_ito_full_proof}
\end{align}

\medskip\noindent\textcolor{royalblue}{\textbf{Step 4: Viscosity subsolution property.}}\par
We combine the dynamic programming principle above with a local contradiction argument to prove the viscosity subsolution inequality. The test-functional and optimal-stopping argument follows the standard PPDE
viscosity method in {\cite{ekren2014ppde}}.
{Let $\varphi\in C^{1,2}$ be an admissible test functional
tangent to $V$ from above at $(t,\omega,H)$ in the mean-tangency sense
of \cite{ren2020comparison}. In particular,}
\begin{equation*}
{V(t,\omega,H)=\varphi(t,\omega,H).}
\end{equation*}
locally. If the subsolution inequality fails, then there exists $\eta>0$ and
a neighbourhood around the contact point on which
\begin{equation}\label{eq:sub_local_proof}
-\partial_s\varphi-\mathcal A\varphi+r_s\varphi-BC\ge\eta.
\end{equation}
To localise the argument, we use the stopped-path metric as
\begin{equation*}
d_\infty\bigl((t,\omega),(t',\omega')\bigr)
:=|t-t'|
+\sup_{0\leq u\leq T}
\left|\omega_{u\wedge t}-\omega'_{u\wedge t'}\right|.
\end{equation*}
{Choose $\delta>0$ sufficiently small so that the localised exit time below does not exceed the localising time of the test functional}
\begin{equation}\label{eq:exit_time_proof}
\tau_\delta=(t+\delta)\wedge
\inf\{s>t:d_\infty((s,\omega_s),(t,\omega))\ge\delta\},
\end{equation}
so that $\tau_\delta<t_{i+1}$. Applying the discounted functional It\^o formula
\eqref{eq:discounted_ito_full_proof} up to $\tau_\delta$ gives
\begin{equation*}
\begin{aligned}
D_{t,\tau_\delta}\varphi(\tau_\delta,\omega_{\tau_\delta},H)
-\varphi(t,\omega,H)
&=\int_t^{\tau_\delta}D_{t,s}
(\partial_s\varphi+\mathcal A\varphi-r_s\varphi)\,ds
+\int_t^{\tau_\delta}D_{t,s}
\sigma(s,\omega_s,H)S_s\partial_\omega\varphi\,dW_s.
\end{aligned}
\end{equation*}
Rearranging it gives
\begin{equation*}
\begin{aligned}
\varphi(t,\omega,H)
&=D_{t,\tau_\delta}\varphi(\tau_\delta,\omega_{\tau_\delta},H)
-\int_t^{\tau_\delta}D_{t,s}
(\partial_s\varphi+\mathcal A\varphi-r_s\varphi)\,ds
-\int_t^{\tau_\delta}D_{t,s}
\sigma(s,\omega_s,H)S_s\partial_\omega\varphi\,dW_s.
\end{aligned}
\end{equation*}
 The stochastic integral is a martingale with zero conditional expectation after localisation. Taking conditional expectations therefore gives
\begin{equation*}
\begin{aligned}
\varphi(t,\omega,H)
&=\mathbb E^{\mathbb Q}_{t,\omega,H}\!\left[
e^{-\int_t^{\tau_\delta}r_u du}
\varphi(\tau_\delta,\omega_{\tau_\delta},H)\right]
-\mathbb E^{\mathbb Q}_{t,\omega,H}\!\left[
\int_t^{\tau_\delta}e^{-\int_t^s r_u du}
(\partial_s\varphi+\mathcal A\varphi-r_s\varphi)ds\right].
\end{aligned}
\end{equation*}
By inequality \eqref{eq:sub_local_proof}, we have
\begin{equation*}
\begin{aligned}
\varphi(t,\omega,H)
&\ge\mathbb E^{\mathbb Q}_{t,\omega,H}\!\left[
e^{-\int_t^{\tau_\delta}r_u du}
\varphi(\tau_\delta,\omega_{\tau_\delta},H)
+\int_t^{\tau_\delta}e^{-\int_t^s r_u du}BC\,ds\right]\notag
+\eta\mathbb E^{\mathbb Q}_{t,\omega,H}\!\left[
\int_t^{\tau_\delta}e^{-\int_t^s r_u du}ds\right].
\end{aligned}
\end{equation*}
{After the
standard discounting transformation, the defining mean-tangency condition gives}
\begin{equation*}
{
\mathbb E^{\mathbb Q}_{t,\omega,H}\!\left[
D_{t,\tau_\delta}
\bigl(\varphi(\tau_\delta,\omega_{\tau_\delta},H)
-V(\tau_\delta,\omega_{\tau_\delta},H)\bigr)\right]\geq0.}
\end{equation*}
{Combining this relation with the contact equality,
the preceding inequality, and the DPP \eqref{eq:DPP_proof} yields}
{
\begin{align*}
V(t,\omega,H)
&=\varphi(t,\omega,H)\\
&\geq\mathbb E^{\mathbb Q}_{t,\omega,H}\!\left[
D_{t,\tau_\delta}\varphi(\tau_\delta,\omega_{\tau_\delta},H)
+A_{t,\tau_\delta}\right]
+\eta \mathbb E^{\mathbb Q}_{t,\omega,H}\!\left[
\int_t^{\tau_\delta}D_{t,s}\,ds\right]\\
&=\mathbb E^{\mathbb Q}_{t,\omega,H}\!\left[
D_{t,\tau_\delta}V(\tau_\delta,\omega_{\tau_\delta},H)
+A_{t,\tau_\delta}\right]\\
&\quad+\mathbb E^{\mathbb Q}_{t,\omega,H}\!\left[
D_{t,\tau_\delta}
\bigl(\varphi(\tau_\delta,\omega_{\tau_\delta},H)
-V(\tau_\delta,\omega_{\tau_\delta},H)\bigr)\right]
+\eta \mathbb E^{\mathbb Q}_{t,\omega,H}\!\left[
\int_t^{\tau_\delta}D_{t,s}\,ds\right]\\
&\geq\mathbb E^{\mathbb Q}_{t,\omega,H}\!\left[
D_{t,\tau_\delta}V(\tau_\delta,\omega_{\tau_\delta},H)
+A_{t,\tau_\delta}\right]
+\eta \mathbb E^{\mathbb Q}_{t,\omega,H}\!\left[
\int_t^{\tau_\delta}D_{t,s}\,ds\right]\\
&=V(t,\omega,H)
+\eta \mathbb E^{\mathbb Q}_{t,\omega,H}\!\left[
\int_t^{\tau_\delta}D_{t,s}\,ds\right].
\end{align*}
}
{Subtracting $V(t,\omega,H)$ from both sides gives}
{
\begin{equation*}
0\geq\eta\mathbb E^{\mathbb Q}_{t,\omega,H}\!\left[
\int_t^{\tau_\delta}D_{t,s}\,ds\right]>0,
\end{equation*}
}
which yields a contradiction. Hence
\begin{equation*}
-\partial_t\varphi-\mathcal A\varphi+r_tV-BC\le0,
\end{equation*}
and $V$ is a viscosity subsolution.

\medskip\noindent\textcolor{royalblue}{\textbf{Step 5: Viscosity supersolution property.}}\par
This step reverses the preceding contact argument to prove the viscosity supersolution inequality.
{Let $\psi\in C^{1,2}$ be an admissible test functional
tangent to $V$ from below at $(t,\omega,H)$ in the same
mean-tangency sense}
\begin{equation}\label{eq:lower_contact_proof}
{V(t,\omega,H)=\psi(t,\omega,H).}
\end{equation}
Suppose the supersolution inequality fails, then there exists $\eta>0$, such that, after reducing the neighbourhood,
\begin{equation*}
-\partial_s\psi-\mathcal A\psi+r_s\psi-BC\le-\eta.
\end{equation*}
Applying \eqref{eq:discounted_ito_full_proof} to $\psi$ up to the exit time \eqref{eq:exit_time_proof} gives
\begin{align}
\psi(t,\omega,H)
&\le\mathbb E^{\mathbb Q}_{t,\omega,H}\!\left[
D_{t,\tau_\delta}\psi(\tau_\delta,\omega_{\tau_\delta},H)
+A_{t,\tau_\delta}\right]
-\eta\mathbb E^{\mathbb Q}_{t,\omega,H}\!\left[
\int_t^{\tau_\delta}D_{t,s}\,ds\right].
\label{eq:strict_super_proof}
\end{align}
{The corresponding mean-tangency condition implies}
\begin{equation*}
{
\mathbb E^{\mathbb Q}_{t,\omega,H}\!\left[
D_{t,\tau_\delta}
\bigl(\psi(\tau_\delta,\omega_{\tau_\delta},H)
-V(\tau_\delta,\omega_{\tau_\delta},H)\bigr)\right]\leq0.}
\end{equation*}
{Combining the mean-tangency inequality with
\eqref{eq:lower_contact_proof}, \eqref{eq:strict_super_proof},
and the DPP gives}
{
\begin{align*}
0
&=\psi(t,\omega,H)-V(t,\omega,H)\\
&\leq\mathbb E^{\mathbb Q}_{t,\omega,H}\!\left[
D_{t,\tau_\delta}
\bigl(\psi(\tau_\delta,\omega_{\tau_\delta},H)
-V(\tau_\delta,\omega_{\tau_\delta},H)\bigr)\right]
-\eta\mathbb E^{\mathbb Q}_{t,\omega,H}\!\left[
\int_t^{\tau_\delta}D_{t,s}\,ds\right]\\
&\leq-\eta\mathbb E^{\mathbb Q}_{t,\omega,H}\!\left[
\int_t^{\tau_\delta}D_{t,s}\,ds\right]<0.
\end{align*}
}
 Therefore
\begin{equation*}
-\partial_t\psi-\mathcal A\psi+r_tV-BC\ge0.
\end{equation*}
Thus $V$ is a viscosity solution of \eqref{eq:unified_interval_ppde}.

\medskip\noindent\textcolor{royalblue}{\textbf{Step 6: Comparison on a monitoring interval.}}\par
We establish that an interval viscosity subsolution cannot exceed an interval viscosity supersolution when their terminal values are ordered.
Let $u$ be a viscosity subsolution and $v$ be a viscosity supersolution respectively of
\eqref{eq:unified_interval_ppde}, both with polynomial growth. Suppose that
$
u^*(t_{i+1}^-,\cdot)\leq v_*(t_{i+1}^-,\cdot).
$
For $\varepsilon>0$, we define
\begin{equation*}
v^\varepsilon:=v+\varepsilon\chi.
\end{equation*}
By \eqref{eq:lyapunov_comparison_proof} and the linearity of the interval operator, $v^\varepsilon$ is a strict viscosity supersolution satisfying
\begin{equation*}
-\partial_tv^\varepsilon-\mathcal Av^\varepsilon
+r_tv^\varepsilon-BC\geq\varepsilon.
\end{equation*}
Because $\chi$ dominates the polynomial growth of $u$ and $v$, so we may select $\chi$ such that $\chi(t,\omega,H)\to\infty$ and
\begin{equation*}
\frac{|u(t,\omega,H)|+|v(t,\omega,H)|}
{\chi(t,\omega,H)}\longrightarrow0
\qquad\text{as }\|\omega\|_t\to\infty.
\end{equation*}
Since $v^\varepsilon=v+\varepsilon\chi$, we have
\begin{align*}
u(t,\omega,H)-v^\varepsilon(t,\omega,H)
&=u(t,\omega,H)-v(t,\omega,H)
-\varepsilon\chi(t,\omega,H)\\
&\leq |u(t,\omega,H)|+|v(t,\omega,H)|
-\varepsilon\chi(t,\omega,H)\\
&=\chi(t,\omega,H)
\left(
\frac{|u(t,\omega,H)|+|v(t,\omega,H)|}
{\chi(t,\omega,H)}-\varepsilon
\right).
\end{align*}
The expression in parentheses converges to $-\varepsilon$, while
$\chi(t,\omega,H)\to\infty$. Therefore,
\begin{equation*}
\limsup_{\|\omega\|_t\to\infty}
\bigl(u(t,\omega,H)-v^\varepsilon(t,\omega,H)\bigr)=-\infty.
\end{equation*}
Any positive maximum of $u-v^\varepsilon$ must therefore lie within a bounded
stopped-path domain. On this domain, the coefficients satisfy the boundedness,
uniform-continuity, and pathwise Lipschitz assumptions of the semilinear PPDE
comparison theorem. Applying Theorem~4.1 of  \textcolor{darkpurple}{\cite{ren2020comparison}} yields
\begin{equation*}
u\leq v^\varepsilon=v+\varepsilon\chi
\qquad\text{on }[t_i,t_{i+1}].
\end{equation*}
Since $\chi$ is finite at every fixed $(t,\omega,H)$, taking the limit as
$\varepsilon\to0^+$ gives
\begin{equation*}
u(t,\omega,H)\leq v(t,\omega,H),
\qquad t\in[t_i,t_{i+1}].
\end{equation*}
This proves the interval comparison principle in the polynomial-growth class.

\medskip\noindent\textcolor{royalblue}{\textbf{Step 7: Monotonicity of the monitoring operator.}}\par
The call-first transmission operator preserves pointwise ordering. This property
allows the comparison result across monitoring dates.
Indeed, the minimum and maximum operations and composition with $\Gamma_i$ or
$\Gamma_i^R$ are all order-preserving. Hence
\begin{equation*}
\phi\le\psi\quad\Longrightarrow\quad
\mathcal J_i\phi\le\mathcal J_i\psi.
\end{equation*}
Within the call region, we have
\begin{align*}
\phi(\Gamma_i x)&\le\psi(\Gamma_i x),\\
\min\{\phi(\Gamma_i x),K_i\}
&\le\min\{\psi(\Gamma_i x),K_i\},\\
\max\left\{\frac{B}{H_i}S_i,
\min\{\phi(\Gamma_i x),K_i\}\right\}
&\le
\max\left\{\frac{B}{H_i}S_i,
\min\{\psi(\Gamma_i x),K_i\}\right\}.
\end{align*}
In the reset and continuation regions, we have
$
\phi(\Gamma_i^Rx)\le\psi(\Gamma_i^Rx),
~
\phi(\Gamma_ix)\le\psi(\Gamma_ix)
$ respectively.
Thus transmission preserves the ordering required by comparison. The strict
and weak trigger inequalities fix the pointwise value on a threshold surface
through \eqref{eq:pointwise_transmission_proof}. Taking upper and lower relaxed
limits gives \eqref{eq:relaxed_transmission_proof}.

\medskip\noindent\textcolor{royalblue}{\textbf{Step 8: Backward induction and global uniqueness.}}\par
Terminal uniqueness, interval comparison, and monotone transmission together
determine the contract value over the entire horizon.
At maturity, the value is uniquely fixed by
\eqref{eq:terminal_payoff_proof}. This gives the terminal condition for the
last interval. For any $i\leq N-2$, once $V(t_{i+1}^+,\cdot)$ is known, the
contract supplies the terminal value for the preceding interval through
\begin{equation}\label{eq:backward_terminal_proof}
V(t_{i+1}^-,\cdot)=\mathcal J_{i+1}V(t_{i+1}^+,\cdot).
\end{equation}
On every open trigger cell, the applicable branch of $\mathcal J_{i+1}$ is
fixed. Thus \eqref{eq:backward_terminal_proof} is continuous there and of
polynomial growth. Let $U$ and $V$ be two
piecewise viscosity solutions satisfying identical terminal and transmission
conditions. At maturity,
\begin{equation*}
U(T,X_N)
=\max\left\{B,\frac{B}{H_T}S_T\right\}
=V(T,X_N).
\end{equation*}
Comparison on $(t_{N-1},T)$ gives
\begin{equation*}
U=V\qquad\text{on }(t_{N-1},T).
\end{equation*}
If $U(t_{i+1}^+)=V(t_{i+1}^+)$, then
\begin{equation*}
U(t_{i+1}^-)
=\mathcal J_{i+1}U(t_{i+1}^+)
=\mathcal J_{i+1}V(t_{i+1}^+)
=V(t_{i+1}^-).
\end{equation*}
Comparison on $(t_i,t_{i+1})$ consequently gives
\begin{equation*}
U(t,\cdot)=V(t,\cdot),
\qquad t\in(t_i,t_{i+1}).
\end{equation*}
Repeating this argument by backward induction for $i=N-2,\ldots,0$ shows that
$U=V$ over the entire horizon. The existence construction starts from the terminal payoff
$G_N(x):=\max\left\{B,\frac{B}{H_T}S_T\right\}.$
{Proceeding backward for $i=N-1,\ldots,0$, suppose that
$G_{i+1}$ has been defined. Define $V$ on $(t_i,t_{i+1})$ by
\eqref{eq:interval_representation_proof}. The moment estimates ensure that the conditional expectation is finite. For $i\geq1$, set
\begin{equation*}
G_i:=\mathcal J_iV(t_i^+,\cdot),
\end{equation*}
which supplies the terminal value for the preceding
interval. Steps 4 and 5 verify the viscosity property on each monitoring interval, while Step 7 verifies the transmission condition. Repeating this backward construction down to $t_0$ yields a piecewise viscosity solution on the entire horizon. This proves existence, and the preceding comparison argument
proves uniqueness. This concludes the proof.}
\end{proof}

\subsection{Proof of Lemma \ref{lem:operator_lipschitz}}\label{appen:lem1}
\textcolor{royalblue}{\begin{proof}
    For any measurable function $\phi$, define
    \[
    Z_\phi(x)
    :=
    \mathbb E^{\mathbb Q}\!\left[
    e^{-r_ih}\phi(X_{i+1})+hBC
    \mid X_i=x
    \right].
    \]
    The coupon term vanishes upon taking differences, applying Jensen's inequality for conditional expectation yields
    \[
    \|Z_\phi-Z_\psi\|_{L^2}
    \leq
    e^{-r_ih}
    \|\phi-\psi\|_{L^2}.
    \] 
    In the continuation region,
    $\mathcal T_i\phi=Z_\phi$. In the call region,
    \[
    \mathcal T_i\phi
    =
    \max\left\{
    \frac{B}{H_i}S_i,\,
    \min(Z_\phi,K_i)
    \right\}.
    \]
    For fixed $A$ and $K$, the mapping
    \[
    z\longmapsto\max\{A,\min(z,K)\}
    \]
    is $1$-Lipschitz. Therefore, across both regions,
    \[
    \|\mathcal T_i\phi-\mathcal T_i\psi\|_{L^2(\mu_i)}
    \leq
    \|Z_\phi-Z_\psi\|_{L^2(\mu_i)}
    \leq
    e^{-r_ih}
    \|\phi-\psi\|_{L^2(\mu_{i+1})}.
    \]
    Since $r_i\geq0$, we have $e^{-r_ih}\leq1$, and therefore
    $\mathcal T_i$ is non-expansive in $L^2$.
    \end{proof}}

\subsection{Proof of Proposition~\ref{prop:errorp}}\label{appen:prop1}
\begin{proof}
Using the exact backward recursion
\[
V_i^h = \mathcal{T}_i V_{i+1}^h,
\]
we write
\[
\widehat V_i - V_i^h
=
\widehat V_i - \mathcal{T}_i V_{i+1}^h
=
\bigl(\widehat V_i - \mathcal{T}_i \widehat V_{i+1}\bigr)
+
\bigl(\mathcal{T}_i \widehat V_{i+1} - \mathcal{T}_i V_{i+1}^h\bigr).
\]
Taking $L^2$ norms and applying the triangle inequality gives
\[
\|\widehat V_i - V_i^h\|_{L^2}
\le
\|\widehat V_i - \mathcal{T}_i \widehat V_{i+1}\|_{L^2}
+
\|\mathcal{T}_i \widehat V_{i+1} - \mathcal{T}_i V_{i+1}^h\|_{L^2}.
\]
By assumption,
\[
\|\widehat V_i - \mathcal{T}_i \widehat V_{i+1}\|_{L^2} \le \varepsilon_i,
\]
and by Lemma~\ref{lem:operator_lipschitz},
\[
\|\mathcal{T}_i \widehat V_{i+1} - \mathcal{T}_i V_{i+1}^h\|_{L^2}
\le
e^{-r_i h} \|\widehat V_{i+1} - V_{i+1}^h\|_{L^2}.
\]
Combining the above inequalities yields
\[
\|\widehat V_i - V_i^h\|_{L^2}
\le
\varepsilon_i + e^{-r_i h}
\|\widehat V_{i+1} - V_{i+1}^h\|_{L^2}.
\]
The desired result follows.
\end{proof}

\subsection{Proof of Theorem~\ref{thm:conver}}\label{appen:thm2}
\textcolor{royalblue}{
\begin{proof}
Define
\[
e_i := \|\widehat V_i - V_i^h\|_{L^2}.
\]
By Proposition~\ref{prop:errorp},
\[
e_i \le \varepsilon_i + e^{-r_i h} e_{i+1},
\qquad i=0,\dots,N-1.
\]
Since the terminal payoff is imposed exactly,
\[
\widehat V_N(X_N)=V_N^h(X_N)=g(X_N),
\]
and hence $e_N=0$. Iterating the preceding inequality backward gives
\begin{align*}
e_0 &\leq \varepsilon_0
+e^{-r_0h}\varepsilon_1
+e^{-(r_0+r_1)h}\varepsilon_2
+\cdots + \exp\left(
-h\sum_{k=0}^{N-2}r_k
\right)\varepsilon_{N-1}.\\
& = \varepsilon_0
+
\sum_{j=1}^{N-1}
\exp\left(
-h\sum_{k=0}^{j-1}r_k
\right)\varepsilon_j.
\end{align*}
Since $r_i\geq0$, every discount factor in this sum is no greater than one. Therefore,
\[
e_0
\leq
\sum_{j=0}^{N-1}\varepsilon_j.
\]
For a fixed time grid, $N$ is finite. Hence, if the accumulated one-step approximation error tends to zero, then
\[
\widehat V_0\longrightarrow V_0^h
\qquad\text{in }L^2.
\]
This completes the proof.
\end{proof}}

\color{royalblue}\section{Discussion of dividends and anti-dilution adjustments }\label{app:antidilution}

\textcolor{royalblue}{The baseline model abstracts from dividends,
rights issues, share placements, and the associated anti-dilution
adjustments in order to focus on the interaction between the
path-dependent reset and call provisions. This appendix briefly
describes how such corporate actions can be incorporated.}

\textcolor{royalblue}{If the underlying stock pays a continuous
dividend yield $q_t$, its risk-neutral dynamics can be written as
\[
dS_t
=
(r_t-q_t)S_t\,dt
+
\sigma(t,\omega_t,H_t)S_t\,dW_t.
\]
Under this dynamics, the stock-price drift component of the local generator
is modified from
$
r_tS_t\partial_{\omega}V$
to
$(r_t-q_t)S_t\partial_{\omega}V.
$
Writing the resulting generator as $\mathcal A^q$, the PPDE in the
continuation region becomes
\[
\partial_tV
+\mathcal A^qV
-r_tV
+BC
=0.
\]
The baseline formulation in \eqref{eq:ppde} is recovered by setting $q_t=0$. Anti-dilution provisions protect holders of convertible securities against changes in conversion value arising from specified corporate actions \citep{woronoff2005antidilution}. Let $\tau$ be a corporate-action date and let $\mathcal E_\tau$ collect the relevant terms of the event. The adjusted conversion price can be represented generally as
\[
H_\tau
=
\Gamma_{\mathrm{ad}}
\bigl(H_{\tau^-};\mathcal E_\tau\bigr),
\]
where $\Gamma_{\mathrm{ad}}$ is determined by the anti-dilution provisions of the bond contract. The corresponding value-matching condition is
\[
V(\tau^-,\omega_{\tau^-},H_{\tau^-})
=
V\!\left(
\tau,\omega_\tau,
\Gamma_{\mathrm{ad}}
\bigl(H_{\tau^-};\mathcal E_\tau\bigr)
\right).
\]
This adjustment is imposed at the corporate-action date and does not enter the local differential operator between corporate-action dates.
}

\textcolor{royalblue}{For the CITIC Bank convertible bond (code:
113021), the general mapping above is specified contractually as
follows \citep{citicbank2019prospectus}. Let $P_0$ and $P_1$ denote
the conversion prices immediately before and after the adjustment,
respectively. Let $n$ be the stock-dividend or capitalisation rate, $k$ the new-share issuance or rights-issue ratio, $A$ the issue or subscription price, and $D$ the cash dividend per share. The contractual adjustment rules are summarised in Table~\ref{tab:citic_antidilution}.}
\textcolor{royalblue}{
\begin{table}[H]
\centering
\captionsetup{labelfont={color=royalblue},textfont={color=royalblue}}
\caption{Conversion-price adjustments for the CITIC Bank convertible bond (code: 113021)}
\label{tab:citic_antidilution}
\color{royalblue} %
\begin{tabularx}{\textwidth}{@{}p{0.18\textwidth}X>{\centering\arraybackslash}p{0.27\textwidth}@{}}
\toprule
Adjustment case & Corporate action & Adjusted conversion price \\
\midrule
Single event
& Stock dividend or capitalisation
& $P_1=\dfrac{P_0}{1+n}$ \\[6pt]
& New-share issuance or rights issue
& $P_1=\dfrac{P_0+Ak}{1+k}$ \\[6pt]
& Cash dividend
& $P_1=P_0-D$ \\
\addlinespace[4pt]
\midrule
\addlinespace[4pt]
Combined events
& Stock dividend or capitalisation together with new-share issuance or rights issue
& $P_1=\dfrac{P_0+Ak}{1+n+k}$ \\[6pt]
& All three events occurring together
& $P_1=\dfrac{P_0-D+Ak}{1+n+k}$ \\
\bottomrule
\end{tabularx}
\end{table}
In the notation of this paper, $P_0=H_{\tau^-}$ and $P_1=H_\tau$,
so the relevant formula gives the explicit form of
$\Gamma_{\mathrm{ad}}(H_{\tau^-};\mathcal E_\tau)$. These adjustments
exclude increases in share capital resulting from conversion of the
bond itself.}
\textcolor{royalblue}{The anti-dilution mapping
$\Gamma_{\mathrm{ad}}$ is distinct from the downward reset mapping
$\Psi$ considered in the main analysis. The reset mapping is
activated by a rolling-window stock price condition, whereas an anti-dilution adjustment is generated by a corporate action and is determined by the corresponding contractual formula. If both adjustments occur on the same date, their order must be specified by the bond contract.}
\textcolor{royalblue}{Dividends may also affect the bondholder's conversion decision because conversion changes the holder's entitlement from bond coupons to stock dividends. A joint treatment of dividend payments, unrestricted early conversion, and contractual anti-dilution adjustments would be a natural extension of the present framework.}

{\color{royalblue}
\section{Greeks across evaluation times}\label{app:greeks_time}

Table~\ref{tab:greeks_time} reports the automatic-differentiation Delta and Gamma estimates under the GBM, CEV, and Heston specifications at three evaluation times around the midpoint of the bond's maturity. This choice avoids the initial and near-maturity regions, where the estimated second-order sensitivities are more strongly affected by sparse-state coverage and terminal-boundary effects.
The selected stock prices represent regions below, at, and above the conversion price $H_t=7.45$.

\begin{table}[H]
\color{royalblue}
\captionsetup{labelfont={color=royalblue},textfont={color=royalblue}}
\centering
\small
\caption{Delta and Gamma estimates across models and evaluation times}
\label{tab:greeks_time}
\begin{tabular*}{0.96\textwidth}{@{\extracolsep{\fill}}lcrrrrrr}
\toprule
& & \multicolumn{3}{c}{Delta} & \multicolumn{3}{c}{Gamma} \\
\cmidrule(lr){3-5}\cmidrule(lr){6-8}
Model & $t$ & $S_t=6.00$ & $S_t=7.45$ & $S_t=9.00$
    & $S_t=6.00$ & $S_t=7.45$ & $S_t=9.00$ \\
\midrule
GBM & 2 & 0.4963 & 4.2313 & 6.8582 & 0.6779 & 2.8952 & 0.4996 \\
    & 3 & 3.5950 & 8.7805 & 8.0942 & 2.3316 & 0.8801 & $-1.4993$ \\
    & 4 & 6.5223 & 10.4686 & 8.8957 & 4.1441 & 1.2353 & $-3.2040$ \\
\addlinespace
CEV & 2 & 0.3847 & 1.9296 & 5.3058 & 0.2936 & 1.5211 & 1.4309 \\
    & 3 & 2.0468 & 6.3439 & 6.7885 & 1.2398 & 2.0219 & $-1.0858$ \\
    & 4 & 3.0471 & 9.0630 & 7.5504 & 2.6026 & 1.5068 & $-2.7225$ \\
\addlinespace
Heston & 2 & 0.5024 & 2.8534 & 5.5170 & 0.4962 & 1.9577 & 1.3667 \\
       & 3 & 1.5651 & 2.5113 & 5.6437 & 0.3916 & 1.5137 & 0.3970 \\
       & 4 & 4.2880 & 7.8627& 8.1144 & 2.3491& 1.3036 & $-1.7125$ \\
\bottomrule
\end{tabular*}

\vspace{0.1cm}
\begin{minipage}{0.95\textwidth}
\footnotesize
\vspace{0.1cm}
\textit{Notes:} The table reports automatic-differentiation estimates under the full-contract specification. Values at the three selected stock prices are obtained by linear interpolation on the stock-price grid. The $t=3$ estimates are taken from the baseline Greek experiment
reported in Section~\ref{sec:ad_greeks}, the $t=2$ and $t=4$ estimates use the same model, contract, simulation, and network settings.
\end{minipage}
\end{table}

}
{\color{royalblue}

\section{Regularisation and out-of-sample stability}
\label{app:regularisation}

This appendix examines whether the baseline neural network estimates are sensitive to explicit regularisation and to the number of simulated training paths. The exercise is motivated by the role of regularised least-squares
methods in controlling the fit of noisy numerical approximations \citep{shirzadi2021generalized}. The main analysis retains the unregularised
specification, corresponding to an $L^2$ penalty of $\lambda=0$. Here, we compare it with
$\lambda\in\{10^{-6},10^{-5},10^{-4},10^{-3}\}$.

The experiment uses the GBM full-contract specification. Each penalty value is evaluated using 12,000 training paths, 4,000 out-of-sample paths, and three random seeds. To eliminate sampling noise across regularisation level, identical simulated paths are reused across all penalty parameters within each seed. At each backward step, 10\% of the training paths are reserved for validation. The model trains for a maximum of 20 epochs, with a patience of three epochs for early stopping and a reduction in the learning rate when the validation loss reaches a plateau. The validation criterion is the mean squared error of standardised continuation-value target.

\begin{table}[H]
\color{royalblue}
\captionsetup{labelfont={color=royalblue},textfont={color=royalblue}}
\centering
\small
\caption{Regularisation path under the GBM full-contract specification}
\label{tab:regularisation_path}
\begin{tabular*}{\textwidth}{@{\extracolsep{\fill}}cccccc}
\toprule
$L^2$ penalty $\lambda$
& Mean price & Price SD
& Validation MSE & MSE SD
& Absolute price change (\%) \\
\midrule
$0$       & 130.96 & 4.11 & 0.040 & 0.0055 & 0.00 \\
\hspace{0.5em}$10^{-6}$ & 131.16 & 4.76 & 0.041 & 0.0060 & 0.15 \\
\hspace{0.5em}$10^{-5}$ & 131.05 & 4.74 & 0.040 & 0.0052 & 0.07 \\
\hspace{0.5em}$10^{-4}$ & 131.30 & 4.32 & 0.039 & 0.0055 & 0.26 \\
\hspace{0.5em}$10^{-3}$ & 130.55 & 6.05 & 0.040 & 0.0043 & 0.32 \\
\bottomrule
\end{tabular*}

\vspace{0.1cm}
\begin{minipage}{0.95\textwidth}
\footnotesize
\vspace{0.1cm}
\textit{Notes:} Each row is based on three random seeds, 12,000 training
paths, and 4,000 out-of-sample paths. Validation MSE is averaged across the
backward steps and seeds. The final column reports the absolute percentage
difference between the mean price at the stated penalty and that obtained with
$\lambda=0$.
\end{minipage}
\end{table}

Table~\ref{tab:regularisation_path} shows that the mean validation MSE is approximately constant, and estimated prices differ by less than 0.32\%. The dispersion across runs is not uniformly reduced by the penalty and is largest at $\lambda=10^{-3}$. This experiment therefore suggests that the valuation is
insensitive to moderate $L^2$ regularisation, which vindicates the choice of unregularised $\lambda=0$ in the baseline specification.

We further assess sensitivity to the training-sample size by repeating the experiment with 3,000 and 6,000 training paths. For these smaller samples, we compare the baseline with $\lambda=10^{-4}$, the penalty producing the lowest mean validation MSE in the 12,000-path experiment.

\begin{table}[H]
\color{royalblue}
\captionsetup{labelfont={color=royalblue},textfont={color=royalblue}}
\centering
\small
\caption{Regularisation results across training-sample sizes}
\label{tab:regularisation_samples}
\begin{tabular*}{\textwidth}{@{\extracolsep{\fill}}cccccc}
\toprule
Training paths & $\lambda$ & Mean price & Price SD
& Validation MSE & Absolute price change (\%) \\
\midrule
\hspace{0.5em}3,000  & $0$       & 128.89 & 8.32 & 0.1780 & 0.00 \\
\hspace{0.5em}3,000  & \hspace{0.5em}$10^{-4}$ & 127.42 & 6.10 & 0.1820 & 1.14 \\
\hspace{0.5em}6,000  & $0$       & 129.93 & 1.81 & 0.1805 & 0.00 \\
\hspace{0.5em}6,000  & \hspace{0.5em}$10^{-4}$ & 129.48 & 2.48 & 0.1760 & 0.35 \\
12,000 & $0$       & 130.96 & 4.11 & 0.0401 & 0.00 \\
12,000 & \hspace{0.5em}$10^{-4}$ & 131.30 & 4.32 & 0.0394 & 0.26 \\
\bottomrule
\end{tabular*}

\vspace{0.1cm}
\begin{minipage}{0.95\textwidth}
\footnotesize
\vspace{0.1cm}
\textit{Notes:} Each row is based on three random seeds and 4,000
out-of-sample paths. For each training-sample size, the final column is measured
relative to the corresponding unregularised estimate.
\end{minipage}
\end{table}

The unregularised mean price increases from 128.89 with 3,000 training paths to 130.96 with 12,000 paths. At the two larger sample sizes, imposing $\lambda=10^{-4}$ changes the mean price by less than 0.35\%; the larger 1.14\% difference at 3,000 paths is consistent with greater sampling variability in the smaller training set. Because the reported standard deviations are based on only three seeds and reflect both simulated-path and network-initialisation variation, they should be interpreted as descriptive rather than as evidence of a monotonic variance effect. The regularisation-path and sample-size results support the numerical stability of the baseline valuation while leaving the main specification unchanged.

}

\clearpage

\begin{landscape}
{\color{royalblue}
\section{Assessment of alternative neural network architectures}
\label{app:architecture_assessment}

Table~\ref{tab:architecture_assessment} reports the initial price estimates, variability across random seeds, and computational performance of alternative network depths, widths, and activation functions under the GBM, CEV, and Heston models.

\begin{table}[H]
\color{royalblue}
\captionsetup{labelfont={color=royalblue},textfont={color=royalblue}}
\centering
\normalsize
\setlength{\tabcolsep}{6.8pt}
\caption{Assessment of alternative neural network architectures across underlying-dynamics specifications}
\label{tab:architecture_assessment}
\begin{tabular}{llrrrrrrr}
\toprule
Model & Measure
& $3\times32$, ReLU
& $2\times64$, ReLU
& \textbf{$3\times64$, ReLU}
& $3\times64$, Tanh
& $3\times64$, ELU
& $4\times64$, ReLU
& $3\times128$, ReLU \\
\midrule
\multirow{5}{*}{GBM}
& Price & 120.77 & 128.42 & \textbf{128.42} & 128.22 & 130.23 & 128.76 & 129.41 \\
& SD & 1.11 & 0.77 & \textbf{0.55} & 0.41 & 0.42 & 0.57 & 0.67 \\
& $\Delta$ (\%) & 5.95 & 0.00 & \textbf{0.00} & 0.15 & 1.41 & 0.27 & 0.77 \\
& Time (s) & 369.5 & 373.5 & \textbf{392.8} & 444.2 & 548.9 & 409.5 & 460.1 \\
\midrule
\multirow{5}{*}{CEV}
& Price & 120.97 & 128.16 & \textbf{128.81} & 127.52 & 129.64 & 128.62 & 129.27 \\
& SD & 1.06 & 0.69 & \textbf{0.87} & 0.63 & 0.83 & 0.61 & 0.15 \\
& $\Delta$ (\%) & 6.08 & 0.50 & \textbf{0.00} & 1.00 & 0.65 & 0.14 & 0.36 \\
& Time (s) & 361.3 & 357.9 & \textbf{464.5} & 419.3 & 427.2 & 495.3 & 541.5 \\
\midrule
\multirow{5}{*}{Heston}
& Price & 127.59 & 128.09 & \textbf{128.69} & 125.61 & 129.15 & 129.26 & 129.20 \\
& SD & 0.44 & 0.26 & \textbf{0.35} & 0.96 & 0.67 & 0.37 & 0.73 \\
& $\Delta$ (\%) & 0.85 & 0.46 & \textbf{0.00} & 2.39 & 0.36 & 0.45 & 0.40 \\
& Time (s) & 727.1 & 739.7 & \textbf{951.3} & 839.3 & 862.4 & 1000.8 & 1064.4 \\
\bottomrule
\end{tabular}%
\vspace{0.1cm}

\begin{minipage}{0.95\linewidth}
\footnotesize
\vspace{0.1cm}
\textit{Notes:}
In the column headings, each architecture is described by its number of hidden layers and its number of neurons per hidden layer. The baseline architecture is shown in bold. Each alternative changes one architectural component while holding the remaining simulation and training settings fixed. All specifications use 104 time steps per year. Price reports the mean estimated initial bond price and SD is its sample standard deviation across three paired simulation and network seeds. Within each seed, common simulated paths and contractual states are used across architectures. The reported $\Delta$ is the absolute percentage deviation from the baseline mean price for the corresponding model.
\end{minipage}
\end{table}
}

\end{landscape}

\clearpage
{\color{royalblue}\section{Notation and definitions}}
\label{app:notation}

\color{royalblue}Table~\ref{tab:notation} summarizes the primary symbols used in the paper. The baseline numerical values of the contract and model parameters are reported in Tables~\ref{tab:cb_params} and~\ref{tab:model_parameters}.
Auxiliary symbols that occur only within individual proofs are defined at their first appearance. For symbols reused across distinct contexts, contextual differences are explicitly clarified to avoid ambiguity..

{\color{royalblue}\small
\renewcommand{\arraystretch}{1.15}
\captionsetup{
  labelfont={color=royalblue},
  textfont={color=royalblue}
}
\begin{longtable}{@{}p{0.25\textwidth}p{0.69\textwidth}@{}}
\caption{Summary of notation and definitions}\label{tab:notation}\\
\toprule
Symbol & Definition \\
\midrule
\endfirsthead

\multicolumn{2}{c}{\tablename~\thetable\ (continued)}\\
\toprule
Symbol & Definition \\
\midrule
\endhead

\midrule
\multicolumn{2}{r}{Continued on next page}\\
\endfoot

\bottomrule
\endlastfoot

\multicolumn{2}{@{}l}{\textit{Contract, path, and PPDE notation}}\\
\addlinespace[2pt]
$t$; $T$ & Current time; contractual maturity.\\
$t_i$; $N$ & The $i$-th monitoring or numerical-grid date; total number of grid intervals.\\
$h$ & Mesh size of the relevant uniform grid. In Theorem~\ref{thm:ppde_cb_viscosity}, it is the spacing between contractual monitoring dates; in Theorem~\ref{thm:conver}, it is the numerical time-step size.\\
$S_t$; $S_i$; $S_0$ & Underlying stock price at time $t$; stock price at $t_i$; initial stock price.\\
$H_t$; $H_i$; $H_T$ & Prevailing conversion price at time $t$; conversion price at $t_i$; conversion price prevailing at maturity.\\
$B$ & Face value of the convertible bond.\\
$C$ & Annual coupon rate; the continuous coupon flow in the PPDE is $BC$.\\
$K_t$; $K_i$ & Contractual redemption price when the issuer call is exercised at time $t$ or $t_i$.\\
$a$; $b$ & Downward-reset threshold ratio; issuer-call threshold ratio.\\
$N_t^{\mathrm{reset}}$, $N_i^{\mathrm{reset}}$ & Number of reset-qualifying observations in the relevant rolling window.\\
$N_t^{\mathrm{call}}$, $N_i^{\mathrm{call}}$ & Number of call-qualifying observations in the relevant rolling window.\\
$N_r$; $N_c$ & Minimum numbers of qualifying observations required to activate reset and call, respectively.\\
$\omega_t$ & Stopped stock-price path $\{S_u:0\leq u\leq t\}$.\\
$\|\omega\|_t$ & Supremum norm of a stopped path, $\sup_{0\leq s\leq t}|\omega_s|$.\\
$\mathcal P_i$ & Ordered rolling-window information or path summary needed to evaluate and update the contractual triggers.\\
$X_i$ & State at $t_i$: $(S_i,H_i,\mathcal P_i)$ under GBM and CEV, and $(S_i,v_i,H_i,\mathcal P_i)$ under Heston.\\
$x$ & A generic realisation of the state $X_i$.\\
$V(t,\omega_t,H_t)$ & Path-dependent convertible-bond value functional. Under a Markovian representation it is written as $V(t,S_t,H_t)$ or $V(t,S_t,v_t,H_t)$.\\
$V^{\mathrm{cont}}$ & Continuation value obtained from the discounted conditional expected future value plus the coupon flow.\\
$g(X_N)$ & Terminal payoff function, $\max\{B,(B/H_T)S_T\}$.\\
$\mathcal C$ & Continuation region in which neither the call nor the reset condition is active.\\
$\Psi(\omega_t,H_t)$ & Contractual downward-reset mapping for the conversion price.\\
$\mathcal A$ & Pathwise infinitesimal generator of the selected stock-price dynamics (GBM, CEV).\\
$\mathcal A^H$ & Pathwise infinitesimal generator associated with the Heston dynamics.\\
$\mu$; $\nu$ & Path-dependent drift and diffusion coefficient functionals used in the interval viscosity analysis.\\
$L$ & Common constant in the pathwise Lipschitz and linear-growth conditions.\\
$\chi$ & Nonnegative Lyapunov functional used in the interval comparison argument.\\
$\bar r$; $p$ & Uniform upper bound for the risk-free rate; exponent in the conversion-ratio moment condition.\\
$C^{1,2}$ & Class of non-anticipative test functionals that are once horizontally and twice vertically differentiable.\\
$\partial_t$; $\partial_\omega$; $\partial^2_{\omega\omega}$ & Horizontal derivative; first- and second-order vertical derivatives in functional It\^o calculus.\\
$\mathbb Q$; $\mathbb E^{\mathbb Q}$ & Risk-neutral probability measure; expectation under that measure.\\
$r_t$; $r_i$; $r$ & Risk-free rate at time $t$; rate at $t_i$; constant rate used in the numerical experiments.\\
$W_t$ & Standard Brownian motion driving the one-factor stock-price dynamics.\\
$\Delta$ & Small continuation-time increment in the definition of $V^{\mathrm{cont}}$; distinct from the Delta Greek $\Delta_t$.\\

\addlinespace[5pt]
\multicolumn{2}{@{}l}{\textit{Model-specific notation}}\\
\addlinespace[2pt]
$\sigma(t,\omega_t,H_t)$ & General path-dependent volatility function.\\
$\sigma$ & Constant volatility parameter under GBM and scale parameter under CEV.\\
$\gamma$ & Elasticity coefficient in the CEV diffusion term $\sigma S_t^\gamma$.\\
$v_t$; $v_0$ & Instantaneous variance under Heston; its initial value.\\
$\kappa$ & Mean-reversion rate of the Heston variance process.\\
$\theta$ & Long-run variance level in the Heston model; distinct from the neural-network parameter vector $\theta_i$.\\
$\eta$ & Volatility-of-variance parameter in the Heston model.\\
$\rho$ & Correlation between the Brownian motions driving the Heston stock and variance processes.\\
$W_t^S$; $W_t^v$ & Brownian motions driving the stock-price and variance processes under Heston.\\
$\boldsymbol\omega_t=(\omega_t^S,\omega_t^v)$ & Augmented stopped path containing the stock-price and variance histories under Heston.\\

\addlinespace[5pt]
\multicolumn{2}{@{}l}{\textit{Monitoring-date transmission and dynamic programming}}\\
\addlinespace[2pt]
$C_i^{\mathrm{trig}}(x)$; $R_i^{\mathrm{trig}}(x)$ & Indicator of call activation; indicator of reset activation at $t_i$.\\
$\Gamma_i x$ & State update that appends the current stock price to the rolling window, removes the oldest observation, and leaves the conversion price unchanged.\\
$\Gamma_i^R x$ & Reset-state update that also replaces the conversion price by the current stock price.\\
$\mathcal J_i$ & Call-first monitoring-date transmission operator linking the post-monitoring and pre-monitoring values.\\
$t_i^-$; $t_i^+$ & Values immediately before and immediately after the contractual decision at monitoring date $t_i$.\\
$V^*$; $V_*$ & Upper- and lower-semicontinuous envelopes of the value functional.\\
$\mathcal T_i$ & Exact one-step backward dynamic-programming operator on the fixed numerical grid.\\
$\phi$, $\psi$ & Generic measurable next-step value functions or viscosity test functions, according to context.\\
$V_i^h$ & Exact value generated by the fixed-grid recursion $V_i^h=\mathcal T_iV_{i+1}^h$.\\
$\mathbb P_i$ & Probability law of the state $X_i$.\\
$L^2(\mathbb P_i)$ & Space of square-integrable functions with respect to $\mathbb P_i$.\\
$\varepsilon_i$ & Population $L^2(\mathbb P_i)$ bound for the one-step neural-network value residual.\\
$D_{t,s}$ & Discount factor $\exp(-\int_t^s r_u\,du)$ used in the viscosity proof.\\
$A_{t,s}$ & Accumulated discounted coupon flow $\int_t^sD_{t,u}BC\,du$.\\
$G_i$ & Terminal value supplied to the interval ending at monitoring date $t_i$ in the backward existence construction.\\

\addlinespace[5pt]
\multicolumn{2}{@{}l}{\textit{Neural-network approximation and numerical assessment}}\\
\addlinespace[2pt]
$\widehat V_i(\cdot;\theta_i)$ & Neural-network approximation of the conditional continuation value at $t_i$.\\
$\widehat{\mathrm{Cont}}_i$; $\mathrm{Cont}_i^{(m)}$ & Fitted continuation value; simulated one-step continuation target for training path $m$.\\
$\mathrm{Conv}_i^{(m)}$ & Conversion value $(B/H_i^{(m)})S_i^{(m)}$ on training path $m$.\\
$\Phi$; $x_i=\Phi(X_i)$ & Feature construction and standardisation map; resulting finite-dimensional network input.\\
$\theta_i$ & Collection of trainable neural-network parameters at time step $i$.\\
$W_\ell$; $b_\ell$ & Weight matrix and bias vector of neural-network layer $\ell$.\\
$h^{(\ell)}$ & Hidden-layer representation at layer $\ell$; unrelated to the time-step size $h$.\\
$\mathcal L_i(\theta_i)$ & Mean squared training loss at time step $i$.\\
$M$; $M_{\mathrm{train}}$; $M_{\mathrm{test}}$ & Number of Monte Carlo paths; numbers of training and independent test paths.\\
$\widetilde X_i^{(m)}$; $\widetilde V_i^{(m)}$ & State and fitted value for test path $m$.\\
$\lambda$ & $L^2$ weight-penalty parameter used in the regularisation analysis.\\
$D_{\mathrm{abs}}$; $D_{\mathrm{rel}}$ & Absolute and relative differences between the DNN and LSMC price estimates.\\
$\Delta_t$; $\Gamma_t$ & Delta $\partial\widehat V_t/\partial S_t$; Gamma $\partial^2\widehat V_t/\partial S_t^2$.\\
$\varepsilon_S$ & Stock-price perturbation used for central finite-difference validation of Delta and Gamma.\\

\addlinespace[5pt]
\multicolumn{2}{@{}l}{\textit{Dividend and anti-dilution extension}}\\
\addlinespace[2pt]
$q_t$ & Continuous dividend yield of the underlying stock.\\
$\tau$; $\mathcal E_\tau$ & Corporate-action date; collection of terms describing that event.\\
$\Gamma_{\mathrm{ad}}$ & Contractual anti-dilution mapping for the conversion price; distinct from the state-update maps $\Gamma_i$ and $\Gamma_i^R$.\\
$P_0$; $P_1$ & Conversion prices immediately before and after an anti-dilution adjustment.\\
$n$; $k$ & Stock-dividend or capitalisation rate; new-share issuance or rights-issue ratio.\\
$A$; $D$ & Issue or subscription price; cash dividend per share in Appendix~\ref{app:antidilution}.\\
\end{longtable}
}
\end{document}